\documentclass[11pt,letterpaper]{article}
\usepackage[utf8]{inputenc}
\usepackage[margin=1.0in]{geometry}
\usepackage{palatino}
\usepackage{caption}
\usepackage{macros}
\usepackage{algorithm}
\usepackage{algorithmic}
\usepackage{booktabs}
\usepackage{bm}
\usepackage{csquotes,fullpage}
\title{Efficient Mass Matrix Estimation with Gaussian Cooling}
\author{Jake Hofgard\thanks{UC Berkeley, \url{hofgard@berkeley.edu}}\and Michael Lindsey\thanks{UC Berkeley and Lawrence Berkeley National Laboratory, \url{lindsey@math.berkeley.edu}}}

\begin{document}
\maketitle 
\begin{abstract}
We introduce an algorithm for efficiently preconditioning log-concave and log-smooth distributions that scales logarithmically with the condition number of the underlying distribution. Based on Gaussian cooling, our multistage method approximately samples from a sequence of well-conditioned distributions to construct a preconditioner for each subsequent stage of the cooling schedule. This method is motivated by existing practical approaches for mass matrix estimation in Markov chain Monte Carlo (MCMC), in which an effective preconditioner for the target distribution is estimated from samples produced during a warm-up period. However, our approach is conceptually distinct from existing approaches and moreover permits theoretical analysis. Our algorithmic framework is flexible, allowing essentially any log-concave sampling algorithm to act as the subroutine within each cooling stage. For instance, using first-order rejection sampling (FORS) as the sampler, the first-order query complexity of our method is $\tilde{O}(d^{4/3} \log \kappa)$ for a large class of log-concave and log-smooth distributions with condition number $\kappa$ and $\tilde{O}(d^{1/3} \log \kappa)$ for translation-invariant distributions. In the translation-invariant setting, we provide numerical experiments for the lattice $\phi^4$ model, a simple lattice field theory.
\end{abstract}


\section{Introduction}
\label{sec:intro}
Given a potential $U : \R^d \to \R$, sampling from the probability distribution specified by the density $\pi(x) \propto \exp(-U(x))$ is a ubiquitous task in scientific computing, arising in generative modeling, molecular dynamics, and a variety of other applications across physics, chemistry, computer science, and computational mathematics. Most modern approaches to this sampling problem involve some form of \textit{Markov chain Monte Carlo} (MCMC), in which a Markov chain that admits (an approximation of) $\pi$ as an invariant distribution is constructed. Given query access to the potential $U$ and, if necessary, its gradient, MCMC practitioners then simulate a prescribed chain, ultimately obtaining a set of approximate samples from $\pi$. Popular examples of such methods include the Metropolis-adjusted Langevin algorithm (MALA), Hamiltonian Monte Carlo (HMC) and its adjusted variant, and ensemble MCMC samplers involving populations of interacting particles. From an algorithmic perspective, the mixing times of MCMC methods are typically bottlenecked by three factors: (i) dimension dependence, (ii) multimodality, and (iii) condition number dependence.

Significant recent progress has been made towards understanding the optimal dimension dependence of MALA \cite{chewi2021mala}, underdamped Langevin Monte Carlo (ULMC) \cite{altschuler2026shifted}, and HMC \cite{zhang2026algorithmic} for sampling from strongly log-concave distributions. Multimodal distributions, which do not fall under the umbrella of (strong) log-concavity, pose many challenges to MCMC methods. While many adaptive methods attempt to address the problem of sampling from multimodal distributions, we do not focus on this setting in our work.

Closely related to the question of optimal condition number dependence is the notion of \textit{affine invariance}, popularized by \cite{Goodman2010}. Loosely speaking, the mixing time of an affine invariant sampler does not depend on affine changes of coordinates, and such samplers should therefore be unfazed by problems with extremely large condition numbers. While the affine invariant ensemble samplers proposed in \cite{Goodman2010}, and more recently, in \cite{chen2025affine}, are distinct from the discretized diffusion processes (e.g., ULMC) discussed above, several algorithms exist that attempt to unify the two approaches. For example, Affine Invariant Langevin Dynamics (ALDI) is an ensemble sampler that encodes affine invariance into a coupled system of SDEs, each of which governs the evolution of a particle according to a preconditioned version of overdamped Langevin dynamics \cite{GarbunoInigo2020}. The entire system of SDEs is coupled by the ensemble's sample covariance, which appears as a preconditioner of the drift term in each individual particle's dynamics. 

The performance of all affine invariant ensemble samplers still depends on the initialization of the ensemble, and a poorly initialized ensemble can lead to degraded performance for ill-conditioned distributions. In particular, the affine invariance guarantees of ALDI and other affine invariant ensemble samplers do not truly apply until the ensemble has mixed, in the sense that the performance of such methods is still impacted by the relative conditioning of the target distribution and the ensemble's initialization. In the worst case, an initial ensemble whose empirical covariance is low-rank or poorly conditioned can severely degrade the performance of an affine invariant sampling algorithm \cite{GarbunoInigo2020}.

In the HMC literature, a linear preconditioner is typically called the \textit{mass matrix} \cite{betancourt2011geometry, betancourt2018conceptual}, and selecting a good mass matrix is of paramount importance in modern implementations of HMC such as Stan \cite{Carpenter2017}. In such software packages, one often estimates the mass matrix from samples during a warmup period. Although it may take a long time to sample along \textit{all} eigendirections specified by the covariance of the target distribution, the mass matrix estimation approach implemented by Stan and other MCMC packages can still work quite well in practice. At least heuristically, the success of this method appears to be explained by its ability to mix along different eigendirections at different timescales. The directions corresponding to the covariance's smallest eigenvalues are explored first (with a very small step size), and during the warmup period, the mass matrix estimate is then updated to accurately reflect the true covariance along larger and larger length scales. In practice, this estimation takes place over a handful of warm-up stages in which the mass matrix is updated several times \cite{Carpenter2017}. Although we are inspired by this multistage approach and view it as a strong practical benchmark, to the best of our knowledge there is no theoretical analysis available that captures its advantages.

In this work, we aim to provide a strong theoretical foundation for mass matrix estimation and a route towards logarithmic condition number dependence for log-concave sampling, at least for distributions that are amenable to linear preconditioning. With an eye towards sampling efficiently from distributions from Bayesian posteriors to lattice field theories, we develop: (i) a natural condition (Equation~\ref{eq:true-conditioning}) that encodes the idea that, in some (a priori unknown) coordinate system, the target distribution is well-conditioned, (ii) an implementable algorithm (Algorithm~\ref{alg:gaussian-cooling}) with logarithmic condition number dependence for mass matrix estimation and sampling in the log-concave setting, and (iii) a detailed theoretical framework (Theorem~\ref{th:gc-query-complexity-general}) for analyzing our algorithm.

\subsection{Contributions}
\label{subsec:contributions}
Suppose that $\pi$ is a distribution on $\R^d$ that admits a density $\pi(x) \propto \exp(-U(x))$ with respect to the Lebesgue measure, where $U : \R^d \to \R$ is a $C^2$ potential. We begin with the standard assumption that the target distribution is $\mu$-strongly log concave and $L$-smooth, in the sense that
\begin{align}
\label{eq:log-concave-smooth}
    \mu I \preceq \nabla^2 U(x) \preceq L I, \quad \forall x \in \R^d
\end{align}
for some $L > \mu > 0$. Even in this well-studied setting, the condition number $\kappa = L / \mu$ may be extremely large, resulting in a computational bottleneck. However, it may be the case that
\begin{align}
\label{eq:true-conditioning}
    \alpha B \preceq \nabla^2 U(x) \preceq B, \quad \forall x \in \R^d
\end{align}
for some positive definite matrix $B \in \R^{d \times d}$ and a parameter $\alpha \in (0, 1)$. Note that any distribution that satisfies Equation~\ref{eq:log-concave-smooth} also satisfies Equation~\ref{eq:true-conditioning} with $B = L I$ and $\alpha = 1 / \kappa$. However, as we shall see in several examples in Section~\ref{sec:num-exps}, it is possible for $\kappa$ to be extremely large but for Equation~\ref{eq:true-conditioning} to hold with $\alpha = \Theta(1)$ for some matrix $B \succ 0$. This discrepancy can arise naturally when a well-conditioned distribution is transformed by an ill-conditioned affine mapping, a setting that affine invariant samplers are designed to address.

After transforming to the coordinate system specified by the matrix $B$ in Equation~\ref{eq:true-conditioning}, to which we do \textit{not} assume access in our algorithms, the issue of a large condition number may be partially or fully alleviated. In particular, we provide an algorithm that, without assuming the structure of the matrix $B$ or the value of the parameter $\alpha$, allows for efficient sampling from distributions whose potential satisfy Equation~\ref{eq:true-conditioning}. On top of assuming that Equation~\ref{eq:log-concave-smooth} and Equation~\ref{eq:true-conditioning} hold for the target distribution $\pi$, we also assume access to $x^* := \argmin_{x \in \R^d} U(x),$ the global minimizer of the target potential. Note that this minimizer is unique due to our strong convexity assumption.

We then define a sequence of perturbed potentials by
\begin{align}
\label{eq:U-k}
    U_k(x) := U(x) + \frac{\gamma^k L}{2} \|x - x^*\|_2^2, \quad k \in \N,
\end{align}
where $\gamma \in (0, 1)$ is a parameter chosen by the user. By construction, $x^*$ is also the unique global minimizer of $U_k$ for each $k \in \N$. At each stage of Gaussian cooling, we aim to sample from the distribution $\pi_k(x) \propto \exp(-U_k(x))$, whose potential approaches that of the target distribution as cooling progresses. Although $U_k$ may still specify a distribution whose condition number is unacceptably large, our scheme uses the sample covariance from each stage of cooling to precondition the next stage. We ensure that, at each stage, the resulting preconditioned distribution $\widetilde{\pi}_k$ is well-conditioned. Once $k$ is sufficiently large, the sample covariance from stage $k$ of Gaussian cooling provides an effective preconditioner for the target potential itself.

Our method does \textit{not} depend on the choice of underlying sampling algorithm. In particular, at each stage, the sampler that produces the samples used to estimate the sample covariance can be drawn from any one of the options described in Appendix~\ref{app:log-concave-complexity}. In Section~\ref{sec:query-complexity}, we first present our theoretical guarantees in a sampler-agnostic manner. Later, we use two explicit choices of sampling algorithm: first-order rejection sampling (FORS) \cite{chen2026exact} yields the best possible query complexity guarantees at the time of writing, while underdamped Langevin Monte Carlo (ULMC) \cite{altschuler-chewi-ulmc} is straightforward to implement and hence our choice of sampler in the experiments in Section~\ref{sec:num-exps}. With this in mind, we may present an informal version of our query complexity result based on the first-order query complexity afforded by FORS.
\begin{theorem}[Informal version of Corollary~\ref{cor:fors-query-complexity}]
\label{th:informal-query-complexity}
Suppose that $U : \R^d \to \R$ satisfies Equation~\ref{eq:log-concave-smooth} and Equation~\ref{eq:true-conditioning}. Then, with high probability, Algorithm~\ref{alg:gaussian-cooling} yields a preconditioner $\widetilde{\Sigma}$ for the Hessian $\nabla^2 U(x)$ of the target potential that satisfies
\begin{align}
\label{eq:informal-precond-guarantee}
    \frac{\alpha}{4} I \preceq \widetilde{\Sigma}^{1/2} \nabla^2 U(x) \widetilde{\Sigma}^{1/2} \preceq \frac{2}{\alpha}I, \quad \forall x \in \R^d
\end{align}
with $\tilde{\Theta}(d^{4/3} \log(\kappa))$ queries to $U$ and its gradient. In particular, the distribution specified by the potential $\widetilde{U}(x) = U(\widetilde{\Sigma}^{1/2}x)$ has condition number $O(1 / \alpha^2)$.
\end{theorem}

We often anticipate $\alpha = \Theta(1)$ for applications of interest, meaning that Algorithm~\ref{alg:gaussian-cooling} provides an efficient pathway to sampling from a well-conditioned distribution. Moreover, when the target distribution has translation-invariant covariance, the dimension dependence of the query complexity presented in Theorem~\ref{th:informal-query-complexity} can be improved even further. In Section~\ref{sec:num-exps} we consider sampling from a lattice $\phi^4$ field theory, a natural example where translation invariance holds. This application is motivated by recent numerical efforts in mass matrix estimation for more complicated lattice field theories \cite{Lunts2023, Patel2024, Patel2025}. In the translation-invariant setting, the query complexity in Theorem~\ref{th:informal-query-complexity} can be further reduced to $\widetilde{\Theta}(d^{1/3} \log(\kappa))$.

\subsection{Related Work}
\paragraph{Gaussian Cooling, Volume Estimation, and Log-Concave Sampling.} Our algorithm is partially inspired by Gaussian cooling \cite{gaussian-cooling}, an algorithm for estimating the volume of convex bodies. The algorithm samples from a series of distributions, beginning with a sharply concentrated Gaussian that ``cools'' according to a schedule that gradually increases the variance of the Gaussian until the resulting samples are approximately uniform from the target convex body. Although our algorithm is conceptually similar, we borrow little from the analysis and specific algorithm from \cite{gaussian-cooling}. Other works that use the notion of Gaussian cooling to analyze the complexity of volume estimation and log-concave sampling include \cite{Lee2018} and \cite{kook24gaussian}. In a similar vein, our method builds upon a large body of recent work aimed at characterizing the query complexity of (strongly) log-concave sampling, both with respect to the underlying dimension $d$ and the condition number $\kappa$ of the target distribution. We provide a brief, non-exhaustive summary of some of these results in Appendix~\ref{app:log-concave-complexity}, split into two categories: low and high accuracy samplers. For a more in-depth survey on log-concave sampling algorithms and their query complexity, see \cite{chewi-book}.

\paragraph{Affine Invariant Sampling.}
In contrast with the setting of log-concave sampling, the literature on affine invariant sampling places less emphasis on first-order query complexity as the guiding standard for algorithmic performance. In fact, for the ``stretch move'' and other affine invariant ensemble samplers introduced in \cite{Goodman2010}, the question of ergodicity remains open \cite{chen2025affine}. This is not the case for ALDI, for which both non-degeneracy of the empirical covariance used to evolve the ensemble and ergodicity have been established in \cite{GarbunoInigo2020} and for various adjusted versions of ALDI as in \cite{Sprungk2025}. A number of techniques also attempt to circumvent condition number dependence by instead discretizing a diffusion process on an appropriate Riemannian manifold \cite{Girolami2011}. Riemannian HMC, for example, yields a state-of-the-art algorithm for polytope volume computation \cite{Lee2018} and other constrained sampling problems \cite{kook23rhmc}. Similarly, the mirror Langevin algorithm presents an alternative route towards condition number-free sampling \cite{ahn2021efficient, gatmiry2022convergence}. Although Riemannian HMC and mirror Langevin come equipped with compelling theoretical guarantees in certain settings, we focus on samplers whose state space is equipped with the Euclidean metric.

\paragraph{Learning Linear Preconditioners.}
Several recent results address the question of efficiently learning a preconditioner \textit{prior} to running a high-accuracy sampler. For example, \cite{hird2025quantifying} establishes that there exist distributions with arbitrarily large condition number $\kappa$ such that no linear preconditioner can improve the conditioning of the sampling problem, akin to the example we provide in Appendix~\ref{app:cautionary-tale}. \cite{hird2025quantifying} then characterizes distributions for which preconditioning with a diagonal preconditioner, the Fisher information, or the sample covariance can reduce the condition number of the underlying problem. 
On the other hand, the authors show that preconditioning with a diagonal covariance can sometimes \textit{increase} the condition number of the problem, even for Gaussian targets. 
Similarly, \cite{hird2026nonasymptotic} quantify the complexity of estimating the sample covariance and Fisher information directly from MCMC samples.

This work is perhaps most conceptually similar to our method, and at a technical level, some of the tools used in \cite{hird2026nonasymptotic} are quite similar to the tools that we use in Appendix~\ref{app:tech-proofs}. However, our algorithmic approach is fundamentally different from that of \cite{hird2026nonasymptotic}, which focuses on learning preconditioners directly from samples during warm-up. Moreover, the results in \cite{hird2026nonasymptotic} require polynomial dependence on the condition number of the underlying distribution in order to estimate a good preconditioner. For the class of distributions that satisfy Equation~\ref{eq:true-conditioning} with $\alpha = \Theta(1)$, our method reduces this dependence to logarithmic rather than polynomial. Other recent work aimed at rigorously describing the complexity of preconditioner estimation includes \cite{kook2026covariance}, where the authors establish an $\tilde{O}\left(\frac{d(\Tr(\Sigma) + C_{\mathrm{PI}}}{\varepsilon^2 \Tr(\Sigma)}\right)$ sample complexity estimate for estimating the covariance $\Sigma$ from a sampler initialized at its stationary distribution, assuming that the Markov kernel of the sampler has Poincaré constant $C_{\mathrm{PI}} > 0$.

\paragraph{Sampling from Lattice Field Theories.}
Our work is partially motivated by the goal of efficiently sampling from lattice field theories, a difficult problem with consequences in condensed matter and high-energy physics. This direction has been explored in some detail in \cite{Lunts2023}, which uses HMC to sample from the $O(3)$ spin-fermion model. Although the action for this model and the exact application of HMC are quite complicated, the authors crucially leverage translation invariance to accelerate HMC, estimating the mass matrix in Fourier space. We draw upon this idea in our application to lattice $\phi^4$ field theory in Section~\ref{subsec:lattice-phi4}. A subset of the same group of authors continued this line of work in \cite{Patel2024} and \cite{Patel2025}, utilizing HMC in a similar manner to investigate other spin-boson and spin-fermion models. We emphasize that while our theoretical guarantees are limited to the setting of strong log-concavity, we hope that our algorithm itself may see practical use in more complicated settings such as the above.

\section{Gaussian Cooling Algorithm}
\label{sec:algorithm}
In this section, we describe our primary algorithmic contribution in detail. Although the structure of the Gaussian cooling algorithm is independent of the choice of sampler used at each stage, its query complexity is dependent on the choice of sampler for each stage of Gaussian cooling. Before presenting two concrete choices for the sampler, we present a general analysis that holds for a broad class of log-concave sampling algorithms, including both low-accuracy samplers (Table~\ref{tab:low}) and high-accuracy samplers (Table~\ref{tab:high}).

Given an evaluation and gradient oracle for the target potential $U : \R^d \to \R$, one can form an evaluation and gradient oracle for any potential of the form in Equation~\ref{eq:U-k}. In fact, if $M \succ 0$ is any positive definite matrix, then we can also form an evaluation and gradient oracle for the augmented potential $x \mapsto U_k(M^{1/2} x)$.
At each stage of Gaussian cooling, we will query such an oracle, where $M$ is the preconditioner formed from samples at the previous stage of cooling. Below, $\widehat{S}_n(x_1, \ldots, x_n)$ is the (biased) sample covariance formed from $n$ independent samples; see Appendix~\ref{app:tech-assumps}. Finally, we write $\mathcal{A}(V, x^*, N, n)$ to denote a call to a generic log-concave sampling algorithm $\mathcal{A}$ that takes as input a potential $V$, a global minimizer $x^*$, the number of steps $N$ to take before returning a sample, and the total number of samples $n$ to produce in parallel.

\begin{algorithm}
\caption{Preconditioning via Gaussian Cooling}
\label{alg:gaussian-cooling}
\begin{algorithmic}[1]
\REQUIRE Sampling algorithm $\mathcal{A}$, smoothness parameter $L > 0$, global minimizer $x^* \in \R^d$, potential $U : \R^d \to \R$ with evaluation and gradient oracle, number of Gaussian cooling stages $K$, number of samples per stage $n$, number of sampler steps per sample $N$.
\ENSURE A preconditioner $\widetilde{\Sigma}_n^{(K)}$ satisfying Equation~\ref{eq:final-precond-guarantee} with high probability.
\STATE $\widetilde{\Sigma}_{n}^{(-1)} \gets L^{-1} I$
\FOR{$k = 0, \ldots, K$}
    \STATE Form $\widetilde{U}_k(x) = U_k((\widetilde{\Sigma}_n^{(k-1)})^{1/2} x)$
    \STATE $x_1^{(k)}, \ldots, x_n^{(k)} \gets \mathcal{A}\left(\widetilde{U}_k, (\widetilde{\Sigma}_n^{(k-1)})^{-1/2} x^*, N, n\right)$
    \STATE $\widehat{\Sigma}_n^{(k)} \gets \widehat{S}_n(x_1^{(k)}, \ldots, x_n^{(k)})$
    \STATE $\widetilde{\Sigma}_n^{(k)} \gets (\widetilde{\Sigma}_n^{(k-1)})^{1/2} \widehat{\Sigma}_n^{(k)}(\widetilde{\Sigma}_n^{(k-1)})^{1/2}$
\ENDFOR
\RETURN $\widetilde{\Sigma}_n^{(K)}$
\end{algorithmic}
\end{algorithm}

As written above, Algorithm~\ref{alg:gaussian-cooling} returns a preconditioner $\widetilde{\Sigma}_n^{(K)}$ for the target distribution in a sense made precise below. If the desired use case is mass matrix tuning in HMC, for example, this output is all that is needed. Of course, if the user's end goal is a high accuracy sample from the target distribution, one additional sampling step is needed.

Recall that we successively sample from distributions $\pi(x) \propto \exp(-U_k(x))$, with potentials $U_k$ defined in Equation~\ref{eq:U-k}. Throughout, we let $\pi_k(x) \propto \exp(-U_k(x))$ be the corresponding distribution and take $\Sigma^{(k)} := \covar(\pi_k)$. We first reduce the analysis of Algorithm~\ref{alg:gaussian-cooling} to establishing an appropriate concentration guarantee for the preconditioner produced at each stage of cooling.

\begin{proposition}
\label{prop:gc-induction}
Suppose that, at each stage of Algorithm~\ref{alg:gaussian-cooling}, the sampling algorithm $\mathcal{A}$ outputs a preconditioner $\widetilde{\Sigma}_n^{(k-1)}$ satisfying
\begin{align}
\label{eq:sample-covar-requirement}
    (1 + \delta)^{-1} \Sigma^{(k-1)} \preceq \widetilde{\Sigma}_n^{(k-1)} \preceq (1 + \delta) \Sigma^{(k-1)},
\end{align}
for some $\delta > 0$. Then, after $K = O\left(\frac{\log(\kappa)}{\log(1 / \gamma)}\right)$ stages of Gaussian cooling, Algorithm~\ref{alg:gaussian-cooling} will yield a preconditioner that satisfies
\begin{align}
\label{eq:final-precond-guarantee}
\frac{\alpha}{(1 + \alpha)(1 + \delta)} I \preceq (\widetilde{\Sigma}_{n}^{(K)})^{1/2} \nabla^2 U(x) (\widetilde{\Sigma}_{n}^{(K)})^{1/2} \preceq \frac{1 + \delta}{\alpha} I.
\end{align}
\end{proposition}

Due to Proposition~\ref{prop:gc-induction}, whose proof is in Appendix~\ref{app:query-comp-proofs}, our task reduces to showing that $\mathcal{A}$ can, with high probability, ensure that Equation~\ref{eq:sample-covar-requirement} holds at every stage of Gaussian cooling.

\section{Query Complexity of Gaussian Cooling}
\label{sec:query-complexity}

We place the following assumption on the sampling algorithm $\mathcal{A}$. For specific choices of $\mathcal{A}$ such as ULMC or FORS, the parameters in the following assumption are made explicit.

\begin{assumption}
\label{asmp:sampling-algorithm}
We say $\mathcal{A}$ satisfies a low accuracy guarantee if, given a target distribution $\pi(x) \propto \exp(-U(x))$ satisfying Equation~\ref{eq:log-concave-smooth} and with global minimizer $x^*$, there exist constants $c_1, c_2, c_3 > 0$ such that for any $\varepsilon > 0$, the sampling algorithm $\mathcal{A} = \mathcal{A}(U, x^*, N, 1)$ can make at most
\begin{align}
\label{eq:low-accuracy-cond}
    N = \widetilde{O}\left(\kappa^{c_1} d^{c_2} \varepsilon^{-c_3}\right)
\end{align}
queries to $U$ and its gradient and output a sample with distribution $\widehat{\pi}_N$ that satisfies $ \sR_2(\widehat{\pi}_N \| \pi) \leq \varepsilon^2$,
where $\sR_2(\cdot \| \cdot)$ is the Rényi $2$-divergence defined in Appendix~\ref{app:tech-assumps}.
If we can instead take
\begin{align}
\label{eq:high-accuracy-cond}
    N = \widetilde{O}\left(\kappa^{c_1} d^{c_2} (\log(\varepsilon^{-1}))^{c_3} + \log(1 / \eta)\right),
\end{align}
to obtain the same guarantee in Rényi-$2$ divergence with probability at least $1 - \eta$, we say that $\mathcal{A}$ satisfies a high accuracy guarantee.
\end{assumption}

At each stage of Gaussian cooling, a sampling algorithm that satisfies Assumption~\ref{asmp:sampling-algorithm} maintains Equation~\ref{eq:sample-covar-requirement} across all stages of cooling, with high probability.

\begin{theorem}
\label{th:gc-query-complexity-general}
Fix $p \in (0, 1)$ and suppose that $\mathcal{A}$ satisfies Assumption~\ref{asmp:sampling-algorithm}. Take $K$ as in Proposition~\ref{prop:gc-induction} and
\begin{align*}
    n = \Theta\left(\frac{d + \log(K / p)}{\widetilde{\alpha}^4 \delta^2}\right)
\end{align*}
samples per stage, where $\widetilde{\alpha} := \min\{\alpha, \gamma\}$. Then, with probability at least $1 - p$, Algorithm~\ref{alg:gaussian-cooling} returns a preconditioner $\widetilde{\Sigma}_n^{(K)}$ satisfying Equation~\ref{eq:final-precond-guarantee}, and each of the $nK$ calls to $\mathcal{A}$ uses at most
\begin{align*}
    N = \begin{cases}
    \widetilde{\Theta}\left(\left(\widetilde{\alpha}\right)^{-2 c_1} d^{c_2} \delta^{-c_3}\right),
    & \mathcal{A} \text{ is a low accuracy sampler}, \\
    \widetilde{\Theta}\left(\left(\widetilde{\alpha}\right)^{-2 c_1}d^{c_2} (\log(\delta^{-1}))^{c_3}\right) + O\left(\log(nK/p)\right),
    & \mathcal{A} \text{ is a high accuracy sampler}
    \end{cases}
\end{align*}
queries to $U$ and its gradient oracle.
\end{theorem}

The proof of this theorem leverages the method developed in Appendix~\ref{app:tech-proofs} to establish that Equation~\ref{eq:sample-covar-requirement} holds at each stage of cooling with the given choice of parameters. Once this is established, Proposition~\ref{prop:gc-induction} yields the desired preconditioning guarantee, and a union bound ensures that the failure probability of Gaussian cooling across all $K$ stages remains controlled.

\subsection{Translation-Invariant Case}
\label{subsec:translation-invariant}
In many applications of interest, the target covariance will be \textit{translation-invariant}, in the sense that $\Sigma$ has block circulant structure. This is the case in the lattice $\phi^4$ model that we consider in Section~\ref{sec:num-exps} and for more complicated lattice field theories \cite{Lunts2023}. Leaving the exact details of the lattice $\phi^4$ model to the next section, we simply note that if $\Sigma \in \R^{d \times d}$ is block circulant with $\sqrt{d} \times \sqrt{d}$ circulant blocks, then it is exactly diagonalized by a discrete Fourier transform (DFT), which we denote by $\sF \in \C^{d \times d}$. In particular, there exists a vector $\sigma \in \R^{d}$ such that $\Sigma = \sF \diag(\sigma) \sF^*$, where $\sF^*$ denotes the adjoint of the discrete Fourier transform. Thus, in the translation-invariant case, our task is reduced from estimating the covariance matrix to estimating the vector $\sigma \in \R^{d}$. In this setting, we replace Line 5 of Algorithm~\ref{alg:gaussian-cooling} with
\begin{align}
\label{eq:translation-invariant-alg}
&\widehat{\sigma}^{(k)} = \frac{1}{n} \sum_{i=1}^n \sF(x_i^{(k)} - \widehat{m}_n(x_1^{(k)}, \ldots, x_n^{(k)})) \odot \overline{\sF(x_i^{(k)} - \widehat{m}_n(x_1^{(k)}, \ldots, x_n^{(k)}))}, \\
&\widehat{\Sigma}^{(k)}_n = \sF \diag(\widehat{\sigma}^{(k)}) \sF^*,
\end{align}
where $\widehat{m}_n(x_1, \ldots, x_n)$ is the sample mean.
With this change in place, Corollary~\ref{cor:translation-invariant-union-bound} then yields the following query complexity guarantee.
\begin{theorem}
\label{th:query-complexity-translation-invariant}
Taking $n$ as in Theorem~\ref{th:gc-query-complexity-general}, with $d + \log(K / p)$ replaced by $\log(d K / p)$, and $N$ and $K$ as above, Algorithm~\ref{alg:gaussian-cooling} returns a preconditioner $\widetilde{\Sigma}_n^{(K)}$ for the target distribution satisfying Equation~\ref{eq:final-precond-guarantee} with probability at least $1 - p$.
\end{theorem}
The above result follows from essentially the same reasoning as the argument preceding Theorem~\ref{th:gc-query-complexity-general}. In particular, the total number of Gaussian cooling stages $k$ and the number of first-order queries per stage $N$ remain unchanged. In fact, the only change in the analysis is the fact that, as in Corollary~\ref{cor:translation-invariant-union-bound}, we can replace all instances of $d^{c_2}(d + \log(1/p))$ that arise in the general case with $d^{c_2} \log(d / p)$. The factor $d^{c_2}$ comes from the number of steps (per sample) at each stage of cooling required by $\mathcal{A}$, while the logarithmic factor comes from the sample complexity estimate in Corollary~\ref{cor:translation-invariant-union-bound}.

\paragraph{Choosing a Sampler.} To demonstrate the utility of Gaussian cooling, we now present explicit query complexity bounds for two samplers: ULMC and FORS. The necessary background on FORS is in Appendix~\ref{app:fors}, while the corresponding background on ULMC can be found in Appendix~\ref{app:ulmc}. In practice, ULMC is simple to implement and parallelize over many chains, making it an attractive choice for estimating preconditioners for high-dimensional distributions. In the numerical experiments in Section~\ref{sec:num-exps}, we therefore use ULMC. However, very recent progress on the optimal dimension and condition-number dependence of log-concave sampling via FORS enables us to present even better theoretical guarantees. The FORS algorithm is equipped with the \textit{high accuracy} mixing guarantee in Theorem~\ref{th:fors-mixing}, so Theorem~\ref{th:gc-query-complexity-general} yields the following corollary.

\begin{corollary}
\label{cor:fors-query-complexity}
FORS (Algorithm~\ref{alg:fors}) is a high accuracy sampler satisfying
Assumption~\ref{asmp:sampling-algorithm} with constants $c_1 = 2/3$, $c_2 = 1/3$, and $c_3 = 5/3$.
Consequently, Gaussian cooling with FORS taking the place of $\mathcal{A}$ returns a preconditioner
satisfying Equation~\ref{eq:final-precond-guarantee} with probability at least $1 - p$ using
\begin{align*}
    N n K = \begin{cases}
    \widetilde{O}\left(\frac{d^{1/3} \left[d + \log(1/p)\right] \log(\kappa)}{\widetilde{\alpha}^{16/3} \delta^{2} \log(1 / \gamma)}\right) & \text{general target}, \\
    \widetilde{O}\left(\frac{d^{1/3} \log(1/p) \log(\kappa)}{\widetilde{\alpha}^{16/3} \delta^{2} \log(1 / \gamma)}\right) & \text{translation-invariant target}
    \end{cases}
\end{align*}
queries to $U$ and its gradient.
\end{corollary}

For ULMC, Theorem~\ref{th:ulmc-mixing} presents a low accuracy guarantee, allowing us to apply Theorem~\ref{th:gc-query-complexity-general} with parameters $c_1 = 3/2$, $c_2 = 1/2$, and $c_3 = 1$. The resulting query complexity is $\widetilde{O}(d^{3/2} \log(\kappa))$ in the general case and $\widetilde{O}(d^{1/2} \log(\kappa))$ in the translation-invariant case; see Corollary~\ref{cor:ulmc-query-complexity} for a full statement of the ULMC analogue of Corollary~\ref{cor:fors-query-complexity}.

\section{Examples and Applications}
\label{sec:num-exps}
In this section, we discuss two examples to which our framework naturally applies. In Section~\ref{subsec:well-cond-transf}, we discuss a generic problem in which a well-conditioned problem is transformed under a potentially ill-conditioned, unknown affine transformation. This setting can arise, for example, in physical problems when the underlying coordinate system changes. In Section~\ref{subsec:lattice-phi4}, we discuss applications to a simple discretized field theory, the lattice $\phi^4$ model. Throughout, we measure the success of various preconditioner estimation strategies in terms of their relative condition number. In particular, if $\Sigma$ is the covariance of the target distribution and $\widehat{\Sigma}$ is an estimate, then we define $\kappa_{\mathrm{rel}}(\widehat{\Sigma}, \Sigma) := \kappa(\widehat{\Sigma}^{-1/2} \Sigma \widehat{\Sigma}^{-1/2})$.
The closer that this quantity is to one, the better the estimate of the true covariance. When it is clear from context, we write $\kappa_{\mathrm{rel}}$ instead of $\kappa_{\mathrm{rel}}(\widehat{\Sigma}, \Sigma)$.

We compare the performance of Algorithm~\ref{alg:gaussian-cooling}, and in the translation-invariant setting, the performance of Algorithm~\ref{alg:gaussian-cooling-translation-invariant}, to two baselines. First, we compare to a ``raw'' ULMC baseline. Given the same number of queries as Algorithm~\ref{alg:gaussian-cooling}, we use Algorithm~\ref{alg:ulmc} to produce approximate samples from the target distribution, without any preconditioning. We then estimate the target covariance from the resulting samples. In the translation-invariant setting, estimating the dense empirical covariance is impractical, so we compute the unpreconditioned baseline in Fourier space using Equation~\ref{eq:translation-invariant-alg} when the dimension of the underlying distribution is greater than $400$. We refer to the translation-invariant version of this baseline as \textit{translation-averaged ULMC}.

We also compare to a stagewise empirical preconditioner estimate akin to the procedure that MCMC packages such as \cite{Carpenter2017} use in practice. Given the same number of chains $n$ and samples per chain $N$, this procedure carries out $K$ stages of preconditioner estimation. At the end of stage $k$, the sample covariance $\widetilde{\Sigma}_n^{(k)}$ is computed and then used as a preconditioner. In all cases, we ensure that any methods that we compare receive the same number of first-order queries. For reference preconditioners (i.e., no adaptive preconditioning or cooling takes place), this simply means that we run $n$ chains for $KN$ steps each.

\subsection{Well-Conditioned Problems Under an Unknown Affine Transformation}
\label{subsec:well-cond-transf}
The most straightforward application of the algorithm developed above is to well-conditioned problems transformed by unknown affine maps. Suppose that $\pi(x) \propto \exp(-U(x))$, where the potential $U : \R^d \to \R$ satisfies Equation~\ref{eq:log-concave-smooth} with condition number $\kappa = L / \mu = O(1)$. Now, given an arbitrary positive definite matrix $B \succ 0$, we consider the transformed distribution $\tilde{\pi}(x) \propto \exp(-U(L^{-1/2} B^{1/2} x))$. The potential $\widetilde{U}(x) := U(L^{-1/2} B^{1/2} x)$ then satisfies $\frac{\mu}{L} B \preceq \nabla^2 \widetilde{U}(x) \preceq B$
for all $x \in \R^d$.
Hence, for $\alpha = \mu / L = \Theta(1)$, the transformed distribution $\tilde{\pi}$ satisfies Equation~\ref{eq:true-conditioning}, with the geometry specified by the transformation $B$.

In the following numerical experiments, we take $\pi$ to be a well-conditioned Gaussian measure on $\R^d$. Then, with only black-box access to the quadratic potential $U$ and its gradient, Algorithm~\ref{alg:gaussian-cooling} provides an efficient method for estimating the covariance of the transformed distribution, with no dependence on the conditioning of the transformation $B$. In this sense, our method allows one to sample efficiently from well-conditioned distributions that have been corrupted by an unknown and potentially ill-conditioned affine transformation. In Figure~\ref{fig:transformed-gaussian-across-kappa-b}, we demonstrate that the quality of the preconditioner formed by Gaussian cooling does not degrade as $\kappa(B)$ increases. Similarly, Figure~\ref{fig:transformed-gaussian-large-kappa-b} shows that for $\kappa(B) = 10^6$, Gaussian cooling is able to produce a preconditioner whose relative condition number is $\Theta(1)$ in approximately seven iterations.

\begin{figure}
    \centering
    \begin{minipage}[t]{0.49\linewidth}
        \centering
        \includegraphics[width=\linewidth]{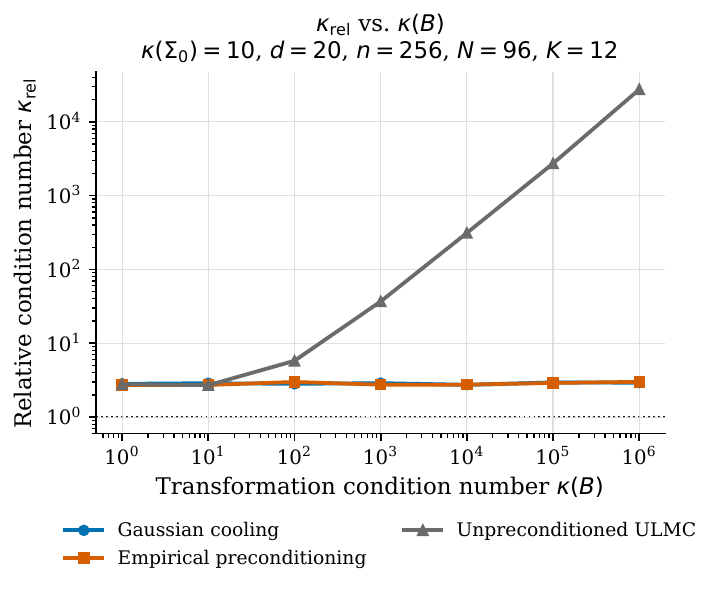}
        \captionof{figure}{As the condition number $\kappa(B)$ of the transformation applied to the well-conditioned distribution increases, empirical covariance from unpreconditioned ULMC degrades. Gaussian cooling as specified in Algorithm~\ref{alg:gaussian-cooling} and stagewise empirical preconditioning perform comparably.}
        \label{fig:transformed-gaussian-across-kappa-b}
    \end{minipage}
    \hfill
    \begin{minipage}[t]{0.49\linewidth}
        \centering
        \includegraphics[width=\linewidth]{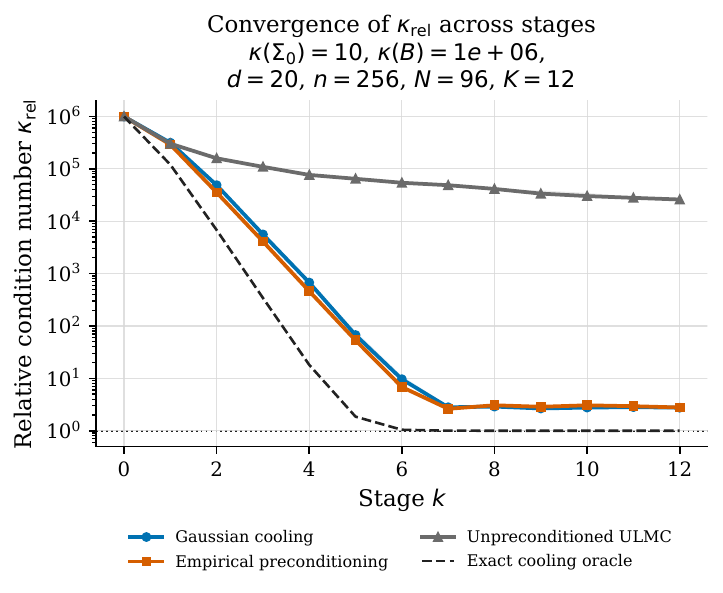}
        \captionof{figure}{When $d = 20$ and $\kappa(B) = 10^6$, both Gaussian cooling and stagewise empirical preconditioning can attain a constant relative condition number within approximately seven stages. The full empirical covariance fails to push the relative condition number below $10^4$.}
        \label{fig:transformed-gaussian-large-kappa-b}
    \end{minipage}
\end{figure}

\subsection{Lattice $\phi^4$ Theory}
\label{subsec:lattice-phi4}
The lattice $\phi^4$ model, a standard discretization of the classical $\phi^4$ field theory \cite[Chapter 1.6]{Bauerschmidt2019}, provides a more specialized testbed for our algorithm. Its quartic term results in a large condition number, even after truncating the quartic term past some radius to ensure smoothness. In practice, one typically estimates the covariance (or a diagonal approximation of it) from samples produced during the warm-up period of an HMC chain, as in \cite{Lunts2023}.
 
Let $\Lambda = \{(x_1, x_2) \in \Z^2 : 0 \leq x_1, x_2 \leq d - 1\}$ be the periodic square lattice of side length $d$, and write $x \sim y$ if $x$ and $y$ are neighbors, with periodic boundary. To each pair of sites we associate a symmetric coupling constant $\beta_{xy} = \beta_{yx} \geq 0$, with $\beta_{xy} = 0$ unless $x \sim y$, and following \cite{Bauerschmidt2019}, we define the weighted Laplacian
\begin{align*}
    (\Delta_\beta \phi)_x := \sum_{y \in \Lambda} \beta_{xy} (\phi_x - \phi_y)
\end{align*}
for each scalar field $\phi : \Lambda \to \R$. The Hamiltonian
\begin{align}
\label{eq:lattice-phi4-potential}
    H(\phi) := \frac{1}{2} \langle \phi, (\Delta_\beta + m I) \phi \rangle + \frac{\lambda}{4} \sum_{x \in \Lambda} \phi_x^4
\end{align}
induces a distribution $\pi(\phi) \propto \exp(-H(\phi))$ over $\R^{d^2}$; see Appendix~\ref{app:additional-phi4-facts} for further background. This model has translation-invariant covariance, so Theorem~\ref{th:query-complexity-translation-invariant} applies.
 
We compare translation-invariant Gaussian cooling (Algorithm~\ref{alg:gaussian-cooling-translation-invariant}) to translation-invariant empirical preconditioning, which estimates the covariance in stages via Equation~\ref{eq:translation-invariant-alg}. Figure~\ref{fig:phi4-large-experiments} shows that translation-invariant Gaussian cooling provides an effective preconditioner on lattices of sidelength up to $d = 512$ and is comparable to empirical preconditioning across all regimes. Note that for a lattice with sidelength $d$, the target distribution is a probability measure over $\R^{d^2}$, so our largest experiments take the target dimension to be $d^2 = 262,144$. In the small mass regime $m = 0.001$, Figure~\ref{fig:phi4-small-mass} shows that both methods begin to outperform the unpreconditioned chain (with the covariance estimated in Fourier space) given the same number of gradient queries. Figure~\ref{fig:phi4-infinite-radius} shows that our method continues to perform well on the true quartic potential, which is not log-smooth.
 
\begin{figure}[!t]
    \centering
    \begin{minipage}[t]{0.49\linewidth}
        \centering
        \includegraphics[width=\linewidth]{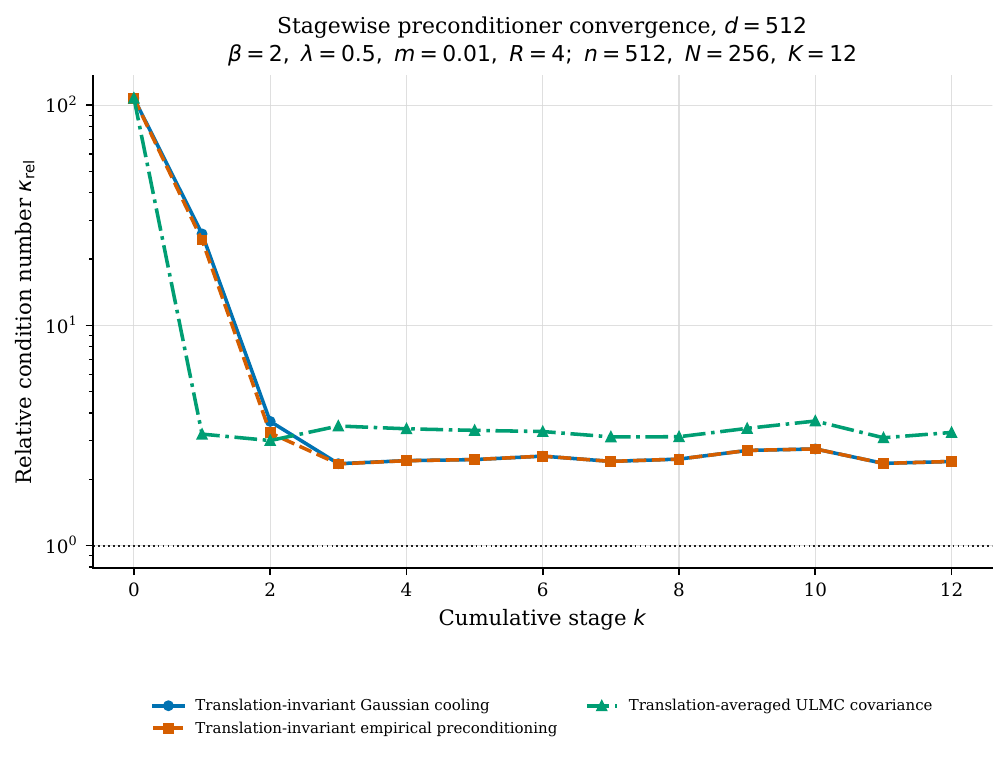}
        \captionof{figure}{On lattices of sidelength up to $d = 512$ with mass $m = 0.01$, both translation-invariant Gaussian cooling and translation-invariant empirical preconditioning obtain a constant relative condition number within several steps.}
        \label{fig:phi4-large-experiments}
    \end{minipage}
    \hfill
    \begin{minipage}[t]{0.49\linewidth}
        \centering
        \includegraphics[width=\linewidth]{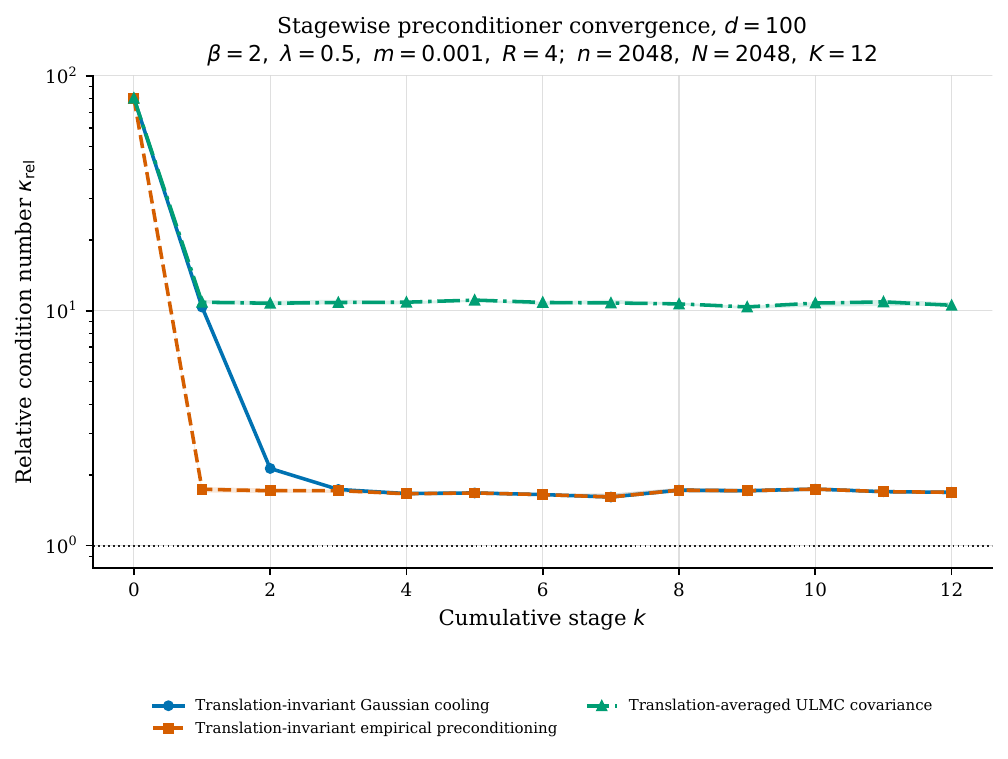}
        \captionof{figure}{For sidelength $d = 100$, decreasing the mass to $m = 0.001$ begins to distinguish Gaussian cooling and empirical preconditioning from translation-averaged ULMC.}
        \label{fig:phi4-small-mass}
    \end{minipage}
\end{figure}
 
In Figure~\ref{fig:phi4-iat}, we use \texttt{emcee} \cite{emcee} to compute the integrated autocorrelation time (IAT) $\tau_{\mathrm{int}}$ of the two-point correlators between the origin and all other sites. The convergence criterion $T \geq 50 \tau_{\mathrm{int}}$ is met for nearly all sites with chain length $T = 10000$ and burn-in $N_{\mathrm{burn}} = 1000$. Appendix~\ref{app:additional-phi4-facts} studies the effect of $N$ and $n$ on the performance of Algorithm~\ref{alg:gaussian-cooling}.
 
\begin{figure}[!t]
    \centering
    \begin{minipage}[t]{0.49\linewidth}
        \centering
        \includegraphics[width=\linewidth]{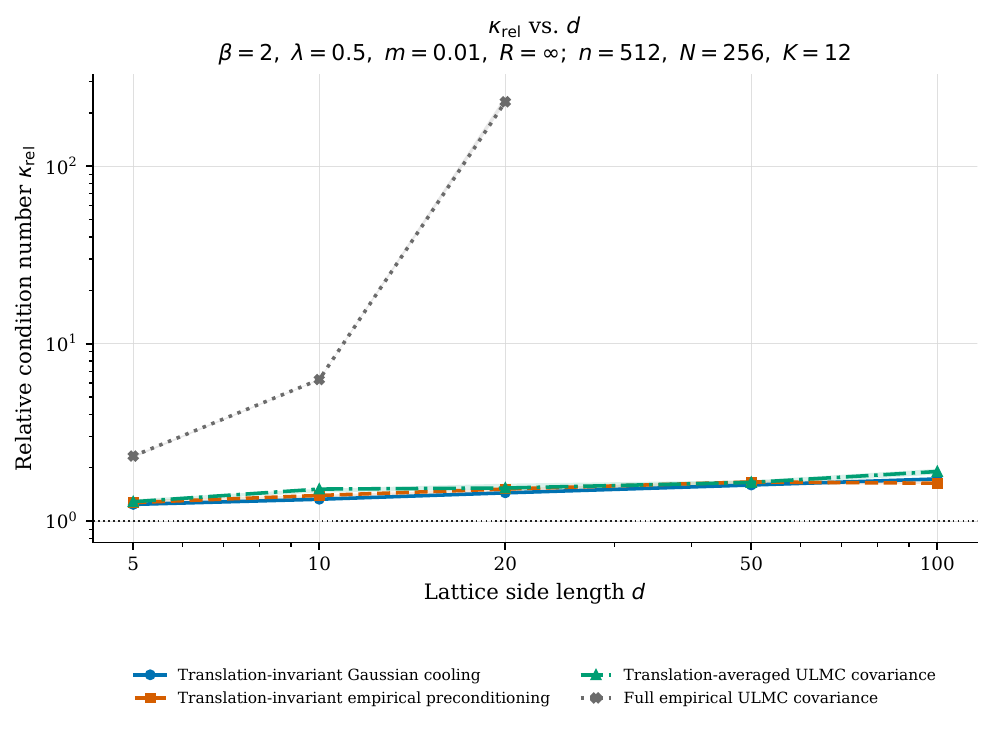}
        \captionof{figure}{For the genuinely quartic (untruncated) potential, all three translation-invariant methods perform comparably up to $d = 100$, while dense ULMC fails for $d = 5, 10, 20$ with $n$ and $N$ held fixed.}
        \label{fig:phi4-infinite-radius}
    \end{minipage}
    \hfill
    \begin{minipage}[t]{0.49\linewidth}
        \centering
        \includegraphics[width=\linewidth]{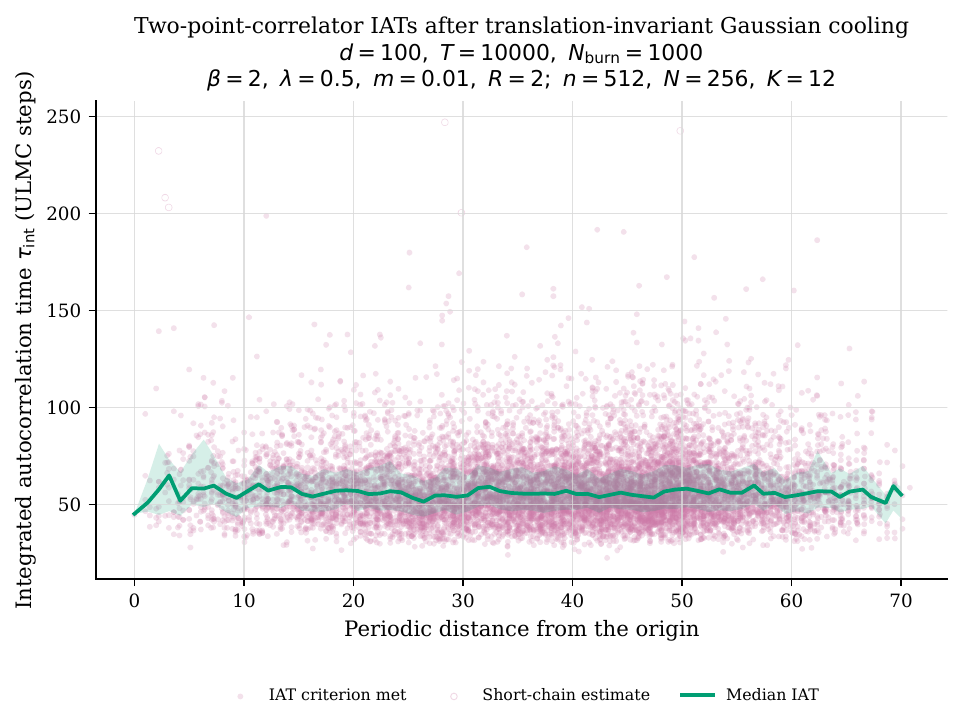}
        \captionof{figure}{With the translation-invariant Gaussian cooling preconditioner at $d = 100$, the two-point correlators between $(0, 0)$ and all other sites have median IAT $\tau_{\mathrm{int}} \approx 50$ for a ULMC chain of length $10,000$.}
        \label{fig:phi4-iat}
    \end{minipage}
\end{figure}

\section{Conclusion}
Gaussian cooling provides a theoretically grounded method for mass matrix estimation across a broad class of sampling problems. The first-order query complexity of our algorithm, as presented in Theorem~\ref{th:gc-query-complexity-general}, will only continue to improve as the theory of log-concave sampling advances. For instance, \cite{chen2026acceleratedhighaccuracysamplingwarm} and \cite{chen2026smoothedpicardhamiltonianmonte} indicate that one may be able to improve upon the first-order query complexity of FORS, obtaining a high accuracy query complexity of $\widetilde{O}(\sqrt{\kappa} d^{1/4} \mathrm{polylog}(1/\varepsilon))$ to obtain a sample that is within $\varepsilon$ distance of a log-concave target in total variation distance. On the practical side, our algorithm provides a compelling, easy-to-implement alternative for preconditioning high-dimensional distributions, matching the performance of widely-used warmup strategies \cite{Carpenter2017} while permitting rigorous guarantees.

\subsection*{AI Use Statement}
In this work, we used generative AI tools for the parts of the implementation of Algorithm~\ref{alg:gaussian-cooling}, with Algorithm~\ref{alg:ulmc} as the underlying sampler. For each implementation, detailed scaffolds of the desired algorithms and functionality were provided in the form of pseudocode. We did not use generative AI tools to help develop theoretical models or conceptual frameworks, formulate mathematical claims, provide critical ingredients for proving mathematical claims, assist in the writing of proofs, propose or refine hypotheses, support qualitative and thematic data analysis, or interpret results.

In addition to the above tasks, we used generative AI tools to find and summarize relevant literature, brainstorm high-level ideas for mathematical proofs, and proofread our final paper for readability and minor errors. All code and figures produced using generative AI tools were checked by the authors. We take responsibility for the final content of this work.

\subsection*{Acknowledgments}
This work was partially supported by
the U.S. Department of Energy, Office of Science,
Accelerated Research in Quantum Computing Centers, Quantum Utility through Advanced Computational Quantum Algorithms, grant no. DE-SC0025572,
the Applied Mathematics Program of the US Department of
Energy (DOE) Office of Advanced Scientific Computing Research under contract number DE-AC02-05CH11231, a Sloan Research Fellowship, the National Science Foundation under Grant No. 2543506, and the National Science Foundation Graduate Research Fellowship Program under Grant No. 2146752.

\printbibliography

\clearpage
\appendix
\section{Outline}
The majority of the mathematical content of our work can be found in the remainder of the appendix, structured as follows. First, Appendix~\ref{app:tech-assumps} outlines the notation and technical assumptions that underlie our theoretical contributions. Then, Appendix~\ref{app:log-concave-complexity} gives a more in-depth overview of the first-order query complexity of log-concave sampling. After this, Appendices~\ref{app:fors} and ~\ref{app:ulmc} describe the two samplers, FORS and ULMC, that we use for our theoretical results and numerical experiments in Section~\ref{sec:query-complexity} and Section~\ref{sec:num-exps}, respectively. 

Once this background information has been provided, Appendix~\ref{app:query-comp-proofs} provides rigorous proofs of Theorem~\ref{th:gc-query-complexity-general} and its translation-invariant counterpart. The bulk of the technical work regarding concentration of the sample covariance and the corresponding preconditioning guarantee that this concentration provides is then given in Appendix~\ref{app:tech-proofs}. Appendix~\ref{app:additional-phi4-facts} gives additional background and experiments on the lattice $\phi^4$ model. 

Appendix~\ref{app:experiment-details} outlines our experimental setup, including all relevant implementation details. All code used to produce the figures in Section~\ref{sec:num-exps} can be found in the associated \href{https://github.com/jakehofgard/gaussian-cooling}{GitHub repository}. We conclude with Appendix~\ref{app:additional-algs}, describing various extensions and variants of Algorithm~\ref{alg:gaussian-cooling} that may be more suitable for various applications, and Appendix~\ref{app:cautionary-tale}, a simple example of a distribution that satisfies Assumptions~\ref{asmp:conditioning}--~\ref{asmp:global-minimizer} but has $\alpha = O(1 / \kappa)$, where $\kappa$ is the condition number of the target distribution.
\section{Notation and Technical Assumptions}
\label{app:tech-assumps}
In this section, we provide notation and well-established results that we will use throughout the analysis of our algorithm. Throughout, we generally use the notation $\|\cdot\|$ for two purposes: the Euclidean norm on $\R^d$ and the spectral norm on matrices in $\R^{d \times d}$. Throughout, we also denote $[n] := \{1, \ldots, n\}$ for concision.

We begin with some preliminaries on the information-theoretic divergences that will aid in our analysis. Throughout, $\sP(\R^d)$ denotes the space of probability measures on $\R^d$ and
\begin{align*}
    \sP_2(\R^d) := \left\{\mu \in \sP(\R^d) : \int \|x\|^2 d\mu(x) < \infty\right\}
\end{align*}
is the space of probability measures with finite second moment.
\begin{definition}
\label{def:chi2-renyi}
If $\mu, \nu \in \sP_2(\R^d)$ are two probability measures such that $\mu \ll \nu$, then the $\chi^2$-divergence of $\mu$ relative to $\nu$ is defined by
\begin{align}
\label{eq:chi2-def}
    \chi^2(\mu \| \nu) := \int_{\R^d} \left(\frac{d \mu}{d\nu} - 1 \right)^2 d\nu = \int_{\R^d} \left(\frac{d \mu}{d\nu}\right)^2 d\nu - 1.
\end{align}
More generally, the R\'enyi divergence of order $q$ between $\mu$ and $\nu$ is defined by
\begin{align}
\label{eq:renyi-q-def}
    \sR_q(\mu \| \nu) := \frac{1}{q - 1} \log\left(\int_{\R^d} \left(\frac{d \mu}{d \nu}\right)^{q} d\nu \right).
\end{align}
Note that
\begin{align*}
    \sR_2(\mu \| \nu) = \log(1 + \chi^2(\mu \| \nu)),
\end{align*}
with both sides infinite if $\mu$ is not absolutely continuous with respect to $\nu$. 
\end{definition}
The $\chi^2$-divergence admits the following concise variational characterization.
\begin{theorem}[Theorem 7.26 and Example 7.4 in \cite{Polyanskiy2024}]
\label{thm:chi2-variational}
Suppose $\mu, \nu \in \sP_2(\R^d)$ are such that $\mu \ll \nu$ and that $\mathcal{G}$ denotes the set of measurable functions $g : \R^d \to \R$ such that $g \in L^2(\nu)$ and $\Var_{x \sim \nu}(g(x)) > 0$. Then, we can express
\begin{align*}
    \chi^2(\mu \| \nu) = \sup_{g \in \mathcal{G}} \frac{(\E_{x \sim \mu}[g(x)] - \E_{x \sim \nu}[g(x)])^2}{\Var_{x \sim \nu}(g(x))}.
\end{align*}
\end{theorem}
Beyond this point, $\E_{\mu}[g(x)] := \E_{x \sim \mu}[g(x)]$ for measurable $g : \R^d \to \R$, with identical notation for variances and covariances.

We also aim to relate the condition in Equation~\ref{eq:true-conditioning} to a condition on the covariance matrix of the target distribution $\pi$. Using a combination of a Brascamp--Lieb inequality and a Cram\'er--Rao bound, if $U$ satisfies the condition in Equation~\ref{eq:true-conditioning}, then we have that
\begin{align}
\label{eq:inv-cov-condition}
    \alpha B \preceq \Sigma^{-1} \preceq B,
\end{align}
where $\Sigma$ is the covariance of the equilibrium measure $\pi(x) \propto \exp(-U(x))$.
To obtain this inequality, we recall the Brascamp--Lieb and Cram\'er--Rao inequalities, respectively.
\begin{theorem}[Theorem 2.2.9 in \cite{chewi-book}]
If $\pi \propto \exp(-V)$, where $V$ is strictly convex in $\R^d$ and twice continuously differentiable, then for all $f : \R^d \to \R$,
\begin{align*}
    \Var_{\pi} f \leq \E_{\pi} \langle \nabla f, (\nabla^2 V)^{-1} \nabla f \rangle.
\end{align*}
\label{th:brascamp-lieb}
\vspace{-0.6cm}
\end{theorem}
The Cram\'er--Rao inequality then provides the accompanying lower bound, albeit in a bit more generality.
\begin{theorem}[Theorem 3.5.8 in \cite{chewi-book}]
If $\pi \propto \exp(-V)$ is a probability measure on $\R^d$, then for any twice continuously differentiable $f : \R^d \to \R$,
\begin{align*}
    \Var_{\pi} f \geq \langle \E_\pi \nabla f, (\E_\pi \nabla^2 V)^{-1} \E_{\pi} \nabla f \rangle.
\end{align*}
\label{th:cramer-rao}
\vspace{-0.6cm}
\end{theorem}
Equation~\ref{eq:inv-cov-condition} then follows immediately upon applying Theorem~\ref{th:brascamp-lieb} and Theorem~\ref{th:cramer-rao} with linear test functions of the form $f_v(x) = \langle x, v \rangle$, where $v \in \R^d$ is arbitrary. More specifically,
\begin{corollary}
\label{cor:sandwich}
Suppose that $\pi \propto \exp(-V)$ is a strictly log-concave measure on $\R^d$, with covariance $\Sigma$.
Then,
\begin{align*}
    (\E_{\pi} \nabla^2 V)^{-1} \preceq \Sigma \preceq \E_{\pi}[(\nabla^2 V)^{-1}].
\end{align*}
\end{corollary}
Equation~\ref{eq:inv-cov-condition} is then an immediate consequence of Corollary~\ref{cor:sandwich} under the assumption in Equation~\ref{eq:true-conditioning}. 

We impose the following assumptions throughout our analysis, unless otherwise specified.

\begin{assumption}
\label{asmp:conditioning}
Assume that there exist a positive-definite matrix $B \in \R^{d \times d}$ and $\alpha \in (0, 1)$ such that
\begin{align}
\label{eq:gen-cond}
\alpha B \preceq \nabla^2 U(x) \preceq B, \quad \forall x \in \R^d.
\end{align}
\end{assumption}

\begin{assumption}
\label{asmp:strong-log-concave-smooth}
Assume that there exist parameters $0 < \mu \leq L$ such that
\begin{align}
\label{eq:strong-log-concave-smooth}
\mu I \preceq \nabla^2 U(x) \preceq L I, \quad \forall x \in \R^d.
\end{align}
\end{assumption}

In line with the literature on log-concave sampling, if $\pi(x) \propto \exp(-U(x))$ for some potential $U : \R^d \to \R$ satisfying Assumption~\ref{asmp:strong-log-concave-smooth}, we say that $\pi$ is $\mu$-strongly log-concave and $L$-smooth. Of course, Assumption~\ref{asmp:strong-log-concave-smooth} implies that the potential $U$ is $\mu$-strongly convex. In particular, it admits a unique global minimizer $x^* \in \R^d$, leading us to our third and final assumption.

\begin{assumption}
\label{asmp:global-minimizer}
We assume oracle access to the global minimizer $x^*$, e.g., via precomputation using a standard convex optimization algorithm.
\end{assumption}

Note that the global minimizer may be computed once using a second-order algorithm such as the Newton--Raphson method, for example. As above, we define $\kappa := L / \mu$ to be the \textit{global} condition number of $\pi$.

To ease notation below, we define two functions that return the empirical mean and (biased) empirical covariance of a set of samples, respectively. In particular, $\widehat{m}_n : (\R^d)^n \to \R^d$ and $\widehat{S}_n : (\R^d)^n \to \R^{d \times d}$ are given by
\begin{align}
\label{eq:sample-mean-cov-def}
    &\widehat{m}_n(x_1, \ldots, x_n) := \frac{1}{n} \sum_{i=1}^n x_i, \notag \\
    &\widehat{S}_n(x_1, \ldots, x_n) := \frac{1}{n} \sum_{i=1}^n (x_i - \widehat{m}_n(x_1, \ldots, x_n))(x_i - \widehat{m}_n(x_1, \ldots, x_n))^\top.
\end{align}
When the set of samples $\{x_1, \ldots, x_n\}$ in question is clear from context, we abuse notation slightly and write $\widehat{m}_n$ and $\widehat{S}_n$. 
\section{First-Order Query Complexity of Log-Concave Sampling}
\label{app:log-concave-complexity}
In this appendix, we provide a high-level overview of recent results on the first-order query complexity of log-concave sampling. We distinguish between low accuracy samplers, whose first-order query complexities depend polynomially on the target error $1/\varepsilon$, and high accuracy samplers, whose dependence on $1/\varepsilon$ is instead polylogarithmic. For the first category of samplers, see Table~\ref{tab:low}. Importantly, ULMC belongs to the family of low accuracy samplers. On the other hand, Table~\ref{tab:high} contains analogous results for high accuracy samplers, including FORS.
\begin{table}
\centering\small
\renewcommand{\arraystretch}{1.3}
\begin{tabular}{@{}llll@{}}
\toprule
\textbf{Method} & \textbf{Metric} & \textbf{Query complexity} & \textbf{Ref.} \\
\midrule
LMC            & $\sqrt{\mu} W_2$           & $\Ot\!\left(\dfrac{\kappa d}{\varepsilon^{2}}\right)$ & \cite{dmm19} \\
LMC            & $\mathcal{R}_q$ & $\Ot\!\left(\dfrac{\kappa^{2} d q}{\varepsilon^{2}}\right)$ & \cite{chewi22poincare} \\
ULMC           & $\mathcal{R}_q$ & $\Ot\!\left(\dfrac{\kappa^{3/2} d^{1/2} q^{1/2}}{\varepsilon}\right)$ & \cite{altschuler-chewi-ulmc} \\
RM--ULMC       & $\mathsf{KL}$   & $\Ot\!\left(\dfrac{\kappa d^{1/3}}{\varepsilon^{2/3}}\right)$ & \cite{altschuler2026shifted} \\
RM--ULMC, low friction & $\mathsf{KL}$ & $\Ot\!\left(\dfrac{\kappa^{5/6} d^{5/3}}{\varepsilon^{2/3}}\right)$ & \cite{altschuler2026shifted} \\
\bottomrule
\end{tabular}

\caption{First-order query complexities for a handful of low-accuracy samplers. RM-ULMC refers to the randomized midpoint version of ULMC, analyzed in \cite{altschuler2026shifted}. Although the query complexity in order to obtain $\varepsilon$ error in $\sqrt{\mathsf{KL}}$ divergence scales better with $\kappa$ and $d$, we require a stronger $\chi^2$-divergence guarantee (equivalently, a Rényi divergence guarantee) for our theoretical analysis.}
\label{tab:low}
\end{table}

\begin{table}
\centering\small
\renewcommand{\arraystretch}{1.3}
\begin{tabular}{@{}lll@{}}
\toprule
\textbf{Method} & \textbf{Query complexity} & \textbf{Ref.} \\
\midrule
Metropolized HMC / MALA, feasible start
  & $\widetilde{\Omega}(\kappa\, d)$ & \cite{lee2021lower} \\
MALA, warm start
  & $\Ot(\kappa\, d^{1/2})$ & \cite{wu2022mala} \\
Prox. sampler $\to$ rejection sampling
  & $\Ot(\kappa\, d^{1/2})$ & \cite{fan23a} \\
ULMC $\to$ MALA $\to$ prox.\ sampler
  & $\Ot(\kappa\, d^{1/2})$ & \cite{altschuler-chewi-ulmc} \\
Metropolized HMC
  & $\Ot\!\bigl((\kappa^{3/2} + \kappa^{1/2}\kappa_{\rm H_2}^{1/2} + \kappa_{\rm H_2}^{2/3})\, d^{1/4}\bigr)$
  & \cite{zhang2026algorithmic} \\
FORS &  $\Ot\!\bigl(\kappa^{2/3} d^{1/3}\bigr)$
  & \cite{chen2026exact} \\
\bottomrule
\end{tabular}

\caption{First-order query complexities for a collection of high-accuracy samplers whose complexity scales as $\mathrm{polylog}(1/\varepsilon)$ with respect to the target error $\varepsilon$. All complexities refer to mixing in TV distance, although the last two rows also apply in $\sqrt{\chi^2}$, $\sqrt{\mu} W_2$, and $\sqrt{\mathsf{KL}}$. The guarantee for metropolized HMC in the penultimate row only holds under the stronger assumption that the Hessian is Lipschitz with respect to the Frobenius norm, with Lipschitz constant $\beta_{\rm H_2}$ and $\kappa_{\rm H_2} := \beta_{\rm H_2} / \mu^{3/2}$.}
\label{tab:high}
\end{table}
\section{First-Order Rejection Sampling}
\label{app:fors}
In order to minimize the dependence on the underlying dimension $d$, we turn to first-order rejection sampling (FORS), as presented in \cite{chen2026exact}. In particular, our theoretical results utilize a particular form of FORS that comes equipped with a strong mixing time guarantee in R\'enyi divergence; see Theorem~\ref{th:fors-mixing} below. For the time being, we only need to provide the formulation of FORS that our algorithm requires.

First, recall the \textit{underdamped Langevin diffusion} with friction parameter $\tilde{\gamma} > 0$ and $V : \R^d \to \R$ is a sufficiently smooth potential:
\begin{align}
\label{eq:uld}
\begin{cases}
d X_t = P_t dt, \\
d P_t = - \nabla V(X_t) dt - \tilde{\gamma} P_t dt + \sqrt{2 \tilde{\gamma}} d B_t.
\end{cases}
\end{align}
The invariant measure of this diffusion process is $\mathbf{\pi}(x, p) \propto \exp(-V(x) - \frac{1}{2}\|p\|_2^2)$, meaning that the $x$ marginal of the invariant measure is precisely the target distribution. Throughout, we abbreviate $z = (x, p) \in \R^{d} \times \R^d$ to denote the coordinates of the diffusion in phase space. To simulate the dynamics in Equation~\ref{eq:uld}, we first fix a base point $z_0 = (x_0, p_0)$. Then, the gradient term in Equation~\ref{eq:uld} is frozen at $X_0 = x_0$, yielding a linear SDE of the form
\begin{align}
\label{eq:linear-uld}
\begin{cases}
d \widetilde{X}_t = \widetilde{P}_t dt, \\
d \widetilde{P}_t = - \nabla V(X_0) dt - \tilde{\gamma} \widetilde{P}_t dt + \sqrt{2 \tilde{\gamma}} d B_t,
\end{cases}
\end{align}
with initial condition $\widetilde{X}_0 = x$ and $\widetilde{P}_0 = p$. Denote by $\bm{Q} := \bm{Q}_{z_0}$ the path measure corresponding to $\widetilde{\bm{Z}} := (\widetilde{X}_t, \widetilde{P}_t)_{t \in [0, T]}$, where $T > 0$ is fixed. Now, let
\begin{align*}
    \mu(x) := \frac{1}{\sqrt{2 \tilde{\gamma}}} (\nabla V(x) - \nabla V(x_0)).
\end{align*}
Applying Girsanov's theorem, there exists a path measure $\bm{P} = \bm{P}_{z_0}$ defined by
\begin{align}
\label{eq:girsanov-change}
\frac{d \bm{P}}{d \bm{Q}} = \exp\left(-\int_0^T \mu(\widetilde{X}_t) \cdot dB_t - \frac{1}{2} \int_0^T \|\mu(\widetilde{X}_t)\|^2 dt\right)
\end{align}
such that the process $\widetilde{B}_t := B_t + \int_0^t \mu(\widetilde{X}_s) ds$ is a $\bm{P}$-Brownian motion and the process $\widetilde{\bm{Z}}$ follows the dynamics in Equation~\ref{eq:uld}, driven by $\widetilde{B}_t$. To implement FORS, we design an estimator for the logarithm of Radon--Nikodym derivative in Equation~\ref{eq:girsanov-change}. For the second term in the exponent, note that such an estimator can be constructed easily by sampling $t \sim \mathrm{Unif}([0, T])$. On the other hand, to handle the stochastic integral, we first must discretize. Given an integer $N \in \N$, take $h = T / N$ to be the corresponding step size. Then, we can approximate the stochastic integral as
\begin{align*}
    \sum_{j=0}^{N-1} \mu(\widetilde{X}_{jh}) \cdot (B_{(j+1)h} - B_{jh}) &= \sum_{j=1}^N (\mu(\widetilde{X}_{jh}) - \mu(\widetilde{X}_{(j-1)h})) \cdot (B_{T} - B_{jh}) \\
    &= \frac{1}{\sqrt{2 \tilde{\gamma}}} \sum_{j=1}^N (\nabla V(\widetilde{X}_{jh}) - \nabla V(\widetilde{X}_{(j-1)h})) \cdot (B_{T} - B_{jh}).
\end{align*}

In particular, to give an unbiased estimate Equation~\ref{eq:girsanov-change}, discretized with step size $h$, we sample $t \sim \mathrm{Unif}([0, T])$. If $t \in [(j-1)h, jh)$, the estimate in question is then
\begin{align}
\label{eq:disc-reweighting-factor}
    -W(t; \widetilde{\bm{Z}}) = \frac{T}{\sqrt{2 \tilde{\gamma}} h} (\nabla V(\widetilde{X}_{jh}) - \nabla V(\widetilde{X}_{(j-1)h})) \cdot (B_{T} - B_{jh}) + \frac{T}{2} \|\mu(\widetilde{X}_t)\|^2.
\end{align}

Using $O(1)$ unbiased samples of the form Equation~\ref{eq:disc-reweighting-factor}, FORS approximates the change-of-measure in Equation~\ref{eq:girsanov-change}, allowing us to approximately sample from the path measure of Equation~\ref{eq:uld}. In order to accomplish this, FORS proceeds using the following rejection sampling process. Below, $\mathrm{Clip}_{[0,1]}(x) := \min\{1, \max\{0, x\}\}$.

\begin{algorithm}
\caption{First-order rejection sampling (FORS)}
\label{alg:fors}
\begin{algorithmic}[1]
\REQUIRE Parameter $B > 0$, proposal distribution $\bm{Q}$, unbiased estimator $W$. 
\FOR{$i = 1, 2, \ldots$}
\STATE Sample $J \sim \mathrm{Poisson}(2 B)$.
\STATE Sample $t_1, \ldots, t_J \sim \mathrm{Unif}([0, T])$.
\STATE Sample $(W(t_1; \bm{\widetilde{Z}}), \ldots, W(t_J; \bm{\widetilde{Z}}), \widetilde{Z}_T)$ following the proposal $\bm{\widetilde{Z}} \sim \bm{Q}$.
\STATE $p_J \gets \prod_{j=1}^J \mathrm{Clip}_{[0,1]}\left(\frac{B + W(t_j; \bm{\widetilde{Z}})}{2B}\right)$
\STATE Output $\widetilde{Z}_T$ with probability $p_J$.
\ENDFOR
\end{algorithmic}
\end{algorithm}

Although the above algorithm may appear mysterious at a glance, its output satisfies a rather concrete guarantee.

\begin{theorem}[Theorem 2.3 in \cite{chen2026exact}]
\label{th:fors-guarantee}
Algorithm~\ref{alg:fors} outputs a random point $\widetilde{Z}_T$ whose distribution $\widehat{\mu}$ is the marginal of a path measure $\widehat{\bm{P}}$, satisfying
\begin{align*}
    \frac{d \widehat{\bm{P}}}{d \bm{Q}} \propto \exp(\widehat{w}(\widetilde{\bm{Z}})), \quad \widehat{w}(\widetilde{\bm{Z}}) := \frac{1}{T} \int_0^T \mathrm{Clip}_{[-B, B]}(W(t; \widetilde{\bm{Z}})) dt,
\end{align*}
and with probability at least $1 - \delta$, the number of queries of $W$ is bounded by $3 B e^{2B} \log(2 / \delta)$.
\end{theorem}

The clipping parameter $B > 0$ directly controls the number of queries to $W$, defined in Equation~\ref{eq:disc-reweighting-factor}, that must be made during a single call to Algorithm~\ref{alg:fors}. Although taking this parameter larger will decrease the error incurred by FORS, the exponential dependence of the query complexity of Algorithm~\ref{alg:fors} on $B$ means that we must take $B$ to be a small constant in practice. For example, the authors of \cite{chen2026exact} take $B \in [1, 2]$. Fortunately, even for this choice of $B$, the error incurred by FORS can be controlled precisely. The following theorem is then adapted from Theorem 3.2(ii) in \cite{chen2026exact}, specialized to the case of a distribution satisfying Assumption~\ref{asmp:strong-log-concave-smooth}.

\begin{theorem}
\label{th:fors-mixing}
Suppose that $\pi(x) \propto \exp(-V(x))$ satisfies Assumption~\ref{asmp:strong-log-concave-smooth}. For any $q \geq 2$, set $\Delta_q := d + \sR_{2q}(\nu \| \pi)$, where $(X_0, P_0) \sim \nu$ and $\pi$ is the invariant measure of the underdamped Langevin diffusion in Equation~\ref{eq:uld} with friction $\widetilde{\gamma} = \Theta(\sqrt{\mu})$. Denoting
\begin{align*}
    M_q := q + \log\left(\frac{q \kappa^{2/3} \Delta_q^{1/3} \log(q \Delta_q / \varepsilon^2)}{\varepsilon^2}\right),
\end{align*}
we may take
\begin{align*}
    T = \Theta\left(\mu^{-1/2} \kappa^{-2/3} \Delta_q^{-1/3} M_q^{-2/3}\right)
\end{align*}
and
\begin{align*}
    N = \Theta\left(\kappa^{2/3} \Delta_q^{1/3} M_q^{2/3} \log(q \Delta_q / \varepsilon^2)\right).
\end{align*}
so that, with probability at least $1 - p$, the output distribution $\widehat{\nu}$ of $N$ steps of Algorithm~\ref{alg:fors} with timestep $T$ satisfies $\sR_q(\widehat{\nu} \| \pi) \leq \varepsilon^2$ with $O(N + \log(1 / p))$ first-order queries.
\end{theorem}

In our setting, we can control the initialization error for approximately sampling from the invariant measure of the underdamped Langevin diffusion in Equation~\ref{eq:uld} by taking an appropriate Gaussian initialization. If the target distribution satisfies Assumption~\ref{asmp:strong-log-concave-smooth}, and Algorithm~\ref{alg:fors} is initialized with $(X_0, P_0) \sim \nu_0 := \sN(x^*, L^{-1} I_d) \otimes \sN(0, I_d)$, then the following technical lemma allows us to bound the initialization error $\sR_{2q}(\nu_0 \| \pi)$, where $\pi(x, p) \propto \exp(-V(x) - \frac{1}{2}\|p\|_2^2)$ is the target measure. For any \textit{any} $q \geq 1$, we can prove the following initialization guarantee.

\begin{lemma}
\label{lem:gaussian-initialization-renyi}
Borrowing notation from \cite{chen2026exact}, if $\pi_V(x) = Z^{-1} \exp(-V(x))$ is the first marginal of the invariant measure of the underdamped Langevin diffusion in Equation~\ref{eq:uld}, then for all $q > 1$ 
\begin{align*}
    \sR_{2q}(\nu_0 \| \pi) = \sR_{2q}(\sN(x^*, L^{-1} I_d) \| \pi_V) \leq \frac{d}{2} \log \kappa.
\end{align*}
\end{lemma}
\begin{proof}
The first equality above holds because the initialization $\nu_0$ and the invariant measure $\pi$ are product measures that share the same marginal in $p$. For the remaining inequality, observe that the density ratio in question can be bounded pointwise as
\begin{align*}
    Z \left(\frac{L}{2\pi}\right)^{d/2} \exp\left(V(x) - \frac{L}{2}\|x - x^*\|_2^2\right) \leq Z \exp(V(x^*)) \left(\frac{L}{2\pi}\right)^{d/2},
\end{align*}
as Assumption~\ref{asmp:strong-log-concave-smooth} ensures that
\begin{align*}
    V(x) \leq V(x^*) + \frac{L}{2}\|x - x^*\|_2^2,
\end{align*}
for all $x \in \R^d$. Meanwhile, strong convexity implies that
\begin{align*}
    Z = \int_{\R^d} \exp(-V(x)) dx &\leq \exp(-V(x^*)) \int_{\R^d} \exp\left(-\frac{\mu}{2}\|x - x^*\|_2^2\right) dx \\
    &= \exp(-V(x^*)) \left(\frac{2\pi}{\mu}\right)^{d/2}.
\end{align*}
Thus, the density ratio that we must evaluate is bounded pointwise $\kappa^{d/2}$. Hence,
\begin{align*}
    \sR_{2q}(\sN(x^*, L^{-1} I_d) \| \pi_V) \leq \log(\kappa^{d/2}) \leq d \log \kappa
\end{align*}
for all $q > 1$, proving the lemma.
\end{proof}

\section{Underdamped Langevin Monte Carlo}
\label{app:ulmc}
Although the query complexity guarantee of FORS is impressive, the algorithm is somewhat difficult to implement in practice. Instead, practitioners may prefer a more standard discretization of the underdamped Langevin diffusion, described in detail in \cite{altschuler-chewi-ulmc}. Rather than using the Euler--Maruyama discretization of Equation~\ref{eq:uld} to obtain a practical algorithm, we fix the nonlinearity that manifests in the gradient term at each time step, obtaining a piecewise linear SDE that is exactly integrable. In particular, if $h > 0$ is a timestep parameter, then we may consider SDE with piecewise linear coefficients given by
\begin{align}
\label{eq:ulmc}
\begin{cases}
d \widehat{X}_t = \widehat{P}_t dt, \\
d \widehat{P}_t = - \nabla V(\widehat{X}_{kh}) dt - \tilde{\gamma} \widehat{P}_t  dt + \sqrt{2 \tilde{\gamma}} d B_t
\end{cases}
\end{align}
on the time interval $t \in [kh, (k + 1)h]$, where $k \in \N$. Conditioned on the value of $(\widehat{X}_{kh}, \widehat{P}_{kh})$, we have the following characterization of the next pair of iterates $(\widehat{X}_{(k+1)h}, \widehat{P}_{(k+1)h})$.
\begin{proposition}[Lemma 5.3.8 in \cite{chewi-book}]
\label{prop:exact-ulmc-int}
Conditioned on the value of $(\widehat{X}_{kh}, \widehat{P}_{kh})$, the next iterate of
Equation~\ref{eq:ulmc} is distributed as
$(\widehat{X}_{(k+1)h}, \widehat{P}_{(k+1)h}) \sim \sN(M_{(k+1)h}, \Sigma_{\tilde{\gamma}, h})$, where
\begin{align*}
    M_{(k+1)h} := \begin{bmatrix}
    \widehat{X}_{kh} + \tilde{\gamma}^{-1}\left(1 - e^{-\tilde{\gamma} h}\right) \widehat{P}_{kh}
    - \tilde{\gamma}^{-1}\left(h - \tilde{\gamma}^{-1}\left(1 - e^{-\tilde{\gamma} h}\right)\right) \nabla V(\widehat{X}_{kh}) \\
    e^{-\tilde{\gamma} h} \widehat{P}_{kh} - \tilde{\gamma}^{-1}\left(1 - e^{-\tilde{\gamma} h}\right) \nabla V(\widehat{X}_{kh})
    \end{bmatrix}
\end{align*}
and
\begin{align*}
    \Sigma_{\tilde{\gamma}, h} := \begin{bmatrix}
    \frac{2}{\tilde{\gamma}}\left[h - \frac{2}{\tilde{\gamma}}\left(1 - e^{-\tilde{\gamma} h}\right) + \frac{1}{2 \tilde{\gamma}}\left(1 - e^{-2 \tilde{\gamma} h}\right)\right]
    & \frac{1}{\tilde{\gamma}}\left(1 - e^{-\tilde{\gamma} h}\right)^2 \\
    \frac{1}{\tilde{\gamma}}\left(1 - e^{-\tilde{\gamma} h}\right)^2
    & 1 - e^{-2\tilde{\gamma} h}
    \end{bmatrix} \otimes I_d.
\end{align*}
\end{proposition}
This result follows straightforwardly upon noting that, given $(\widehat{X}_{kh}, \widehat{P}_{kh})$, the SDE describing the evolution of $\widehat{P}_t$ on the interval $[kh, (k + 1)h]$ is simply the Ornstein--Uhlenbeck process, which can be integrated exactly. In turn, the evolution for $\widehat{X}_t$ is then a linear ODE that depends only on $\widehat{P}_t$. 
\begin{algorithm}
\caption{Underdamped Langevin Monte Carlo (ULMC)}
\label{alg:ulmc}
\begin{algorithmic}[1]
\REQUIRE Potential $V : \mathbb{R}^d \to \mathbb{R}$ with gradient oracle $\nabla V$, global minimizer $x^*$, friction parameter $\tilde{\gamma} > 0$, smoothness parameter $L > 0$, step size $h > 0$, number of iterations $N$ 
\STATE Draw $(\widehat{X}_0, \widehat{P}_0) \sim \sN(x^*, L^{-1} I_d) \otimes \sN(0, I_d)$
\FOR{$k = 0, 1, \ldots, N-1$}
    \STATE $\widehat{V}_{kh} \gets \nabla V(\widehat{X}_{kh})$
    \STATE Sample $(\eta_k, \xi_k) \sim \mathcal{N}(0, \Sigma_{\tilde{\gamma}, h})$
    \STATE $\widehat{X}_{(k+1)h} \gets \widehat{X}_{kh} + \dfrac{1 - e^{-\tilde{\gamma} h}}{\tilde{\gamma}}\, \widehat{P}_{kh} - \dfrac{\tilde{\gamma} h - 1 + e^{-\tilde{\gamma} h}}{\tilde{\gamma}^2} \widehat{V}_{kh} + \eta_k$
    \STATE $\widehat{P}_{(k+1)h} \gets e^{-\tilde{\gamma} h}\, \widehat{P}_{kh} - \dfrac{1 - e^{-\tilde{\gamma} h}}{\tilde{\gamma}} \widehat{V}_{kh} + \xi_k$
\ENDFOR
\RETURN $\widehat{X}_{Nh}$
\end{algorithmic}
\end{algorithm}

Algorithm~\ref{alg:ulmc} forms the backbone of Gaussian cooling, allowing us to sample efficiently at each stage of the algorithm. In practice, we run many ULMC chains in parallel. More concretely, if $n \in \N$, then we may produce $n$ approximate samples from a target distribution with potential $V$ by running Algorithm~\ref{alg:ulmc} $n$ times in parallel. In shorthand, we refer to this procedure by
\begin{align}
\label{eq:n-ulmc-call}
    x_1, \ldots, x_n = \mathrm{ULMC}(V, x^*, \tilde{\gamma}, L, h, N, n)
\end{align}
when Algorithm~\ref{alg:ulmc} is utilized as a subprocess of Gaussian cooling. For simplicity, we summarize all of the parameters needed for an implementation of ULMC:
\begin{itemize}
    \item Relative preconditioning error $\delta > 0$.
    \item Conditioning parameter $\alpha \in (0, 1)$.
    \item Cooling parameter $\gamma \in (0, 1)$.
    \item Augmented conditioning parameter $\Tilde{\alpha} = \min\{\alpha, \gamma\}$.
    \item  Smoothness parameter $L = \tilde{\alpha}^{-1} (1 + \delta)$.
\end{itemize}
The choice of $L = \widetilde{\alpha}^{-1} (1 + \delta)$ comes from Equation~\ref{eq:good-covar-precond}, in which the preconditioned potential is $L$-smooth with this choice of $L$. We replace $\alpha$ with $\tilde{\alpha} \leq \alpha$ because while $\alpha$ itself may not be known, in all applications we consider in this work, a reasonable lower bound on $\alpha$ is known and used as the cooling parameter $\gamma$.

We invoke an intermediate result on the optimal scaling of ULMC, based on the analysis in \cite[Theorem 4.1]{altschuler-chewi-ulmc}. The specific discretization that we take for ULMC is as in \cite[Section 5.3.2]{chewi-book}. Below, $V$ is a potential from which we want to sample. We then form a piecewise system of linear SDEs, discretizing the underdamped Langevin dynamics, that can be integrated \textit{exactly} as in \cite[Lemma 5.3.8]{chewi-book}:
\begin{align}
\label{eq:ulmc-disc}
\begin{cases}
d \widehat{X}_t = \widehat{P}_t dt & t \geq 0, \\
d \widehat{P}_t = - (\nabla V(\widehat{X}_{nh}) + \widetilde{\gamma} \widehat{P}_t)dt + \sqrt{2 \widetilde{\gamma}} d B_t & t \in [nh, (n + 1)h].
\end{cases}
\end{align}
Above, $\widetilde{\gamma}$ is a friction parameter and $h$ is a step size parameter, both of which we specify below. For this discretization of ULMC, we have the following R\'enyi divergence guarantee, which naturally translates to a guarantee in $\chi^2$-divergence.
\begin{theorem}[Theorem 6.3.6 in \cite{chewi-book}]
\label{th:ulmc-mixing}
Assume that $V$ is $\mu$-strongly convex and $L$-smooth and that $\pi(x) \propto \exp(-V(x))$ for $x \in \R^d$, with $x^* \in \R^d$ the unique minimizer of $V$. For each $n \in \N$, let $\widehat{\pi}_{nh} = \mathrm{Law}(\widehat{X}_{nh})$ in Equation~\ref{eq:ulmc-disc}, with friction parameter $\widetilde{\gamma} = \sqrt{2 L}$ and initialization $\widehat{\pi}_0 = \sN(x^*, L^{-1} I_d) \otimes \sN(0, I_d)$. Then, for any $q \geq 2$ and step size $h = \widetilde{\Theta}(\eps / \sqrt{L \kappa d q})$, we obtain $\sR_q(\widehat{\pi}_{N h} \| \pi) \leq \eps^2$ after
\begin{align*}
    N = \widetilde{\Theta}\left(\frac{\kappa^{3/2} d^{1/2} q^{1/2}}{\eps}\right)
\end{align*}
iterations of Equation~\ref{eq:ulmc-disc}.
\end{theorem}
Using this as our sampler at each stage, we can fully characterize the query complexity of Gaussian cooling as in Corollary~\ref{cor:ulmc-query-complexity}. To be as explicit as possible, we quantify all the parameters that are required to use ULMC as the sampler for Gaussian cooling in Table~\ref{tab:params}.

\begin{table}
\centering\small
\renewcommand{\arraystretch}{1.3}
\begin{tabular}{@{}lll@{}}
\toprule
\textbf{Quantity} & \textbf{Symbol} & \textbf{Value} \\
\midrule
ULMC step size                  & $h$                & $\widetilde{\Theta}\!\left(\dfrac{\delta \widetilde{\alpha}^{3/2}}{(1 + \delta)^{3/2} d^{1/2}}\right)$ \\
ULMC friction parameter         & $\widetilde{\gamma}$ & $\sqrt{2 \widetilde{\alpha}^{-1} (1 + \delta)}$ \\
ULMC steps per sample           & $N$                & $\widetilde{\Theta}\!\left(\dfrac{(1 + \delta)^3 d^{1/2}}{\widetilde{\alpha}^{3} \delta}\right)$ \\
Samples per cooling stage       & $n$                & $\Theta\!\left(\dfrac{d (1 + \delta)^4}{\widetilde{\alpha}^4 \delta^2}\right)$ \\
Stages of Gaussian cooling      & $K$                & $O\!\left(\dfrac{\log \kappa}{\log(1 / \gamma)}\right)$ \\
\bottomrule
\end{tabular}

\caption{Parameters passed into ULMC and the number of iterations required for each subprocess of the Gaussian cooling algorithm.}
\label{tab:params}
\end{table}

We can now provide a precise analogue of Corollary~\ref{cor:fors-query-complexity} for ULMC:

\begin{corollary}
\label{cor:ulmc-query-complexity}
ULMC, as presented in Algorithm~\ref{alg:ulmc}, is a low accuracy sampler satisfying
Assumption~\ref{asmp:sampling-algorithm} with constants $c_1 = 3/2$, $c_2 = 1/2$, and $c_3 = 1$.
Consequently, Gaussian cooling with ULMC taking the place of $\mathcal{A}$ satisfies
Equation~\ref{eq:final-precond-guarantee} with probability at least $1 - p$ using
\begin{align*}
    N n K = \begin{cases}
    \widetilde{O}\left(\frac{d^{1/2} \left[d + \log(1/p)\right] \log(\kappa)}{\widetilde{\alpha}^{7} \delta^{3} \log(1 / \gamma)}\right) & \text{general target}, \\
    \widetilde{O}\left(\frac{d^{1/2} \log(1/p) \log(\kappa)}{\widetilde{\alpha}^{7} \delta^{3} \log(1 / \gamma)}\right) & \text{translation-invariant target}
    \end{cases}
\end{align*}
queries to $U$ and its gradient.
\end{corollary}

\section{Query Complexity Proofs}
\label{app:query-comp-proofs}
We begin with the proof of Proposition~\ref{prop:gc-induction}, which inductively establishes the preconditioning guarantee afforded by Algorithm~\ref{alg:gaussian-cooling} and gives a bound on the total number of cooling stages required to obtain a good preconditioner of the target distribution, as specified in Equation~\ref{eq:final-precond-guarantee}.
\begin{proof}[Proof of Proposition~\ref{prop:gc-induction}]
By construction,
\begin{align*}
    \alpha B + \gamma^k L I \preceq \nabla^2 U_k(x) \preceq B + \gamma^k L I
\end{align*}
for all $x \in \R^d$, so by Corollary~\ref{cor:sandwich}, it follows that
\begin{align*}
    \alpha B + \gamma^k L I \preceq (\Sigma^{(k)})^{-1} \preceq B + \gamma^k L I.
\end{align*}
Now, in the first step of the algorithm, we have that
\begin{align*}
    L I \preceq \nabla^2 U_0(x) \preceq 2L I, \quad x \in \R^d,
\end{align*}
so $U_0$ is strongly convex with condition number at most $2$. We now proceed inductively, showing that if
\begin{align}
    (1 + \delta)^{-1} \Sigma^{(k-1)} \preceq \widetilde{\Sigma}_n^{(k-1)} \preceq (1 + \delta) \Sigma^{(k-1)},
\end{align}
where $\widetilde{\Sigma}_n^{(k-1)}$ is the preconditioner formed in Line 6 of Algorithm~\ref{alg:gaussian-cooling} from $n$ samples approximately drawn from $\pi_{k-1}$, then $\widetilde{\Sigma}_n^{(k-1)}$ is a good preconditioner for $\pi_{k}$. Indeed, we have that
\begin{align*}
    \nabla^2 U_k(x) \preceq \alpha^{-1}(\alpha B + \alpha \gamma^k L I) \preceq \alpha^{-1} (\Sigma^{(k-1)})^{-1} \preceq \alpha^{-1} (1 + \delta) (\widetilde{\Sigma}_n^{(k-1)})^{-1}
\end{align*}
and
\begin{align*}
    \nabla^2 U_k(x) \succeq \widetilde{\alpha}(B + \gamma^{k-1} L I) \succeq \widetilde{\alpha} (\Sigma^{(k-1)})^{-1} \succeq \widetilde{\alpha} (1 + \delta)^{-1} (\widetilde{\Sigma}_n^{(k-1)})^{-1},
\end{align*}
In particular,
\begin{align*}
    \widetilde{\alpha} (1 + \delta)^{-1} (\widetilde{\Sigma}_n^{(k-1)})^{-1} \preceq \nabla^2 U_k(x) \preceq \alpha^{-1} (1 + \delta) (\widetilde{\Sigma}_n^{(k-1)})^{-1}, \quad x \in \R^d.
\end{align*}
If we take
\begin{align*}
    \widetilde{U}_k(x) := U_k((\widetilde{\Sigma}_n^{(k-1)})^{1/2} x),
\end{align*}
it therefore follows that
\begin{align}
\label{eq:good-covar-precond}
    \widetilde{\alpha} (1 + \delta)^{-1} I \preceq \nabla^2 \widetilde{U}_k(x) \preceq \alpha^{-1} (1 + \delta) I, \quad \forall x \in \R^d.
\end{align}
Hence, $\widetilde{\pi}_k(x) \propto \exp(-\widetilde{U}_k(x))$ is strongly log-concave with condition number at most $\frac{(1 + \delta)^2}{\alpha \widetilde{\alpha}}$.

Finally, suppose that $\gamma^k L \leq \mu$, which holds once $k \geq \log(\kappa) / \log(1/\gamma)$. By Corollary~\ref{cor:sandwich}, Assumption~\ref{asmp:conditioning}, and the fact that $\mu I \preceq \nabla^2 U(x)$, this choice of $k$ allows us to write
\begin{align*}
    (\Sigma^{(k)})^{-1} \preceq B + \gamma^k L I \preceq \alpha^{-1} \nabla^2 U(x) + \nabla^2 U(x)
    = \frac{1 + \alpha}{\alpha} \nabla^2 U(x)
\end{align*}
and
\begin{align*}
    (\Sigma^{(k)})^{-1} \succeq \alpha B + \gamma^k L I \succeq \alpha B \succeq \alpha \nabla^2 U(x).
\end{align*}
Combining these with the inequality $(1 + \delta)^{-1} \Sigma^{(k)} \preceq \widetilde{\Sigma}_n^{(k)} \preceq (1 + \delta) \Sigma^{(k)}$ gives
\begin{align*}
    \frac{\alpha}{(1 + \alpha)(1 + \delta)} (\widetilde{\Sigma}_n^{(k)})^{-1} \preceq \nabla^2 U(x)
    \preceq \frac{1 + \delta}{\alpha} (\widetilde{\Sigma}_n^{(k)})^{-1}
\end{align*}
for all $x \in \R^d$. In turn, $\widetilde{U}(x) := U((\widetilde{\Sigma}_n^{(k)})^{1/2} x)$ has condition number at most 
\begin{align*}
    \frac{(1 + \delta)^2 (1 + \alpha)}{\alpha^2} \leq \frac{2 (1 + \delta)^2}{\alpha^2},
\end{align*}
allowing us to sample from the target distribution with a good preconditioner.
\end{proof}

We now prove Theorem~\ref{th:gc-query-complexity-general}, distinguishing between low accuracy and high accuracy samplers as in Assumption~\ref{asmp:sampling-algorithm}.

\begin{proof}[Proof of Theorem~\ref{th:gc-query-complexity-general} for Low Accuracy $\mathcal{A}$]
In both the low accuracy and high accuracy settings, stage $k = 0$ remains the same. The target distribution is induced by the potential
\begin{align*}
    \widetilde{U}_0(x) := U_0(L^{-1/2} x).
\end{align*}
Because $\nabla^2 \widetilde{U}_0(x) = L^{-1} \nabla^2 U_0(L^{-1/2} x)$, it follows that
\begin{align*}
    I \preceq \nabla^2 \widetilde{U}_0(x) \preceq 2 I, \quad \forall x \in \R^d,
\end{align*}
so that the corresponding measure $\widetilde{\pi}^{(0)} \in \sP_2(\R^d)$ is $1$-strongly log-concave and $2$-smooth. By Theorem~\ref{th:sample-cov-concentration}, if we take $n = \Theta\left(\frac{d}{\delta^2}\right)$ approximate samples from $\widetilde{\pi}^{(0)}$, then with high probability, the resulting sample covariance $\widehat{\Sigma}_n^{(0)}$ satisfies
\begin{align*}
    (1 + \delta)^{-1} \covar(\widetilde{\pi}^{(0)}) \preceq \widehat{\Sigma}_n^{(0)} \preceq (1 + \delta) \covar(\widetilde{\pi}^{(0)}).
\end{align*}
In the low accuracy case, each of the samples can be produced with
\begin{align*}
    N = \widetilde{O}\left(d^{c_2} \delta^{-c_3}\right)
\end{align*}
first-order queries. Now, because
\begin{align*}
    \Sigma^{(0)} = L^{-1} \covar(\widetilde{\pi}^{(0)}),
\end{align*}
it follows that
\begin{align*}
    (1 + \delta)^{-1} \Sigma^{(0)} \preceq L^{-1} \widehat{\Sigma}_n^{(0)} \preceq (1 + \delta) \Sigma^{(0)}.
\end{align*}
Taking $\widetilde{\Sigma}_n^{(0)} := L^{-1} \widehat{\Sigma}_n^{(0)}$, it follows that $\widetilde{\Sigma}_n^{(0)}$ is a good preconditioner for $\pi_0$.

At stage $k$ of Gaussian cooling, assume that we have formed a sample covariance matrix $\widetilde{\Sigma}_{n}^{(k-1)}$ such that
\begin{align*}
    (1 + \delta)^{-1} \Sigma^{(k-1)} \preceq \widetilde{\Sigma}_{n}^{(k-1)} \preceq (1 + \delta) \Sigma^{(k-1)},
\end{align*}
allowing us to form a potential
\begin{align*}
    \widetilde{U}_k(x) := U_k((\widetilde{\Sigma}_{n}^{(k-1)})^{1/2} x)
\end{align*}
with corresponding measure $\widetilde{\pi}^{(k)} \in \sP_2(\R^d)$ that is $\widetilde{\alpha}(1 + \delta)^{-1}$-strongly log-concave and $\widetilde{\alpha}^{-1}(1 + \delta)$-smooth.

Note that if $\sR_2(\widehat{\pi} \| \pi) \leq \varepsilon^2$, then
\begin{align*}
    \chi^2(\widehat{\pi} \| \pi) \leq \exp(\varepsilon^2) - 1.
\end{align*}
Leveraging the guarantee in Theorem~\ref{th:sample-cov-concentration}, we take
\begin{align*}
    c(\delta) := \min\left\{\left(\sqrt{3 - (1 + \delta)^{-1/2}} - \sqrt{2}\right)^2, \frac{1}{16}\right\}
\end{align*}
In turn, let $\kappa = \frac{(1 + \delta)^2}{\widetilde{\alpha}^2}$, $L = \widetilde{\alpha}^{-1}(1 + \delta)$, and $\varepsilon = \sqrt{\log(c(\delta) + 1)}$. With these parameters, after
\begin{align*}
    N = \widetilde{\Theta}\left(\kappa^{c_1} d^{c_2}(\log(c(\delta) + 1))^{-c_3 / 2}\right) = \widetilde{\Theta}\left(\left(\frac{1 + \delta}{\widetilde{\alpha}}\right)^{2 c_1} d^{c_2} \delta^{-c_3}\right)
\end{align*}
first-order queries, Assumption~\ref{asmp:sampling-algorithm} ensures that the resulting samples $X_1, \ldots, X_n \overset{i.i.d.}{\sim}\widehat{\pi}_{N}^{(k)}$ are such that
\begin{align*}
    \chi^2(\widehat{\pi}_{N}^{(k)} \| \widetilde{\pi}^{(k)}) \leq c(\delta).
\end{align*}
As long as we take
\begin{align*}
    n = \Theta\left(\frac{d (1 + \delta)^4}{\widetilde{\alpha}^4 \delta^2}\right)
\end{align*}
samples, Theorem~\ref{th:sample-cov-concentration} then implies that, with high probability,
\begin{align}
\label{eq:high-prob-intermediate}
    (1 + \delta)^{-1} \covar(\widetilde{\pi}^{(k)}) \preceq \widehat{\Sigma}_n^{(k)} \preceq (1 + \delta) \covar(\widetilde{\pi}^{(k)}),
\end{align}
where $\widehat{\Sigma}_n^{(k)}$ is the sample covariance formed from $n$ samples drawn from $\widehat{\pi}_{N}^{(k)}$.
As
\begin{align*}
    \Sigma^{(k)} := \covar(\pi^{(k)}) = (\widetilde{\Sigma}_n^{(k-1)})^{1/2} \covar(\widetilde{\pi}^{(k)}) (\widetilde{\Sigma}_n^{(k-1)})^{1/2},
\end{align*}
it follows that
\begin{align*}
    (1 + \delta)^{-1}  \Sigma^{(k)} \preceq (\widetilde{\Sigma}_n^{(k-1)})^{1/2} \widehat{\Sigma}_n^{(k)} (\widetilde{\Sigma}_n^{(k-1)})^{1/2} \preceq (1 + \delta)  \Sigma^{(k)},
\end{align*}
allowing us to take $\widetilde{\Sigma}_n^{(k)} := (\widetilde{\Sigma}_n^{(k-1)})^{1/2} \widehat{\Sigma}_n^{(k)} (\widetilde{\Sigma}_n^{(k-1)})^{1/2}$ as a preconditioner for the next stage of cooling.

Finally, we perform a simple union bound over the probability that the concentration in Theorem~\ref{th:sample-cov-concentration} fails at each stage of Gaussian cooling. Specifically, across $K = O\left(\frac{\log(\kappa)}{\log(1 / \gamma)}\right)
$ stages of cooling, there exists constants $C_1, C_2 > 0$ such that as long as we take
\begin{align*}
    n = \frac{C_1 d(1 + \delta)^4}{\widetilde{\alpha}^4 \delta^2}
\end{align*}
samples at each stage, then the probability that our concentration inequality fails at each stage is at most $4 \exp\left(\frac{-C_2 \widetilde{\alpha}^4 \delta^2 n}{(1 + \delta)^4}\right)$.
By a union bound, the total probability of failure is thus at most
\begin{align*}
    4 C_3 \exp\left(-C_2 \frac{\widetilde{\alpha}^4 \delta^2 n}{(1 + \delta)^4}\right) \cdot \frac{\log(\kappa)}{\log(1 / \gamma)}
\end{align*}
for some constant $C_3 > 0$. Rearranging to find the dependence of $n$ on the remaining parameters, we find that in order for the above quantity to be at most $p$ for some arbitrary $p \in (0, 1)$, we must take
\begin{align*}
    n = \Theta\left(\frac{(1 + \delta)^4 \left[d + \log(K / p)\right]}{\widetilde{\alpha}^4 \delta^2}\right).
\end{align*}
samples at each stage. In turn, taking $n$ as above will ensure that Gaussian cooling succeeds with probability at least $1 - p$. This validates the query complexity reported above, as we have only incurred an additional sublogarithmic factor in $\kappa$ to ensure that the algorithm succeeds with high probability.
\end{proof}

In the high accuracy case, the above proof only needs to be modified slightly. Namely, the dependence on $\delta$ improves slightly due to the polylogarithmic dependence of the query complexity of a high accuracy sampler on the target accuracy. However, we must perform an additional union bound to account for the nonzero probability that the number of first-order queries at each stage exceeds the threshold in Assumption~\ref{asmp:sampling-algorithm}. We remark that this additional consideration is specifically due to the fact that FORS relies on rejection sampling.

\begin{proof}[Proof of Theorem~\ref{th:gc-query-complexity-general} for High Accuracy $\mathcal{A}$]
With $c(\delta)$, $\kappa$, $L$, and $\varepsilon$ as in the low accuracy case, the high accuracy setting in Assumption~\ref{asmp:sampling-algorithm} means that a single call to $\mathcal{A}$ produces a sample from a distribution $\widehat{\pi}^{(k)}$ with $\chi^2(\widehat{\pi}^{(k)} \| \widetilde{\pi}^{(k)}) \leq c(\delta)$, and that, with probability at least $1 - \eta$, the call uses at most
\begin{align*}
    N := N_0 + O(\log(1/\eta)), \qquad
    N_0 := \widetilde{\Theta}\left(\left(\frac{1 + \delta}{\widetilde{\alpha}}\right)^{2 c_1} d^{c_2}
    (\log(\delta^{-1}))^{c_3}\right)
\end{align*}
first-order queries. Across all $K$ stages of Gaussian cooling, the sampler $\mathcal{A}$ is called $nK$ times. Performing a union bound, the probability that any call uses more than $N$ queries is at most $nK\eta$. We therefore choose
\begin{align*}
    \eta := \frac{p}{2nK},
\end{align*}
so that this probability is at most $p/2$. We can then directly compute, due to our choice of $n$,
\begin{align*}
    \log(1/\eta) = \log\left(\frac{2nK}{p}\right) = O\left(\log d + \log(\widetilde{\alpha}^{-1}) + \log(\delta^{-1}) + \log K + \log(1/p)\right),
\end{align*}
which is polylogarithmic in all problem parameters. Now, the concentration argument and the union bound over the $K$ stages of cooling from the low accuracy case apply verbatim, with $p/2$ in place of $p$. In particular, taking
\begin{align*}
    n = \Theta\left(\frac{(1 + \delta)^4 \left[d + \log(K / p)\right]}{\widetilde{\alpha}^4 \delta^2}\right)
\end{align*}
samples per stage ensures that, with probability at least $1 - p/2$, the preconditioning guarantee holds at every stage. By a union bound over these two possible options for failure (i.e., the number of queries exceeds the desired threshold or concentration of the sample covariance fails), with probability at least $1 - p$ the preconditioning guarantee holds at every stage and every call to $\mathcal{A}$ uses at most $N$ first-order queries, for a total of at most
\begin{align*}
    nKN = nK N_0 + O\left(nK \log(nK/p)\right)
\end{align*}
first-order queries.
\end{proof}

In the translation-invariant case, instead suppose that we have an estimate $\widehat{\sigma}$ of $\sigma$ such that
\begin{align}
\label{eq:fourier-space-rel-preconditioning}
    \|\diag(\sigma)^{-1/2} \diag(\widehat{\sigma}) \diag(\sigma)^{-1/2} - I\|_2 \leq C(\delta),
\end{align}
where $C(\delta) := 1 - (1 + \delta)^{-1/2}$. Because $\sF$ is unitary, assuming that it is scaled appropriately, we have that
\begin{align*}
    \|\sF \diag(\sigma)^{-1/2} \sF^* \sF \diag(\widehat{\sigma}) \sF^* \sF \diag(\sigma)^{-1/2}  \sF^* - I\|_2 \leq C(\delta),
\end{align*}
which holds if and only if
\begin{align*}
    \|\Sigma^{-1/2} \widehat{\Sigma} \Sigma^{-1/2} - I\|_2 \leq C(\delta),
\end{align*}
where
\begin{align*}
    \widehat{\Sigma} = \sF \diag(\widehat{\sigma}) \sF^*.
\end{align*}
As seen in Appendix~\ref{app:tech-proofs}, this is sufficient to conclude that
\begin{align*}
    (1 + \delta)^{-1/2} \Sigma \preceq \widehat{\Sigma} \preceq (1 + \delta)^{1/2} \Sigma,
\end{align*}
which is precisely the guarantee that we require for our analysis in Appendix~\ref{app:tech-proofs}. Moreover, for Equation~\ref{eq:fourier-space-rel-preconditioning} to hold, it suffices to show that
\begin{align}
\label{eq:fourier-vector-concentration}
    \|\widehat{\sigma} \odot \sigma^{-1} - \mathbf{1}\|_\infty \leq C(\delta),
\end{align}
where $v \odot w \in \R^d$ denotes the Hadamard (entrywise) product of two vectors $v, w \in \R^d$ and $\mathbf{1} \in \R^d$ is the vector of all ones.

Given a set of samples $x_1, \ldots, x_n \overset{i.i.d.}{\sim} \mu$, where $\mu$ is a distribution whose covariance $\Sigma = \covar(\mu)$ we wish to estimate, the natural empirical counterpart to $\sigma$ is
\begin{align}
\label{eq:emp-fourier-sigma}
    \widehat{\sigma} := \frac{1}{n} \sum_{i=1}^n \sF(x_i - \widehat{m}_n) \odot \overline{\sF(x_i - \widehat{m}_n)},
\end{align}
where the complex conjugate on the second term is taken entrywise.
To see where this estimate comes from, note that $\sF^* \Sigma \sF = \diag(\sigma)$. If we similarly conjugate by $\sF$ the (biased) sample covariance $\widehat{\Sigma}_n$ formed from the samples $x_1, \ldots, x_n$, we obtain
\begin{align*}
    \sF^* \widehat{\Sigma}_n \sF &= \frac{1}{n} \sum_{i=1}^n \sF^* (x_i - \widehat{m}_n) (x_i - \widehat{m}_n)^\top \sF \\
    &= \frac{1}{n} \sum_{i=1}^n ((x_i - \widehat{m}_n)^\top \sF)^* (x_i - \widehat{m}_n)^\top \sF.
\end{align*}
Extracting the diagonal from the above expression, our estimate $\widehat{\sigma}$ of $\sigma$ is thus the vector with entries
\begin{align*}
    \widehat{\sigma}_j = \frac{1}{n} \sum_{i=1}^n \overline{(x_i - \widehat{m}_n)^\top \sF e_j}(x_i - \widehat{m}_n)^\top \sF e_j, \quad \forall j \in [d].
\end{align*}
Finally, because the discrete Fourier transform is symmetric, we have that
\begin{align*}
    (x_i - \widehat{m}_n)^\top \sF e_j = (\sF(x_i - \widehat{m}_n))_j,
\end{align*}
which yields the estimate proposed in Equation~\ref{eq:emp-fourier-sigma}. In Appendix~\ref{app:tech-proofs}, we analyze the number of samples required in Equation~\ref{eq:emp-fourier-sigma} to ensure that Equation~\ref{eq:fourier-vector-concentration} holds with high probability, after which the  analysis in both the low and high accuracy cases can be repeated. 

\section{Approximating the Sample Covariance}
\label{app:tech-proofs}
In the proofs below, we will require the following standard definitions pertaining to subgaussian and subexponential distributions, taken from \cite[Chapter 2]{Vershynin2018}.

\begin{definition}
\label{def:subgaussian}
A real-valued random variable $X$ is \textit{subgaussian} if there exists $K_1 > 0$ such that
\begin{align*}
    \P(|X| \geq t) \leq 2 \exp(-t^2 / K_1^2), \quad t \geq 0.
\end{align*}
Equivalently, $X$ is subgaussian if there exists $K_2 > 0$ such that
\begin{align*}
    \E[\exp(X^2 / K_2^2)] \leq 2,
\end{align*}
and we define the corresponding subgaussian norm by
\begin{align*}
    \|X\|_{\psi_2} := \inf\{K > 0 : \E[\exp(X^2 / K^2)] \leq 2\}.
\end{align*}
\end{definition}

The definition of a \textit{subexponential} random variable is entirely analogous, so we only mention the relevant subexponential norm, given by
\begin{align*}
    \|X\|_{\psi_1} := \inf\{K > 0 : \E[\exp(|X| / K)] \leq 2\}.
\end{align*}
The proof of our main result relies on three technical tools. First, we leverage the variational characterization of $\chi^2$-divergence provided in Theorem~\ref{thm:chi2-variational} to describe the relative distance, in spectral norm, between the covariance matrices of two distributions that are close in $\chi^2$-divergence. Then, we show that closeness in $\chi^2$-divergence preserves subgaussianity of log-concave measures. Finally, subgaussianity of the sampled distribution will allow us to show that the sample covariance concentrates well, providing a good preconditioner for the next stage of Gaussian cooling. At a high level, the primary technical contribution of this section is an answer to the following question.
\begin{question}
\label{q:successive-precond}
Suppose that $\{x_1, \ldots, x_n\}$ is a set of approximate samples from a log-concave measure $\mu$, in the sense that $x_i \overset{i.i.d}{\sim} \widehat{\mu}$, where
\begin{align*}
    \mathcal{D}(\widehat{\mu} \| \mu) \leq \varepsilon
\end{align*}
for some divergence $\mathcal{D}$ and all $i \in [n]$. If $\delta > 0$ is fixed, then how many samples do we need to ensure that the empirical covariance $\widehat{\Sigma}_n$ of $\{x_1, \ldots, x_n\}$ satisfies
\begin{align} 
\label{eq:sigma-hat-bound}
(1 + \delta)^{-1} \Sigma \preceq \widehat{\Sigma}_n \preceq (1 + \delta) \Sigma,
\end{align}
where $\Sigma$ is the covariance of $\mu$?
\end{question}

By answering Question~\ref{q:successive-precond}, we guarantee that stage $k$ of Gaussian cooling provides a good preconditioner for stage $k + 1$, as specified by the parameter $\delta > 0$.

\subsection{Preconditioning Under Closeness in $\chi^2$-Divergence}
Perhaps the most natural characterization of $\chi^2$-divergence for our purposes is its variational formulation, elucidated in Appendix~\ref{app:tech-assumps}. For completeness, we recall it here:
\begin{align}
\label{eq:chi-2-var}
\chi^2(\mu \| \nu) = \sup_{g : \R^d \to \R} \frac{(\E_{\mu}[g(X)] - \E_{\nu}[g(Y)])^2}{\Var_{\nu}(g(Y))},
\end{align}
where the supremum is taken over all functions $g: \R^d \to \R$ such that $g \in L^2(\nu)$ and $\Var_{Y \sim \nu}(g(Y)) > 0$. Note that for distributions with all moments finite, such as subgaussian distributions, this class of functions includes all polynomial functions of the form $g(x) = \langle v, x - w\rangle^p$, where $p \in \N$ and $v,w  \in \R^d$ are arbitrary vectors.

To begin, one may attempt to take $g(x) = \langle v, x \rangle$ for some arbitrary unit vector $v \in \R^d$, which yields
\begin{align*}
    \frac{(\E_{\mu}[\langle v, X \rangle] - \E_{\nu}[\langle v, Y \rangle])^2}{\Var_{\nu}(\langle v, Y \rangle)} = \frac{(\langle v, \E[X] - \E[Y] \rangle)^2}{\langle v, \Sigma_2 v \rangle}.
\end{align*}
Hence, if $\chi^2(\mu \| \nu) < \varepsilon$ for some $\varepsilon > 0$, then
\begin{align}
\label{eq:linear-test-bound}
    |\langle v, \E[X] - \E[Y] \rangle| < \sqrt{\varepsilon \langle v, \Sigma_2 v \rangle}.
\end{align}
Although this doesn't give immediate control over the difference between the covariance matrices of $\mu$ and $\nu$ respectively, Equation~\ref{eq:linear-test-bound} will prove to be useful later.

If we instead take $m_1 := \E_{\mu}[X]$ and $\Sigma_1 = \covar_{\mu}(X)$, and similarly let $m_2 := \E_{\nu}[X]$ and $\Sigma_2 := \covar_{\nu}(X)$, we may instead consider the test function
\begin{align*}
    g_v(x) := \langle v, x - m_2\rangle^2.
\end{align*}
Then, we have that
\begin{align*}
    \E_{\mu}[g_{v}(X)] &= \E_{\mu}[\langle v, X - m_2 \rangle^2] \\
    &= \E_{\mu}[(\langle v, X - m_1 \rangle + \langle v, m_1 - m_2 \rangle)^2] \\
    &= \E_{\mu}[\langle v, X - m_1 \rangle^2] + \langle v, m_1 - m_2 \rangle^2 + 2 \langle v, m_1 - m_2 \rangle \E_{\mu} [\langle v, X - m_1 \rangle] \\
    &= \langle v, \Sigma_1 v \rangle + \langle v, m_1 - m_2 \rangle^2.
\end{align*}
On the other hand,
\begin{align*}
    \E_{\nu}[g_{v}(Y)] &= \E_{\nu}[\langle v, Y - m_2 \rangle^2] = \langle v, \Sigma_2 v \rangle.
\end{align*}
Hence, the numerator of the above expression becomes
\begin{align}
\label{eq:chi-2-num}
    (\E_{\mu}[g(X)] - \E_{\nu}[g(Y)])^2 = (\langle v, \Sigma_1 v \rangle - \langle v, \Sigma_2 v \rangle + \langle v, m_1 - m_2 \rangle^2)^2
\end{align}
On the other hand, we compute
\begin{align}
\label{eq:chi-2-denom}
    \Var_{\nu}(g(Y)) &= \E_{\nu}[\langle v, Y - m_2 \rangle^4] - (\E_{\nu}[\langle v, Y - m_2 \rangle^2])^2 \notag \\
    &= \E_{\nu}[\langle v, Y - m_2 \rangle^4] - \langle v, \Sigma_2 v \rangle^2.
\end{align}
To motivate the computation that is to come, observe that if $\nu$ is the law of a centered Gaussian $Y \sim \sN(0, I_d)$, then
\begin{align*}
    \E_{\nu}[\langle v, Y - m_2 \rangle^4] = \E_{g_i \sim \sN(0, 1)}\left[\left( \sum_{i=1}^d v_i g_i \right)^4 \right] &= \E_{g_i \sim \sN(0, 1)}\left[\sum_{1 \leq i \neq j \leq d} v_i^2 v_j^2 g_i^2 g_j^2 + \sum_{i=1}^d v_i^4 g_i^4\right] \\
    &= 6 \sum_{1 \leq i < j \leq d} v_i^2 v_j^2 + 3 \sum_{i=1}^d v_i^4 \\
    &= 3 \left(\sum_{i=1}^d v_i^2 \right)^2 \\
    &= 3\E_{g_i \sim \sN(0, 1)} \left[\left(\sum_{i=1}^d v_i g_i \right)^2\right]^2 \\
    &= 3 \E_{\nu}[\langle v, Y - m_2 \rangle^2]^2 \\
    &= 3 \langle v, \Sigma_2 v \rangle^2.
\end{align*}
Motivated by this simple example, we anticipate that if $\nu$ is subgaussian, then we can control the fourth moments of their one-dimensional marginals by the corresponding second moments.
\subsection{$L^4$--$L^2$ Equivalence for Log-Concave Measures}
Several recent results concerning one-dimensional moment comparison bounds come to the rescue \cite{Eitan2024,Murawski2026}. To reduce to the one-dimensional case, observe that if $Y$ is a centered log-concave random vector, then $\langle v, Y \rangle$ is a centered log-concave random variable for any $v \in \R^d$. Indeed, if $U \in \R^{d \times d}$ is an orthogonal matrix such that $U v = e_1$, where $e_1$ is the first standard basis vector in $\R^d$, then $\langle v, Y \rangle = \langle e_1, U Y \rangle$. Because log-concavity is preserved under affine transformations and marginalization, the above claim follows. 

Hence, to obtain the $L^4$--$L^2$ norm equivalence that we require to continue our analysis in $\chi^2$-divergence, it suffices to find the optimal constant $C > 0$ such that
\begin{align}
\label{eq:1d-l4-l2}
    \E[X^4] \leq C \E[X^2]^2
\end{align}
for all centered, real-valued log-concave distributions. As one might suspect, the optimal constant in Equation~\ref{eq:1d-l4-l2} is attained by a centered exponential distribution. Precisely, we have the following result from \cite{Eitan2024}.
\begin{proposition}[Theorem 4 in \cite{Eitan2024}]
\label{prop:eitan}
Let $X$ be a centered, real-valued log-concave random variable. Then, for all $q \in \N$, we have that
\begin{align*}
    \frac{\|X\|_{2q}}{\|X\|_2} \leq (!(2q))^{1/(2q)},
\end{align*}
with equality if and only if $X$ has the law of $\pm \Gamma$, where $!n$ is the number of derangements of $n$ and $\Gamma$ has density $f(x) = e^{-(x + 1)}$ for $x \geq -1$.
\end{proposition}
Above, another formula for the number of derangements of $n$ (i.e., permutations with no fixed points) is
\begin{align*}
    !n = n! \sum_{k=0}^n \frac{(-1)^k}{k!},
\end{align*}
so for $2q = 4$, we have that $!(2q) = 9$. Hence, Proposition~\ref{prop:eitan} implies that for any centered, real-valued log-concave random variable $X$,
\begin{align*}
    \E[X^4] = \|X\|_4^4 \leq 9 \|X\|_2^4 = 9 \E[X^2]^2,
\end{align*}
and that the constant $C = 9$ is attained by the centered exponential distribution. Returning to the multivariate setting, we immediately obtain the following corollary.
\begin{corollary}
\label{cor:log-concave-l4-l2}
Suppose that $\mu$ is a log-concave distribution on $\R^d$. Then, for any $v \in \R^d$, we have that
\begin{align*}
    \E[\langle v, X - \E[X] \rangle^4] \leq 9 \langle v, \Sigma v \rangle^2.
\end{align*}
\vspace{-0.6cm}
\end{corollary}
This follows immediately upon applying the following result from \cite{Stavrakakis2013} to $Y = X - \E[X]$, whose law is still log-concave under $\mu$, and considering the seminorm $f_v(x) = |\langle v, x\rangle|$.

\begin{proposition}
\label{prop:log-concave-moment-comparison}
Suppose that $\mu$ is a log-concave probability measure on $\R^d$ and that $f : \R^d \to \R_+$ is a seminorm. Then, for any $1 \leq p < q$, we have that
\begin{align*}
    \|f\|_{L^q(\mu)} \leq \frac{c q}{p} \|f\|_{L^p(\mu)},
\end{align*}
where $c > 0$ is a universal constant.
\end{proposition}

By restricting our attention to seminorms of the form $f_{v}(x) = |\langle v, x\rangle|$, where $v \in \R^d$, the $L^4$--$L^2$ comparison bound that we require for our $\chi^2$-divergence analysis becomes apparent.

Now, from the variational characterization of $\chi^2$-divergence in Equation~\ref{eq:chi-2-var}, if $\chi^2(\mu \| \nu) < \eps$, then Equation~\ref{eq:chi-2-num} and Equation~\ref{eq:chi-2-denom} combine into
\begin{align*}
    (\langle v, \Sigma_1 v \rangle - \langle v, \Sigma_2 v \rangle + \langle v, m_1 - m_2 \rangle^2)^2 < \eps (\E_{\nu}[\langle v, Y - m_2 \rangle^4] - \langle v, \Sigma_2 v \rangle^2).
\end{align*}
If $\nu$ is also log-concave (as will be the case for us), then it follows that
\begin{align*}
    (\langle v, \Sigma_1 v \rangle - \langle v, \Sigma_2 v \rangle + \langle v, m_1 - m_2 \rangle^2)^2 \leq 8 \eps \langle v, \Sigma_2 v \rangle^2,
\end{align*}
applying Corollary~\ref{cor:log-concave-l4-l2}
Taking square roots of both sides,
\begin{align*}
    \langle v, \Sigma_1 v \rangle - \langle v, \Sigma_2 v \rangle \leq \langle v, \Sigma_1 v \rangle - \langle v, \Sigma_2 v \rangle + \langle v, m_1 - m_2 \rangle^2 \leq \sqrt{8 \eps}  \langle v, \Sigma_2 v \rangle,
\end{align*}
and similarly, that
\begin{align*}
    \langle v, \Sigma_1 v \rangle - \langle v, \Sigma_2 v \rangle + \langle v, m_1 - m_2 \rangle^2 \geq - \sqrt{8 \eps} \langle v, \Sigma_2 v \rangle.
\end{align*}
Additionally, we already showed that in Equation~\ref{eq:linear-test-bound} that if $\chi^2(\mu \| \nu) < \eps$, then $\langle v, m_1 - m_2 \rangle^2 < \eps \langle v, \Sigma_2 v \rangle.$
Hence, we have that
\begin{align*}
    \langle v, \Sigma_1 v \rangle - \langle v, \Sigma_2 v \rangle \geq -\left(\sqrt{8 \eps} + \eps\right)\langle v, \Sigma_2 v \rangle.
\end{align*}
To summarize:
\begin{lemma}
\label{lem:chi2-to-precond-1}
If $\mu, \nu \in \sP_2(\R^d)$ are two distributions such that $\nu$ is log-concave and $\chi^2(\mu \| \nu) < \eps$ for some $\eps > 0$, then
\begin{align*}
    |\langle v, \Sigma_1 v \rangle - \langle v, \Sigma_2 v \rangle| \leq \left(\sqrt{8 \eps} + \eps\right) \langle v, \Sigma_2 v \rangle,
\end{align*}
where $v \in \R^d$ is an arbitrary unit vector.
\end{lemma}
The following corollary is immediate upon carrying out some straightforward algebraic manipulations.
\begin{corollary}
\label{cor:chi2-to-precond-1}
Suppose that $\mu, \widehat{\mu} \in \sP_2(\R^d)$ are such that $\covar(\mu) = \Sigma$ and $\covar(\hat{\mu}) = \widehat{\Sigma}$, with $\mu$ log-concave. If
\begin{align*}
    \chi^2(\widehat{\mu} \| \mu) < \left(\sqrt{3 - (1 + \delta)^{-1/2}} - \sqrt{2}\right)^{2}
\end{align*}
for some $\delta > 0$, then
\begin{align*}
    (1 + \delta)^{-1/2} \Sigma \preceq \widehat{\Sigma} \preceq (1 + \delta)^{1/2} \Sigma.
\end{align*}
\end{corollary}
For example, for $\delta = 1$, we require that $\chi^2(\widehat{\mu} \| \mu) \approx 0.01$.
For a more precise characterization of the first-order contribution, note that because $\sqrt{1 + x} = 1 + \frac{x}{2} + O(x^2)$ for $|x| < 1$, we have that
\begin{align*}
    \left(\sqrt{3 - (1 + \delta)^{-1/2}} - \sqrt{2}\right)^{2} = O(\delta^2)
\end{align*}
for $|\delta| < 1$. Hence, for small $\delta > 0$, we only require that $\chi^2(\widehat{\mu} \| \mu) = O(\delta^2)$ in order to obtain the desired preconditioning guarantee. 

\subsection{Preconditioning with the Sample Covariance}
At a high level, we require a strong matrix concentration bound to show that the sample covariance is a viable preconditioner at each stage of Gaussian cooling, and closeness in $\chi^2$-divergence preserves subgaussianity, a sufficient property to apply such a concentration inequality. The argument in the next section then proceeds in two steps:
\begin{enumerate}
\item[(1)] If the target distribution $\mu$ at step $k$ is $\alpha$-strongly log-concave and $(1 / \alpha)$-smooth, then the approximate distribution, furnished by the user's sampler of choice after sufficiently many steps to satisfy the hypothesis of Corollary~\ref{cor:chi2-to-precond-1}, will be subgaussian, with subgaussian norm that depends on $\alpha$.
\item[(2)] Having characterized the subgaussianity of the sampled distribution, we can apply an appropriate matrix concentration inequality for matrices with independent, centered, isotropic, and subgaussian rows to the sample covariance formed by i.i.d. samples, produced using parallel chains.
\end{enumerate}
\subsubsection{Subgaussianity of Sampled Distribution}
As described in the preceding section, it will be to our advantage to first quantitatively characterize the subgaussianity of a distribution $\widehat{\mu}$ such that $\chi^2(\widehat{\mu} \| \mu) < \eps$, where $\mu$ is $\alpha$-strongly log-concave and $(1 / \alpha)$-smooth. 
\begin{lemma}
\label{lem:chi-2-to-sg}
Suppose that $\mu \in \sP_2(\R^d)$ is an $\alpha$-strongly log-concave distribution and $\widehat{\mu} \in \sP_2(\R^d)$ is such that $\chi^2(\widehat{\mu} \| \mu) < \frac{1}{2}$. Then, for all unit vectors $v \in \R^d$,
\begin{align*}
    \P_{\widehat{\mu}}(|\langle v, Y - \E_{\widehat{\mu}}[Y] \rangle| \geq t) \leq 2 \exp\left(\frac{-\alpha t^2}{8}\right)
\end{align*}
for all $t \geq 0$ so that $Y \sim \widehat{\mu}$ is $(\sqrt{8 / \alpha})$-subgaussian.
\end{lemma}
\begin{proof}
If $\mu$ is such a distribution, then for all unit vectors $v \in \R^d$,
\begin{align*}
    \P_{\mu}(|\langle v, X - \E_{\mu}[X] \rangle| \geq t) \leq 2 \exp\left(-\frac{\alpha t^2}{2}\right),
\end{align*}
which aligns with the definition of $(\sqrt{2 / \alpha})$-subgaussianity for an $\R^d$-valued random vector provided in \cite[Section 3.4]{Vershynin2018}. 
Now, take $Y \sim \widehat{\mu}$ such that $\chi^2(\widehat{\mu} \| \mu) < \eps$ for some $\eps \in (0, 1 / 2)$. Fixing a unit vector $v \in \R^d$, we compute
\begin{align*}
    \P_{\widehat{\mu}}(|\langle v, Y - \E_{\widehat{\mu}}[Y] \rangle| \geq t) &= \int_{|\langle v, y - \E_{\widehat{\mu}}[Y] \rangle| \geq t} d\widehat{\mu}(y) \\
    &= \int_{|\langle v, y - \E_{\widehat{\mu}}[Y] \rangle| \geq t} \frac{d\widehat{\mu}}{d\mu}(y) d\mu(y) \\
    &\leq \left(\int_{|\langle v, y - \E_{\widehat{\mu}}[Y] \rangle| \geq t} \left(\frac{d\widehat{\mu}}{d\mu}(y)\right)^{2} d\mu(y)\right)^{1/2} \left(\int_{|\langle v, y - \E_{\widehat{\mu}}[Y] \rangle| \geq t} d\mu(y)\right)^{1/2} \\
    &= \left(\int_{|\langle v, y - \E_{\widehat{\mu}}[Y] \rangle| \geq t} \left(\frac{d\widehat{\mu}}{d\mu}(y)\right)^{2} d\mu(y)\right)^{1/2} \left(\P_{\mu}(|\langle v, X - \E_{\widehat{\mu}}[Y] \rangle| \geq t)\right)^{1/2}.
\end{align*}
We can now analyze both terms separately, relating the first term to the fact that $\chi^2(\widehat{\mu} \| \mu) < \eps$ and bounding the second term using the subgaussianity of $\mu$. For the first term, we write
\begin{align*}
    \int_{|\langle v, y - \E_{\widehat{\mu}}[Y] \rangle| \geq t} \left(\frac{d\widehat{\mu}}{d\mu}(y)\right)^{2} d\mu(y) \leq \int_{\R^d} \left(\frac{d\widehat{\mu}}{d\mu}(y)\right)^{2} d\mu(y) = \chi^2(\widehat{\mu} \| \mu) + 1 < \frac{3}{2},
\end{align*}
as long as $\eps \in (0, 1 / 2)$.
On the other hand, because
\begin{align*}
    |\langle v, X - \E_{\widehat{\mu}}[Y] \rangle| &\leq |\langle v, X - \E_{\mu}[X] \rangle| + |\langle v, \E_{\mu}[X] - \E_{\widehat{\mu}}[Y] \rangle| \\
    &\leq |\langle v, X - \E_{\mu}[X] \rangle| + \|\E_{\mu}[X] - \E_{\widehat{\mu}}[Y]\|_2
\end{align*}
we have that
\begin{align*}
    \P_{\mu}(|\langle v, X - \E_{\widehat{\mu}}[Y] \rangle| \geq t) &\leq \P_{\mu}(|\langle v, X - \E_{\mu}[X] \rangle| + \|\E_{\mu}[X] - \E_{\widehat{\mu}}[Y]\|_2 \geq t) \\
    &=  \P_{\mu}(|\langle v, X - \E_{\mu}[X] \rangle| \geq t - \|\E_{\mu}[X] - \E_{\widehat{\mu}}[Y]\|_2).
\end{align*}
Because $\chi^2(\widehat{\mu} \| \mu) < \eps$, we have that
\begin{align*}
    \|\E_{\mu}[X] - \E_{\widehat{\mu}}[Y]\|_2 \leq \sqrt{\eps \lambda_{\max}(\Sigma_{\mu})} < \frac{1}{2} \sqrt{2 \lambda_{\max}(\Sigma_{\mu})} \leq \frac{\sqrt{2}}{2}\alpha^{-1/2},
\end{align*}
taking the supremum over all unit vectors $v \in \R^d$ in Equation~\ref{eq:linear-test-bound}. We also apply the fact that $\mu$ is $\alpha$-strongly log-concave so that $\lambda_{\max}(\Sigma_{\mu}) \leq 1 / \alpha$. Thus, by subgaussianity of $\mu$, for $t > \frac{\sqrt{2}}{2}\alpha^{-1/2}$, we have that
\begin{align*}
    \P_{\mu}(|\langle v, X - \E_{\widehat{\mu}}[Y] \rangle| \geq t) &\leq \P_{\mu}\left(|\langle v, X - \E_{\mu}[X] \rangle| \geq t - \frac{\sqrt{2}}{2}\alpha^{-1/2}\right) \\
    &\leq 2 \exp\left(-\frac{\alpha (t - \frac{\sqrt{2}}{2}\alpha^{-1/2})^2}{2}\right).
\end{align*}
Now, we apply the fact that $(t - c)^2 \geq \frac{1}{2}t^2 - c^2$ for all $t, c \in \R$ to obtain a bound of 
\begin{align*}
    2 \exp\left(-\frac{\alpha (\frac{1}{2}t^2 - \frac{1}{2}\alpha^{-1})}{2}\right) = 2 \exp(1/4) \exp\left(-\frac{\alpha t^2}{4}\right).
\end{align*}
Moreover, if $t \leq \frac{\sqrt{2}}{2}\alpha^{-1/2}$, then we can simply bound
\begin{align*}
    \P_{\mu}(|\langle v, X - \E_{\widehat{\mu}}[Y] \rangle| \geq t) \leq 1 \leq 2 \exp(1/4) \exp\left(-\frac{\alpha t^2}{4}\right),
\end{align*}
so the above bound holds for all $t > 0$. Returning to the original expression that we sought to bound to establish the subgaussianity of $Y \sim \widehat{\mu}$, we have that
\begin{align*}
    \P_{\widehat{\mu}}(|\langle v, Y - \E_{\widehat{\mu}}[Y] \rangle| \geq t) \leq \sqrt{3} \exp(1 / 8) \exp\left(\frac{-\alpha t^2}{8}\right) \leq 2 \exp\left(\frac{-\alpha t^2}{8}\right)
\end{align*}
for all $t \geq 0$.
\end{proof}
Crucially, we can now safely assume that the sampled distribution $\widehat{\mu}$ is subgaussian, paving the way for good concentration of the sample covariance.
\subsubsection{Concentration of Sample Covariance}
Consider taking $n$ independent, identically distributed samples $X_1, \ldots, X_n \sim \widehat{\mu}$, where we use $\widehat{\mu}$ to indicate that we are drawing samples from parallel ULMC chains. From these samples, we form the empirical mean
\begin{align*}
    \widehat{m}_n := \frac{1}{n} \sum_{i=1}^n X_i
\end{align*}
and the empirical covariance
\begin{align*}
    \widehat{\Sigma}_n := \frac{1}{n} \sum_{i=1}^n (X_i - \widehat{m}_n) (X_i - \widehat{m}_n)^\top.
\end{align*}
With this normalization, $\widehat{\Sigma}_n$ provides a \textit{biased} estimate of the population covariance
\begin{align*}
    \widehat{\Sigma}  = \E_{\widehat{\mu}}[(X - \E_{\widehat{\mu}}[X])(X - \E_{\widehat{\mu}}[X])^\top].
\end{align*}
However, we only seek a guarantee that the sample covariance satisfies
\begin{align}
\label{eq:emp-cov-des}
    (1 + \delta)^{-1/2} \widehat{\Sigma} \preceq \widehat{\Sigma}_n \preceq (1 + \delta)^{1/2} \widehat{\Sigma},
\end{align}
and the biased estimate above eases the analysis slightly. If Equation~\ref{eq:emp-cov-des} holds, and 
\begin{align*}
    (1 + \delta)^{-1/2} \Sigma \preceq \widehat{\Sigma} \preceq (1 + \delta)^{1/2} \Sigma
\end{align*}
holds for the same constant $\delta > 0$, we then have that
\begin{align*}
    (1 + \delta)^{-1} \Sigma \preceq (1 + \delta)^{-1/2} \widehat{\Sigma} \preceq \widehat{\Sigma}_n \preceq (1 + \delta)^{1/2} \widehat{\Sigma} \preceq (1 + \delta) \Sigma,
\end{align*}
allowing us to use the sample covariance as a preconditioner at step $k$ of Gaussian cooling. The biased estimate of the sample covariance turns out to suffice for these purposes, with cleaner analysis. Moreover, observe that Equation~\ref{eq:emp-cov-des} is equivalent to showing that
\begin{align*}
   [(1 + \delta)^{-1/2} - 1] \widehat{\Sigma} \preceq \widehat{\Sigma}_n - \widehat{\Sigma} \preceq [(1 + \delta)^{1/2} - 1] \widehat{\Sigma}.
\end{align*}
Equivalently, we have that
\begin{align*}
    [(1 + \delta)^{-1/2} - 1] I \preceq \widehat{\Sigma}^{-1/2} \widehat{\Sigma}_n\widehat{\Sigma}^{-1/2} - I \preceq [(1 + \delta)^{1/2} - 1] I.
\end{align*}
In turn, this holds true if
\begin{align}
\label{eq:2-norm-precond}
    \|\widehat{\Sigma}^{-1/2} \widehat{\Sigma}_n\widehat{\Sigma}^{-1/2} - I\|_2 \leq 1 - (1 + \delta)^{-1/2}.
\end{align}
We now aim to quantify how many samples are needed in order for Equation~\ref{eq:2-norm-precond} to hold with high probability, to be made precise below.

To quantify the sample complexity of obtaining Equation~\ref{eq:2-norm-precond}, we begin by decomposing $\widehat{\Sigma}_n$ as
\begin{align*}
    \widehat{\Sigma}_n &= \frac{1}{n} \sum_{i=1}^n (X_i - \widehat{m} + \widehat{m} - \widehat{m}_n) (X_i - \widehat{m} + \widehat{m} - \widehat{m}_n)^\top \\
    &= \frac{1}{n} \sum_{i=1}^n (X_i - \widehat{m}) (X_i - \widehat{m})^\top - (\widehat{m}_n - \widehat{m})(\widehat{m}_n - \widehat{m})^\top,
\end{align*}
where $\widehat{m} := \E_{\widehat{\mu}}[X]$. Hence, we have that
\begin{align}
\label{eq:bv-decomp}
    \widehat{\Sigma}^{-1/2} \widehat{\Sigma}_n\widehat{\Sigma}^{-1/2} &= \frac{1}{n} \sum_{i=1}^n \widehat{\Sigma}^{-1/2} (X_i - \widehat{m}) (X_i - \widehat{m})^\top \widehat{\Sigma}^{-1/2} - \widehat{\Sigma}^{-1/2} (\widehat{m}_n - \widehat{m})(\widehat{m}_n - \widehat{m})^\top \widehat{\Sigma}^{-1/2}
\end{align}
Define
\begin{align*}
    Z_i := \widehat{\Sigma}^{-1/2} (X_i - \widehat{m}), \quad i \in [n].
\end{align*}
Mirroring the proof of \cite[Theorem 4.7.1]{Vershynin2018}, if we can show that the $Z_i \in \R^d$ are $\kappa$-subgaussian vectors for some $\kappa > 0$, then we will be able to show that
\begin{align*}
    \P\left(\left\|\frac{1}{n} \sum_{i=1}^n Z_i Z_i^\top - I\right\|> \kappa^2 \max\{\delta(t), \delta(t)^2\}\right) \leq 2 \exp(-t^2),
\end{align*}
where
\begin{align*}
    \delta(t) := C \sqrt{\frac{d}{n}} + \frac{t}{\sqrt{n}}
\end{align*}
for some absolute constant $C > 0$.

To proceed in this manner, we observe that for any unit vector $v \in \R^d$ and $t \geq 0$,
\begin{align*}
    \P(|\langle Z_i, v \rangle| \geq t) = \P\left(\left|\langle \widehat{\Sigma}^{-1/2} (X_i - \widehat{m}), v \rangle\right| \geq t\right) &= \P\left(\left|\langle X_i - \widehat{m}, \widehat{\Sigma}^{-1/2} v \rangle\right| \geq t\right) \\
    &= \P\left(\left|\left\langle X_i - \widehat{m}, \frac{\widehat{\Sigma}^{-1/2} v}{\|\widehat{\Sigma}^{-1/2} v\|_2}\right\rangle\right| \geq \frac{t}{\|\widehat{\Sigma}^{-1/2} v\|_2}\right) \\
    &\leq 2 \exp\left(-\frac{\alpha t^2}{8 \|\widehat{\Sigma}^{-1/2} v\|_2^2}\right) \\
    &\leq 2 \exp\left(-\frac{\alpha t^2}{8 \lambda_{\max}(\widehat{\Sigma}^{-1})}\right) \\
    &= 2 \exp\left(-\frac{\alpha \lambda_{\min}(\widehat{\Sigma}) t^2}{8}\right),
\end{align*}
applying Lemma~\ref{lem:chi-2-to-sg} to the distribution $\widehat{\mu}$ and the fact that $\|\widehat{\Sigma}^{-1/2} v\|_2^2 = v^\top \widehat{\Sigma}^{-1} v \leq \lambda_{\max}(\widehat{\Sigma}^{-1})$ for any unit vector $v \in \R^d$. Returning to Lemma~\ref{lem:chi2-to-precond-1}, as long as $\varepsilon < (\sqrt{3} - \sqrt{2})^2$, we have that
\begin{align*}
    \langle v, \widehat{\Sigma} v \rangle \geq (1 - \sqrt{8 \varepsilon} - \varepsilon) \langle v, \Sigma v \rangle \geq (1 - \sqrt{8 \varepsilon} - \varepsilon) \lambda_{\min}(\Sigma).
\end{align*}
Hence, if $\mu$ is $(1 / \alpha)$-smooth (as will be the case for Gaussian cooling), it follows that
\begin{align*}
    \lambda_{\min}(\widehat{\Sigma}) \geq (1 - \sqrt{8 \varepsilon} - \varepsilon) \alpha.
\end{align*}
If we assume, say, that $\varepsilon < \frac{1}{16}$, then $1 - \sqrt{8 \varepsilon} - \varepsilon > \frac{1}{5}$, and we may conclude that
\begin{align*}
    \P(|\langle Z_i, v \rangle| \geq t) \leq 2 \exp\left(-\frac{\alpha^2 t^2}{40}\right),
\end{align*}
showing that once we have placed the sample distribution in isotropic position, the resulting random vectors $Z_i$ are $(\sqrt{40} / \alpha)$-subgaussian, as long as $\chi^2(\widehat{\mu} \| \mu) < \frac{1}{16}$. We may now apply the relevant subgaussian matrix concentration bound from \cite{Vershynin2018} directly, which we include for the sake of completeness.
\begin{theorem}[Theorem 4.6.1 in \cite{Vershynin2018}]
\label{th:sg-matrices}
Let $A \in \R^{n \times d}$ be a random matrix with independent, mean-zero, $\kappa$-subgaussian, isotropic rows $A_i$. Then, for any $t \geq 0$, we have that
\begin{align}
\label{eq:vershynin-two-side}
    \P\left(\left\|\frac{1}{n} A^\top A - I\right\|> \kappa^2 \max\{\delta(t), \delta(t)^2\}\right) \leq 2 \exp(-t^2),
\end{align}
where
\begin{align*}
    \delta(t) := C \sqrt{\frac{d}{n}} + \frac{t}{\sqrt{n}}
\end{align*}
for some absolute constant $C > 0$.
\end{theorem}

Indeed, taking $A$ to the $n \times d$ matrix with rows $Z_i$, we have that
\begin{align*}
    \frac{1}{n} A^\top A = \frac{1}{n} \sum_{i=1}^n Z_i Z_i^\top,
\end{align*}
and it therefore follows immediately from Theorem~\ref{th:sg-matrices} that
\begin{align*}
    \P\left(\left\|\frac{1}{n} \sum_{i=1}^n Z_i Z_i^\top - I\right\|> \frac{40}{\alpha^2} \max\{\delta(t), \delta(t)^2\}\right) \leq 2 \exp(-t^2).
\end{align*}
In fact, with some simple algebraic manipulations (c.f. Remark 4.7.3 and Exercise 4.49 in \cite{Vershynin2018}), one can show that Equation~\ref{eq:vershynin-two-side} implies that
\begin{align*}
    \P\left(\left\|\frac{1}{n} \sum_{i=1}^n Z_i Z_i^\top - I\right\|> \frac{40}{\alpha^2} \left(\frac{d + u}{n} + \sqrt{\frac{d + u}{n}}\right)\right) \leq 2 \exp(-u), \quad u \geq 0.
\end{align*}
For notational simplicity, let $C(\delta) = 1 - (1 + \delta)^{-1/2}$ as above. We aim to choose $u$ and $n$ so that
\begin{align*}
    \frac{40}{\alpha^2} \left(\frac{d + u}{n} + \sqrt{\frac{d + u}{n}}\right) < \frac{1}{2} C(\delta),
\end{align*}
while still maintaining a high probability guarantee. Now, if we consider the inequality
\begin{align*}
    \frac{40}{\alpha^2}(s + \sqrt{s}) < \frac{1}{2} C(\delta),
\end{align*}
we require that $\sqrt{s} < \frac{1}{80} \alpha^2 C(\delta)$, which implies that choosing $s = \Theta(\alpha^4 \delta^2)$ will suffice. Moreover, we require that $\frac{d + u}{n} \leq s$, which holds if and only if $u \leq s n - d$.

Taking $u = \frac{s}{2} n$, and rearranging, we will have the desired bound as long as
\begin{align*}
    n \geq \frac{2 d}{s} = \Omega\left(\frac{d}{\alpha^4 \delta^2}\right).
\end{align*}
In particular, for this choice of $n$, we will have that
\begin{align*}
    \P\left(\left\|\frac{1}{n} \sum_{i=1}^n Z_i Z_i^\top - I\right\| > \frac{1}{2} C(\delta)\right) \leq 2 \exp\left(-\frac{s}{2}n\right) = 2 \exp(-C_1 \alpha^4 \delta^2 n),
\end{align*}
where $C_1 > 0$ is an absolute constant.

We now turn our attention to the second term, which is slightly easier to analyze. Take
\begin{align*}
    R_n := \widehat{\Sigma}^{-1/2}(\widehat{m}_n - \widehat{m}) = \frac{1}{n} \sum_{i=1}^n \widehat{\Sigma}^{-1/2} \left(X_i - \widehat{m}\right) = \frac{1}{n} \sum_{i=1}^n Z_i,
\end{align*}
so that
\begin{align*}
    \|\widehat{\Sigma}^{-1/2} (\widehat{m}_n - \widehat{m})(\widehat{m}_n - \widehat{m})^\top \widehat{\Sigma}^{-1/2}\|_2 = \|R_n R_n^\top\|_2 = \|R_n\|_2^2.
\end{align*}
Noting that the $Z_i$ are independent, centered, and isotropic by construction, it follows that
\begin{align*}
    Y_n := \sqrt{n} R_n = \frac{1}{\sqrt{n}} \sum_{i=1}^n Z_i
\end{align*}
is centered and isotropic. Below, we take the Orlicz norm associated to the function $\Phi(u) = \exp(u^2) - 1$ to be $\|X\|_{\psi_2} := \inf\{\kappa > 0 : \E[\exp(X^2 / \kappa^2)] \leq 2\}$, as in \cite[Definition 2.6.4]{Vershynin2018}. Then, if $X$ is $\kappa$-subgaussian, $\|X\|_{\psi_2} \leq c \kappa$ for some universal $c > 0$.

Now, because the $Z_i$ are independent, centered, $\kappa$-subgaussian random variables, \cite[Proposition 2.7.1]{Vershynin2018} reveals that for some universal $C' > 0$,
\begin{align*}
    \|Y_n\|_{\psi_2}^2 = \frac{1}{n} \left\|\sum_{i=1}^n Z_i\right\|_{\psi_2}^2 \leq \frac{C'}{n} \sum_{i=1}^n \|Z_i\|_{\psi_2}^2 = C' c^2 \kappa^2
\end{align*}
so that $Y_n$ itself is $\Theta(\kappa)$-subgaussian. To proceed, we require a standard $\epsilon$-net argument.
\begin{lemma}
\label{lem:norm-eps-net}
With $R_n$ and $C(\delta)$ defined as above, if $n = \Omega\left(\frac{d}{\alpha^4 \delta^2}\right)$, then
\begin{align*}
    \P\left(\|R_n\|_2^2 > \frac{1}{2} C(\delta)\right) \leq \exp(-C_2 \alpha^2 \delta n).
\end{align*}
\end{lemma}
\begin{proof}
Choose $\eps = 1 / 2$. By \cite[Corollary 4.2.11]{Vershynin2018}, we can construct an $\eps$-net $\sN$ of $\mathbb{S}^{d-1}$ of size $|\sN| \leq 5^d$. Choosing $v \in \R^d$ such that $\|R_n\|_2 = \langle R_n, v \rangle$, we may find $v' \in \sN$ such that $\|v - v'\|_2 \leq \frac{1}{2}$. Thus,
\begin{align*}
    |\langle R_n, v - v' \rangle| \leq \|R_n\|_2 \|v - v'\|_2 \leq \frac{1}{2} \|R_n\|_2,
\end{align*}
implying that
\begin{align*}
    |\langle R_n, v' \rangle| \geq |\langle R_n, v \rangle| - |\langle R_n, v - v' \rangle| \geq \frac{1}{2} \|R_n\|_2.
\end{align*}
In particular, $\|R_n\|_2 \leq 2 \max_{v \in \sN} |\langle R_n, v \rangle|.$
Now, as $R_n$ is $\Theta(\kappa / \sqrt{n})$-subgaussian, for any fixed $v \in \sN$, we have that
\begin{align*}
    \P(|\langle R_n, v \rangle| > t) \leq 2 \exp\left(-\frac{c_0 n t^2}{\kappa^2}\right)
\end{align*}
for some absolute constant $c_0 > 0$. Via a union bound over all $v \in \sN$, we have that
\begin{align*}
    \P(\|R_n\|_2 > t) \leq \P\left(\max_{v \in \sN} |\langle R_n, v \rangle| > \frac{t}{2}\right) &\leq 2 \cdot 5^d \exp\left(-\frac{c_0 n t^2}{\kappa^2}\right) = 2 \exp\left(-\frac{c_0 n t^2}{\kappa^2} + d \log 5\right)
\end{align*}
taking $c_0$ smaller if necessary. Hence, for $t = \sqrt{C(\delta) / 2}$, it follows that
\begin{align*}
    \P(\|R_n\|_2^2 > C(\delta) / 2) \leq 2 \exp\left(-\frac{c_0 n C(\delta)}{\kappa^2} + d \log 5\right).
\end{align*}
In particular, there exists universal $C_2 > 0$ such that for $n = \Omega\left(\frac{d}{\alpha^2 \delta}\right)$, we have that
\begin{align*}
    \P\left(\|R_n\|_2^2 > \frac{1}{2} C(\delta)\right) \leq 2 \exp(-C_2 \alpha^2 \delta n),
\end{align*}
using the fact that $\kappa = \Theta(1 / \alpha)$ and $C(\delta) = \Theta(\delta)$ for $\delta$ sufficiently small.
\end{proof}

Combining our estimates of both terms in Equation~\ref{eq:bv-decomp}, we have that for $n = \Omega\left(\frac{d}{\alpha^4 \delta^2}\right)$,
\begin{align*}
    \P\left(\|\widehat{\Sigma}^{-1/2} \widehat{\Sigma}_n\widehat{\Sigma}^{-1/2} - I\|_2 > C(\delta)\right) &\leq \P\left(\left\|\frac{1}{n} \sum_{i=1}^n Z_i Z_i^\top - I\right\| > \frac{1}{2} C(\delta)\right) + \P\left(\|R_n\|_2^2 > \frac{1}{2} C(\delta)\right) \\
    &\leq 2 \exp(-C_1 \alpha^4 \delta^2 n) + 2 \exp(-C_2 \alpha^2 \delta n) \\
    &\leq 4 \exp(-C_3 \alpha^4 \delta^2 n)
\end{align*}
for some universal $C_3 > 0$. This lemma is sufficient to conclude the main result of this section.
\begin{theorem}
\label{th:sample-cov-concentration}
Suppose that $\mu \in \sP_2(\R^d)$ is an $\alpha$-strongly log-concave and $(1 / \alpha)$-smooth measure on $\R^d$ with $\Sigma := \covar(\mu)$ and that $\widehat{\mu} \in \sP_2(\R^d)$ is such that for fixed $\delta \in (0, 1)$,
\begin{align*}
    \chi^2(\widehat{\mu} \| \mu) < \min\left\{\left(\sqrt{3 - (1 + \delta)^{-1/2}} - \sqrt{2}\right)^2, \frac{1}{16}\right\}.
\end{align*}
There exist universal constants $C_1, C_2 > 0$ such that if $X_1, \ldots, X_n \overset{i.i.d.}{\sim} \widehat{\mu}$ and
\begin{align*}
    \widehat{\Sigma}_n := \frac{1}{n} \sum_{i=1}^n \left(X_i - \frac{1}{n} \sum_{j=1}^n X_j\right)\left(X_i - \frac{1}{n} \sum_{j=1}^n X_j\right)^\top
\end{align*}
is the corresponding (biased) sample covariance, formed from at least $n = \frac{C_1 d}{\alpha^4 \delta^2}$ samples, then
\begin{align*}
    (1 + \delta)^{-1} \Sigma \preceq \widehat{\Sigma}_n \preceq (1 + \delta) \Sigma
\end{align*}
with probability at least $1 - 4 \exp(-C_2 \alpha^4 \delta^2 n)$.
\end{theorem}

\subsection{Concentration in the Translation-Invariant Case}
As in the general case in the preceding section, we can reduce to well-established concentration inequalities for subgaussian random variables. However, we only need to establish that the $d$ entries of the vector $\widehat{\sigma}$, defined in Equation~\ref{eq:emp-fourier-sigma}, concentrate well. For completeness, recall that in the translation-invariant case, the target covariance $\Sigma$ is diagonalized by a discrete Fourier transform in the sense that
\begin{align*}
    \Sigma = \sF \diag(\sigma) \sF^*
\end{align*}
for some $\sigma \in \R^d$. We then must establish that the entries of our estimate $\widehat{\sigma}$ satisfy
\begin{align}
\label{eq:sigma-concentration}
    \left|\frac{\widehat{\sigma}_i}{\sigma_i} - 1\right| \leq C(\delta), \quad i \in [d],
\end{align}
simultaneously, as seen in Equation~\ref{eq:fourier-vector-concentration}. To this end, we may show that Equation~\ref{eq:sigma-concentration} holds individually for each $i \in [d]$ and then perform a union bound over the $d$ entries of $\widehat{\sigma}$. We begin by writing
\begin{align*}
    \widehat{\sigma}_i = \frac{1}{n} \sum_{j=1}^n |\langle e_i, \sF(x_j - \widehat{m}_n)\rangle|^2 =  \frac{1}{n} \sum_{j=1}^n |\langle x_j - \widehat{m}_n, v_i \rangle|^2,
\end{align*}
where $v_i = \sF^* e_i \in \C^d$ (taking the physics convention, the inner product on $\C^d$ is linear in its second argument). Note that $\|v_i\| = 1$, as $\sF^*$ is an isometry. On the other hand, if $X \sim \mu$ is drawn from the target distribution, we have that
\begin{align*}
    \sigma_i = \langle e_i, \sF^* \Sigma \sF e_i \rangle &= \E[\langle e_i, \sF^* (X - \E[X])(X - \E[X])^\top \sF e_i \rangle] \\
    &= \E[|\langle X - \E[X], \sF^* e_i\rangle|^2] \\
    &= \E[|\langle X - \E[X], v_i\rangle|^2].
\end{align*}
Splitting $v_i$ into its real and imaginary components, we observe that if $\mu$ is an $(1/\alpha)$-smooth distribution, then
\begin{align*}
    \E[|\langle X - \E[X], v_i\rangle|^2] = \E[|\langle X - \E[X], \Re(v_i)\rangle|^2] + \E[|\langle X - \E[X], \Im(v_i)\rangle|^2] \geq \alpha,
\end{align*}
invoking Corollary~\ref{cor:sandwich} to obtain a lower bound on the covariance of $\mu$. Hence, the righthand side of Equation~\ref{eq:sigma-concentration} is well-defined, and in fact, it suffices to show that
\begin{align*}
    |\widehat{\sigma}_i - \sigma_i| \leq \alpha C(\delta), \quad i \in [d].
\end{align*}

Fixing $v \in \R^d$ to be a vector with $\|v\| \leq 1$ (as is the case for both $\Re(v_i)$ and $\Im(v_i)$ above), it then suffices to obtain a sample complexity estimate such that
\begin{align*}
    \left|\frac{1}{n} \sum_{j=1}^n |\langle x_j - \widehat{m}_n, v \rangle|^2 - \E[|\langle X - \widehat{m}, v \rangle|^2]\right| \leq \frac{\alpha C(\delta)}{2}
\end{align*}
with high probability, where $\widehat{m} = \E_{\widehat{\mu}}[X]$. The extra factor of two in the denominator arises because we will eventually need to apply such an inequality twice, to the real and imaginary components of the vector $v_i = \sF^* e_i$. Now, for each $j \in [n]$, define
\begin{align}
\label{eq:Zj}
    Z_j := \langle x_j - \widehat{m}, v \rangle.
\end{align}
We also take
\begin{align}
\label{eq:Zhat}
    \widehat{Z}_n := \frac{1}{n} \sum_{j=1}^n Z_j = \frac{1}{n} \sum_{j=1}^n \langle x_j - \widehat{m}, v \rangle = \langle \widehat{m}_n - \widehat{m}, v \rangle.
\end{align}
Now, it is easy to see that
\begin{align*}
    \frac{1}{n} \sum_{j=1}^n |\langle x_j - \widehat{m}_n, v \rangle|^2 &= \frac{1}{n} \sum_{j=1}^n |Z_j - \widehat{Z}_n|^2 \\
    &= \frac{1}{n} \sum_{j=1}^n Z_j^2 - \widehat{Z}_n^2.
\end{align*}
If $Z \overset{d}{=} \langle X - \widehat{m}, v \rangle$, then
\begin{align*}
    \E[|\langle X - \widehat{m}, v \rangle|^2] = \E[Z^2],
\end{align*}
and have reduced to showing that for $Z_1, \ldots, Z_n \overset{i.i.d.}{\sim} \mathrm{Law}(Z)$,
\begin{align}
\label{eq:translation-invariant-decomp}
    \left|\frac{1}{n} \sum_{j=1}^n Z_j^2 - \E[Z^2]\right| + \widehat{Z}_n^2 \leq \frac{\alpha C(\delta)}{2}
\end{align}
via triangle inequality.
In a very similar manner to the general case, the first term can be handled via an appropriate concentration inequality for sub-exponential random variables. 
\begin{lemma}
\label{lem:subexp-concentration-1}
Suppose that $\mu \in \sP_2(\R^d)$ is an $\alpha$-strongly log-concave and $(1 / \alpha)$-smooth measure on $\R^d$ and that $\widehat{\mu} \in \sP_2(\R^d)$ is such that for fixed $\delta \in (0, 1)$,
\begin{align*}
    \chi^2(\widehat{\mu} \| \mu) < \min\left\{\left(\sqrt{3 - (1 + \delta)^{-1/2}} - \sqrt{2}\right)^2, \frac{1}{16}\right\}.
\end{align*}
There exists a universal constant $C > 0$ such that if $X_1, \ldots, X_n \overset{i.i.d.}{\sim} \widehat{\mu}$, then for $Z_1, \ldots, Z_n$ as defined in Equation~\ref{eq:Zj} and $Z \sim \mathrm{Law}(Z_1)$,
\begin{align*}
    \left|\frac{1}{n} \sum_{j=1}^n Z_j^2 - \E[Z^2]\right| \leq \frac{\alpha C(\delta)}{4}
\end{align*}
with probability at least $1 - 2 \exp(- C n \alpha^4 \delta^2)$.
\end{lemma}
\begin{proof}
We write
\begin{align*}
    \left|\frac{1}{n} \sum_{j=1}^n Z_j^2 - \E[Z^2]\right| = \left|\sum_{j=1}^n \frac{1}{n}(Z_j^2 - \E[Z^2])\right|.
\end{align*}
Now, the random variables $Y_j := Z_j^2 - \E[Z^2]$ are independent, mean-zero, and $O(1/\alpha)$-subexponential, assuming that $\widehat{\mu}$ is the same distribution discussed in Lemma~\ref{lem:chi-2-to-sg}. Indeed, $Z_j = \langle X_j - \widehat{m}, v \rangle$ is shown to be $(\sqrt{8/ \alpha})$-subgaussian in said lemma, and a combination of \cite[Lemma 2.8.5 and Equation (2.26)]{Vershynin2018} show that the centered variance $Z_j^2 - \E[Z^2]$ is then subexponential with
\begin{align*}
    \|Z_j^2 - \E[Z^2]\|_{\psi_1} \leq C\|Z_j^2\|_{\psi_1} = C\|Z_j\|_{\psi_2}^2 = \frac{8 C}{\alpha},
\end{align*}
where $C > 0$ is a universal constant and the (subexponential) Orlicz norm $\|\cdot\|_{\psi_1}$ is defined by $\|X\|_{\psi_1} := \inf\{K > 0 : \E[\exp(|X| / K)] \leq 2\}$. This observation allows us to apply the subexponential Bernstein inequality from \cite[Corollary 2.9.2]{Vershynin2018}, which yields
\begin{align*}
    \P\left(\left|\sum_{j=1}^n \frac{1}{n}(Z_j^2 - \E[Z^2])\right| \geq \frac{\alpha C(\delta)}{2}\right) &\leq 2 \exp\left[-C n \min\left\{\frac{\alpha^2 C(\delta)^2}{\|Z_j^2 - \E[Z^2]\|_{\psi_1}^2}, \frac{\alpha C(\delta)}{\|Z_j^2 - \E[Z^2]\|_{\psi_1}}\right\}\right]\\
    &\leq 2 \exp\left(- C n \alpha^4 C(\delta)^2\right),
\end{align*}
where $C > 0$ is some universal constant (taken larger if necessary in the second inequality). The conclusion then follows because $C(\delta) = \Theta(\delta)$.
\end{proof}
If $\varepsilon > 0$ is fixed, then as long as we take
\begin{align}
\label{eq:translation-inv-sample-complexity}
    n = \Omega\left(\frac{\log(d / \varepsilon)}{\alpha^4 \delta^2}\right)
\end{align}
we can ensure that
\begin{align*}
    \P\left(\left|\sum_{j=1}^n \frac{1}{n}(Z_j^2 - \E[Z^2])\right| \geq \frac{\alpha C(\delta)}{4}\right) \leq \frac{\varepsilon}{4 d}.
\end{align*}

Turning to the second term in Equation~\ref{eq:translation-invariant-decomp}, we observe that
\begin{align*}
    \P\left(\widehat{Z}_n^2 \geq \frac{\alpha C(\delta)}{4}\right) = \P\left(\left|\frac{1}{n} \sum_{j=1}^n Z_j \right| \geq \sqrt{\frac{\alpha C(\delta)}{4}}\right).
\end{align*}
As noted in the proof Lemma~\ref{lem:subexp-concentration-1}, the $Z_j$ are $(\sqrt{8/ \alpha})$-subgaussian. As a result, we can straightforwardly apply Hoeffding's inequality. From \cite[Theorem 2.7.3]{Vershynin2018}, we know that
\begin{align*}
    \P\left(\left|\frac{1}{n} \sum_{j=1}^n Z_j \right| \geq \sqrt{\frac{\alpha C(\delta)}{4}}\right) &\leq 2 \exp\left(-c n^2 \frac{\alpha C(\delta)}{\sum_{j=1}^n \|Z_j\|_{\psi_2}^2}\right) \\
    &= 2 \exp(-c n \alpha^2 \delta),
\end{align*}
for some universal $c > 0$,
recalling that $\|Z_j\|_{\psi_2} = O(\alpha^{-1/2})$. Hence, taking $n$ samples, where $n$ is given as in Equation~\ref{eq:translation-inv-sample-complexity}, suffices to ensure that
\begin{align*}
    \widehat{Z}_n^2 \leq \frac{\alpha C(\delta)}{4}
\end{align*}
with probability at least $1 - \frac{\varepsilon}{4d}$. We have in particular shown the following:
\begin{lemma}
\label{lem:subexp-concentration-combined}
Suppose that $\mu, \widehat{\mu} \in \sP_2(\R^d)$ are as in Lemma~\ref{lem:subexp-concentration-1}.
Then, for any fixed $v \in \R^d$ with $\|v\| \leq 1$ and $\varepsilon > 0$, as long as $n$ is as in Equation~\ref{eq:translation-inv-sample-complexity}, we have that
\begin{align*}
    \P\left(\left|\frac{1}{n} \sum_{j=1}^n |\langle x_j - \widehat{m}_n, v \rangle|^2 - \E[|\langle X - \widehat{m}, v \rangle|^2]\right| \geq \frac{\alpha C(\delta)}{2}\right) \leq \frac{\varepsilon}{2 d}.
\end{align*}
\end{lemma}

We conclude, in the translation-invariant setting, by simply taking a union bound over the $2 d$ vectors $\Re(v_i)$ and $\Im(v_i)$, where $v_i = \sF^* e_i$ as identified above.

\begin{corollary}
\label{cor:translation-invariant-union-bound}
Suppose that $\mu \in \sP_2(\R^d)$ is an $\alpha$-strongly log-concave and $(1 / \alpha)$-smooth measure on $\R^d$ with translation-invariant covariance $\Sigma = \sF \diag(\sigma) \sF^*$ and $\widehat{\mu} \in \sP_2(\R^d)$ is such that
\begin{align*}
    \chi^2(\widehat{\mu} \| \mu) < \min\left\{\left(\sqrt{3 - (1 + \delta)^{-1/2}} - \sqrt{2}\right)^2, \frac{1}{16}\right\}.
\end{align*}
If $n = \Omega\left(\frac{\log(d / \varepsilon)}{\alpha^4 \delta^2}\right)$ and $X_1, \ldots, X_n \overset{i.i.d.}{\sim} \widehat{\mu}$, then $\widehat{\sigma}$ from Equation~\ref{eq:emp-fourier-sigma} satisfies
\begin{align*}
    \|\widehat{\sigma} \odot \sigma^{-1} - \mathbf{1}\|_\infty \leq C(\delta).
\end{align*}
with probability at least $1 - \varepsilon$.
\end{corollary}
\begin{proof}
As noted above, this claim will follow if we can show that $|\widehat{\sigma}_i - \sigma_i| \leq \alpha C(\delta)$ with probability at least $1 - \frac{\varepsilon}{d}$ for each $i \in [d]$, with the final claim then following directly from a union bound over all $i \in [d]$. If $v_i = \sF^* e_i \in \C^d$, then
\begin{align*}
    \left|\sigma_i - \widehat{\sigma}_i\right| &= \left|\frac{1}{n}\sum_{j=1}^n |\langle x_j - \widehat{m}_n, v_i \rangle|^2 - \E[|\langle X - \widehat{m}, v_i \rangle|^2]\right|  \\
    &\leq \left|\frac{1}{n}\sum_{j=1}^n |\langle x_j - \widehat{m}_n, \Re(v_i) \rangle|^2 - \E[|\langle X - \widehat{m}, \Re(v_i) \rangle|^2]\right| \\
    &\qquad + \left|\frac{1}{n}\sum_{j=1}^n |\langle x_j - \widehat{m}_n, \Im(v_i) \rangle|^2 - \E[|\langle X - \widehat{m}, \Im(v_i) \rangle|^2]\right|.
\end{align*}
By Lemma~\ref{lem:subexp-concentration-combined}, both summands on the righthand side of the preceding inequality are at most $\frac{\alpha C(\delta)}{2}$ with probability at least $1 - \frac{\varepsilon}{2 d}$ given our choice of $n$. Hence, for each $i \in [d]$, we have that
\begin{align*}
    \P(|\sigma_i - \widehat{\sigma}_i| \geq \alpha C(\delta)) \leq \frac{\varepsilon}{d}
\end{align*}
by union bounding over the event that either summand exceeds the threshold $\alpha C(\delta)$.
\end{proof}

\section{Additional Lattice $\phi^4$ Proofs and Experiments}
\label{app:additional-phi4-facts}

In this appendix, we expand upon and provide proofs of several facts about lattice $\phi^4$ distribution, as introduced in Section~\ref{sec:num-exps}. Additionally, we determine the sensitivity of translation-invariant Gaussian cooling to the number of ULMC steps per chain $N$ and the number of chains per stage $n$ in Figure~\ref{fig:phi4-scaling}.

\subsection{Additional Proofs}
In full generality, the lattice $\phi^4$ Hamiltonian takes the following form:
\begin{align}
\label{eq:lattice-hamiltonian}
    H(\phi) := \frac{1}{4}\sum_{x, y \in \Lambda} \beta_{xy}|\phi_x - \phi_y|^2 + h \sum_{x \in \Lambda} \phi_x + \sum_{x \in \Lambda} w(\phi_x)
\end{align}
Above, $\beta_{xy} = \beta_{yx} \geq 0$ is a symmetric coupling constant for each pair $(x, y) \in \Lambda \times \Lambda$, $h \in \R$ is an external magnetic field, and $w : \R \to \R$ is a scalar potential. For the lattice $\phi^4$ model, we then take the scalar potential $w$ to be
\begin{align*}
    w(\phi_x) = \frac{\lambda}{4} \phi_x^4 + \frac{m}{2} \phi_x^2.
\end{align*}

First, to understand why the lattice $\phi^4$ model provides a useful testbed for Gaussian cooling, we compute the Hessian of the Hamiltonian:
\begin{align*}
    \nabla^2 H(\phi) = (\Delta_\beta + m I) + 3 \lambda \diag(\phi^2).
\end{align*}
where $\phi^2$ is interpreted elementwise. As long as $m > 0$ and the coupling constant $\beta$ is also non-negative, having specialized to the case that $\Delta_\beta$ is a weighted graph Laplacian, it follows immediately that
\begin{align*}
    \nabla^2 H(\phi) \succeq m I, \quad \forall \phi \in \R^{\Lambda}.
\end{align*}
In other words, the $H$ is $m$-strongly convex. However, it is \textit{not} the case that $H$ is $L$-smooth for any $L > 0$. Indeed, because the Hessian of the quartic term in the lattice $\phi^4$ potential remains dependent on $\phi$, taking $|\phi| \to \infty$ allows us to deduce that for any $L > 0$, there exists $\phi \in \R^\Lambda$ such that
\begin{align*}
    \nabla^2 H(\phi) \succeq 3 \lambda \diag(\phi^2) \succeq L I.
\end{align*}

In order for this example to fit into our framework, we therefore need to truncate the scalar $\phi^4$ potential to a quadratic potential past some threshold. More specifically, we replace the scalar potential $w$ defined above with a truncated scalar potential of the form:
\begin{align}
\label{eq:truncated-potential}
w_R(\phi_x) := \begin{cases}
\frac{\lambda}{4} R^4 + \frac{m}{2} R^2 - (\lambda R^3 + m R) (\phi_x + R) + \frac{1}{2}(3 \lambda R^2 + m) (\phi_x + R)^2 & \phi_x < - R, \\
w(\phi_x) & |\phi_x| \leq R, \\
\frac{\lambda}{4} R^4 + \frac{m}{2} R^2 + (\lambda R^3 + m R) (\phi_x - R) + \frac{1}{2}(3 \lambda R^2 + m) (\phi_x - R)^2 & \phi_x > R.
\end{cases}
\end{align}
By construction, the truncated scalar potential $w_R$ is twice continuously differentiable, and we have that
\begin{align*}
    w_R''(\phi_x) = \begin{cases}
    w''(\phi_x) & |\phi_x| \leq R, \\
    3 \lambda R^2 + m & |\phi_x| > R.
    \end{cases}
\end{align*}
Because $|w''(\phi_x)| \leq 3 \lambda R^2 + m$ for $|\phi_x| \leq R$ and the contribution of the scalar potential to the Hessian of the Hamiltonian is diagonal, the truncated Hamiltonian $H_R(\phi)$ now satisfies, at the very least,
\begin{align*}
 \Delta_\beta + m I \preceq \nabla^2 H_R(\phi) \preceq \Delta_\beta + (3 \lambda R^2 + m) I.
\end{align*}
for all $\phi \in \R^\Lambda$. 
If $\lambda_1 = \lambda_{\min}(\Delta_\beta)$ and $\lambda_{n} = \lambda_{\max}(\Delta_\beta)$, then we conclude that the measure $\pi_R(\phi) \propto \exp(-H_R(\phi))$ induced by the truncated Hamiltonian is indeed
strongly log-concave and log-smooth, as
\begin{align*}
    \nabla^2 H_R(\phi) \preceq (\lambda_n + 3 \lambda R^2 + m)I.
\end{align*}
We can guarantee that Assumptions~\ref{asmp:conditioning}--\ref{asmp:strong-log-concave-smooth} are satisfied with the following parameters:
\begin{align}
\label{eq:lattice-phi4-params}
\begin{split}
&\alpha = \left(\frac{3 \lambda R^2}{\lambda_{\min}(\Delta_\beta) + m} + 1\right)^{-1}, \\
&\mu = \lambda_{\min}(\Delta_\beta) + m, \\
&L = \lambda_{\max}(\Delta_\beta) + m + 3 \lambda R^2.
\end{split}
\end{align}

To show that the covariance is block circulant with circulant blocks, it suffices to show that for all $x, y \in \Lambda_d$
\begin{align*}
    \E_{\phi \sim \pi}[\phi_x \phi_y] - \E_{\phi \sim \pi}[\phi_x] \E_{\phi \sim \pi}[\phi_y] = C([x - y]),
\end{align*}
where $C : [d - 1] \times [d - 1] \to \R$ is a function of $[x - y] := (x_1 - y_1 \bmod{d}, x_2 - y_2 \bmod{d})$. Indeed, if this holds for all $x, y \in \Lambda_d$, then we can partition $\Lambda_d \times \Lambda_d$ into $d \times d$ blocks by ordering the lattices sites $(x, y)$ in ascending order as $(0, 0), (0, 1), \ldots, (0, d - 1), \ldots, (d - 1, 0,), \ldots, (d - 1, d - 1)$. It is then evident that the resulting covariance matrix $\Sigma = \covar(\pi) \in \R^{d^2 \times d^2}$ is block circulant with $d \times d$ circulant blocks. Consequently, it is diagonalized by a discrete Fourier transform, exactly as in Section~\ref{subsec:translation-invariant}.

\begin{lemma}
\label{lem:block-circulance}
The entries $\Sigma_{xy}$ of the covariance $\Sigma := \covar(\pi)$ of the (truncated) lattice $\phi^4$ measure depend only on the increments $[x - y] := (x_1 - y_1 \bmod{d}, x_2 - y_2 \bmod{d})$.
\end{lemma}
\begin{proof}
Fix some $x, y \in \Lambda$. Define the linear shift operator $P_y : \R^{\Lambda} \to \R^{\Lambda}$ by $(P_y \phi)_x = \phi_{[x + y]}$. We observe that $P_y$, which simply permutes the entries of the periodic lattice $\Lambda$ by constant shift of $y$, is unitary, with $(P_y^\dagger \phi)_x = \phi_{[x - y]}$. Now, let $\psi := P_y \phi$.

As a result, if we consider the original lattice $\phi^4$ potential in Equation~\ref{eq:lattice-phi4-potential}, we have that
\begin{align*}
    \sum_{x \in \Lambda} w(\psi_x) = \sum_{x \in \Lambda} w(\phi_{[x + y]}) = \sum_{z \in \Lambda} w(\phi_z),
\end{align*}
with the same holding true upon replacing $w$ with the truncated potential $w_R$ in Equation~\ref{eq:truncated-potential}. We also compute
\begin{align*}
    \langle \phi, \Delta_\beta \phi \rangle = \langle P_y^\dagger \psi, \Delta_\beta P_y^\dagger \psi \rangle = \langle \psi, P_y \Delta_\beta P_y^\dagger \psi \rangle.
\end{align*}
Now, by definition of the Laplacian, we observe that for any $v, w \in \Lambda$,
\begin{align*}
    \langle e_v, P_y \Delta_\beta P_y^\dagger e_w \rangle = \langle P_y^\dagger e_v, \Delta_\beta P_y^\dagger e_w \rangle &= \langle e_{[v + y]}, \Delta_\beta e_{[w + y]} \rangle \\
    &= \begin{cases}
    \beta & [v + y] \sim [w + y], \\
    0 & \text{otherwise},
    \end{cases} \\
    &= (\Delta_\beta)_{vw}.
\end{align*}
Indeed, in the $d$-periodic two-dimensional lattice, $[v + y] \sim [w + y]$ if and only if $v \sim w$. Hence, $\langle \phi, \Delta_\beta \phi \rangle = \langle \psi, \Delta_\beta \psi \rangle$. With both of these facts together, we see that 
\begin{align*}
    H(\phi) = \frac{1}{2}\langle \phi, \Delta_\beta \phi \rangle + \sum_{x \in \Lambda} w(\phi_x) = \frac{1}{2}\langle \psi, \Delta_\beta \psi \rangle + \sum_{x \in \Lambda} w(\psi_x) = H(\psi)
\end{align*}
is translation-invariant, with the same holding true for the truncated potential.

We can now directly compute each entry of the covariance to establish the claimed translation invariance. As noted above, this same computation goes through with the original potential and the truncated potential, so we only consider the former. Letting $Z$ denote the constant normalization of $\pi(\phi) \propto \frac{1}{Z} \exp(-H(\phi))$ and $\mu^\Lambda$ denote the Lebesgue measure on $\R^{\Lambda}$, we first compute
\begin{align*}
    \E_{\pi}[\phi_x \phi_y] &= \frac{1}{Z} \int_{\R^{\Lambda}} \phi_x \phi_y \exp(-H(\phi)) d\mu^\Lambda(\phi) \\
    &= \frac{1}{Z} \int_{\R^\Lambda} (P_y^\dagger \psi)_x (P_y^\dagger \psi)_y \exp(-H(P_y^\dagger \psi)) d\mu^\Lambda(\psi) \\
    &= \frac{1}{Z} \int_{\R^\Lambda} \psi_{[x - y]} \psi_0 \exp(-H(\psi)) d\mu^\Lambda(\psi) \\
    &= \E_\pi[\phi_{[x - y]} \phi_0].
\end{align*}
Above, we change variables to $\psi = P_y \phi$, use the fact that $H(\phi) = H(P_y^\dagger \psi) = H(\psi)$, and note that $P_y$ is a unitary and hence the Jacobian of the map $\phi \mapsto P_y \phi$ is one. An identical computation shows that $\E_{\pi}[\phi_x] = \E_{\pi}[\phi_{[x - y]}]$ and that $\E_{\pi}[\phi_y] = \E_{\pi}[\phi_0]$, applying the same change of variables in both cases. It follows that
\begin{align*}
    \Sigma_{xy} = \E_{\pi}[\phi_x \phi_y] - \E_{\pi}[\phi_x] \E_{\pi}[\phi_y] &= \E_{\pi}[\phi_{[x - y]} \phi_0] - \E_{\pi}[\phi_{[x - y]}] \E_{\pi}[\phi_0] \\
    &= \Sigma_{[x - y] 0},
\end{align*}
which depends only on the increment $[x - y]$.
\end{proof}

We can now verify Equation~\ref{eq:lattice-phi4-params} directly. Observe first that
\begin{align*}
    \kappa(\nabla^2 H_R(\phi)) \leq \frac{\lambda_n + 3 \lambda R^2 + m}{\lambda_1 + m}.
\end{align*}
One may choose the entire quadratic part of the potential as the preconditioner, i.e., taking $M_\beta := \Delta_\beta + m I$. In this case, we obtain
\begin{align*}
    I \preceq M_\beta^{-1/2} \nabla^2 H(M_\beta^{-1/2} \phi) M_\beta^{-1/2} \preceq I + 3 \lambda R^2 M_{\beta}^{-1}.
\end{align*}
In particular, the resulting condition number is
\begin{align*}
    \kappa(M_\beta^{-1/2} \nabla^2 H(M_\beta^{-1/2} \phi) M_\beta^{-1/2}) = 1 + \frac{3 \lambda R^2}{\lambda_1 + m},
\end{align*}
which retains some undesirable but unavoidable dependence on $R$, the radius of truncation.
To reconcile the example of lattice $\phi^4$ theory with the condition in Equation~\ref{eq:true-conditioning}, observe that because $M_\beta \succeq (\lambda_1 + m)I$, we have that
\begin{align*}
    \frac{3 \lambda R^2}{\lambda_1 + m} M_\beta \succeq 3 \lambda R^2 I.
\end{align*}
Hence, we have that
\begin{align*}
    M_\beta \preceq \nabla^2 H(\phi) \preceq M_\beta + 3 \lambda R^2 I \preceq \left(\frac{3 \lambda R^2}{\lambda_1 + m} + 1\right) M_\beta.
\end{align*}
If we define
\begin{align*}
    P_\beta := \left(\frac{3 \lambda R^2}{\lambda_1 + m} + 1\right) M_\beta,
\end{align*}
then we further have that 
\begin{align*}
    \left(\frac{3 \lambda R^2}{\lambda_1 + m} + 1\right)^{-1} P_\beta \preceq \nabla^2 H(\phi) \preceq P_\beta.
\end{align*}

\subsection{Additional Experiments}

The heatmap in Figure~\ref{fig:phi4-scaling} displays the sensitivity of the preconditioner estimate that Algorithm~\ref{alg:gaussian-cooling} outputs to both the number of chains per stage $n$ and the number of ULMC steps per chain $N$ for a lattice with sidelength $d = 100$. For all experiments, the number of stages of Gaussian cooling is limited to $K = 12$, with $\lambda = 0.5$, and $\beta = 2$. For this particular experiment, a truncation radius of $R = 4$ was used. The number of ULMC steps per chain appears to have a larger impact on the relative conditioning than the number of chains used.

\begin{figure}[h]
    \centering
    \includegraphics[width=0.6\linewidth]{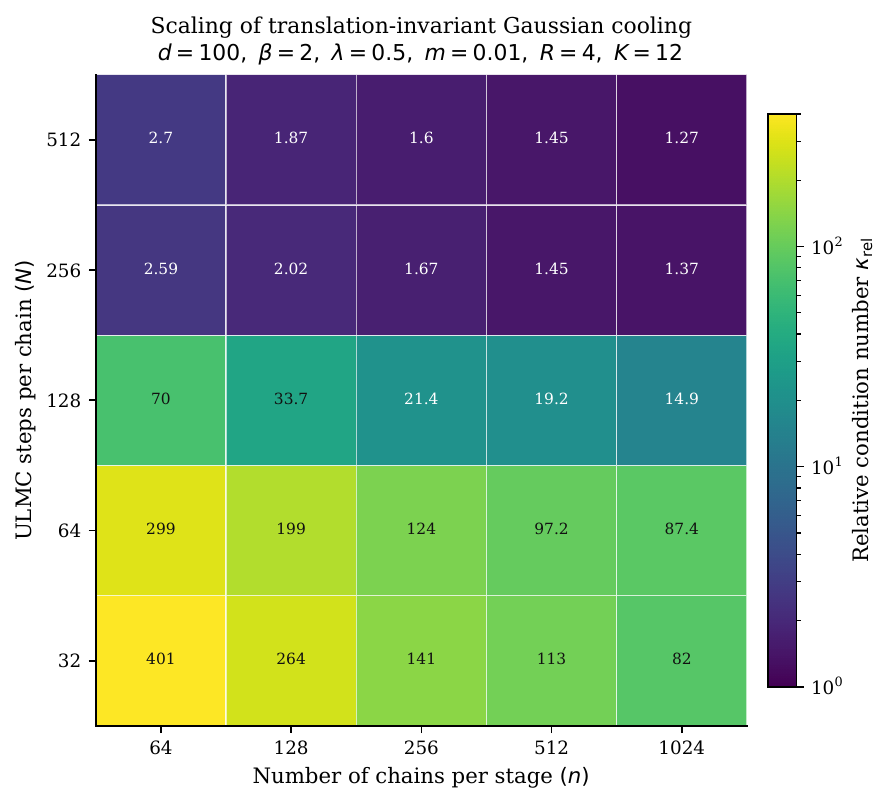}
    \caption{A comparison of translation-invariant Gaussian cooling performance across a range of $N$, the number of steps per chain, and $n$, the number of ULMC chains. For the lattice with sidelength $d= 100$ and mass $m = 0.01$, it is possible to reach nearly optimal preconditioning with relatively few chains and steps per chain.}
    \label{fig:phi4-scaling}
\end{figure}

\section{Experimental Details}
\label{app:experiment-details}
All experiments were implemented fully in JAX and run on a single NVIDIA H100 SXM GPU with 80 GB VRAM. As stated in our AI disclosure, these experiments were implemented with the aid of generative AI tools, and users that wish to modify or build upon our code are encouraged to interface with the codebase through an LLM. We provide a \href{https://github.com/jakehofgard/gaussian-cooling}{GitHub repository} that contains code for reproducing all experiments in Section~\ref{sec:num-exps}.
\begin{itemize}
    \item \texttt{README.md} contains a high-level overview of the code, including descriptions of the files and directories listed below.
    \item \texttt{gaussian\_cooling\_algs.py} contains implementations of Algorithm~\ref{alg:gaussian-cooling} for both general and translation-invariant targets. The sampler for all experiments is ULMC, as described in Appendix~\ref{app:ulmc}.
    \item \texttt{experiment\_transformed\_gaussian.py} contains the code for the experiments in Section~\ref{subsec:well-cond-transf}, where we sample from a Gaussian that has been transformed by an ill-conditioned affine map. This file contains the code for reproducing Figures~\ref{fig:transformed-gaussian-across-kappa-b} and ~\ref{fig:transformed-gaussian-large-kappa-b}.
    \item \texttt{phi4} and the files therein contain the code for the lattice $\phi^4$ experiments in Section~\ref{subsec:lattice-phi4}. In particular, \texttt{PHI4\_EXPERIMENTS.md} contains an in-depth description of all experiments used to produce the four figures in Section~\ref{subsec:lattice-phi4} and Appendix~\ref{app:additional-phi4-facts}.
\end{itemize}

To generate the ill-conditioned transformations in Section~\ref{subsec:well-cond-transf}, we begin with a sequence of geometrically spaced eigenvalues in the interval $[1 / \kappa(B), 1]$, where $\kappa(B)$ is the target condition number. Then, the eigenvectors of $B$ are produced by a sampling a Haar random orthogonal matrix, at which point $B$ can be recovered from its eigenvalues and eigenvectors. Table~\ref{tab:transformed-gaussian-parameters} contains the default parameters used in Section~\ref{subsec:well-cond-transf}, all of which are also presented in the relevant figure subtitles.

\begin{table}[h]
\centering\small
\renewcommand{\arraystretch}{1.3}
\begin{tabular}{@{}lll@{}}
\toprule
\textbf{Parameter} & \textbf{Symbol} & \textbf{Default value} \\
\midrule
Target dimension
  & $d$ & $20$ \\
Base covariance condition number
  & $\kappa_0$ & $10$ \\
Transformation condition numbers
  & $\kappa(B)$ & $\{1,10,10^2,10^3,10^4,10^5,10^6\}$ \\
Chains per stage
  & $n$ & $256$ \\
ULMC steps per stage
  & $N$ & $96$ \\
Cooling stages
  & $K$ & $12$ \\
Cooling rate
  & $\gamma$ & $0.05$ \\
ULMC friction
  & $\widetilde{\gamma}$ & $1$ \\
ULMC step size
  & $h$ & $0.06$ \\
Repeats
  & --- & $3$ \\
\bottomrule
\end{tabular}
\caption{Default parameters for the experiments in Section~\ref{subsec:well-cond-transf}.}
\label{tab:transformed-gaussian-parameters}
\end{table}

Similarly, Table~\ref{tab:phi4-experimental-parameters} contains all of default parameters for the experiments in Section~\ref{subsec:lattice-phi4} and Appendix~\ref{app:additional-phi4-facts}. For different experiments (e.g., for small mass $m = 0.001$), note that we needed to change the default parameters to obtain adequate performance. We also varied the cooling rate across experiments. Taking $\gamma$, the cooling rate, smaller pushes the performance of Gaussian cooling towards that of translation-invariant empirical preconditioning, as expected. For Figure~\ref{fig:phi4-scaling}, however, we took $\gamma = 0.35$ to verify that Gaussian cooling is still able to reduce the condition number of the target distribution to a constant near one, even if cooling is carried out slowly.

\begin{table}[h]
\centering\small
\renewcommand{\arraystretch}{1.3}
\setlength{\tabcolsep}{5pt}
\begin{tabular}{@{}lccccc@{}}
\toprule
\textbf{Parameter}
 & \textbf{Figure~\ref{fig:phi4-large-experiments}}
 & \textbf{Figure~\ref{fig:phi4-small-mass}}
 & \textbf{Figure~\ref{fig:phi4-infinite-radius}}
 & \textbf{Figure~\ref{fig:phi4-iat}}
 & \textbf{Figure~\ref{fig:phi4-scaling}} \\
\midrule
Lattice side $d$
 & $512$ & $100$ & $100$ & $100$ & $100$ \\
Mass $m$
 & $0.01$ & $0.001$ & $0.01$ & $0.01$ & $0.01$ \\
Truncation radius $R$
 & $4$ & $4$ & $\infty$ & $2$ & $4$ \\
Chains per stage $n$
 & $512$ & $2048$ & $512$ & $512$
 & $2^6$ through $2^{10}$ \\
Steps per stage $N$
 & $256$ & $2048$ & $256$ & $256$
 & $2^5$ through $2^9$ \\
Cooling stages $K$
 & $12$ & $12$ & $12$ & $12$ & $12$ \\
Cooling rate $\gamma$
 & $0.05$ & $0.05$ & $0.05$ & $0.03$ & $0.35$ \\
ULMC step size $h$
 & $0.1$ & $0.1$ & $0.1$ & $0.03$ & $0.01$ \\
Repeats
 & $1$ & $3$ & $3$ & $3$ & $5$ \\
\midrule
Reference chains $n_{\mathrm{ref}}$
 & $512$ & $2048$ & $2048$ & $512$ & $4096$ \\
Reference steps $N_{\mathrm{ref}}$
 & $256$ & $2048$ & $2048$ & $256$ & $1634$ \\
Reference step size $h_{\mathrm{ref}}$
 & $0.001$ & $0.001$ & $0.001$ & $0.03$ & $0.003$ \\
\bottomrule
\end{tabular}
\caption{Parameters for the lattice $\phi^4$ experiments in Section~\ref{subsec:lattice-phi4}. All configurations use coupling $\beta=2$, quartic coefficient $\lambda=0.5$, and ULMC friction $\widetilde{\gamma}=1$. Learning repeats denote independent preconditioner estimates. The IAT experiment selects one of its three learned preconditioners and runs one diagnostic trajectory with $1,000$ burn-in steps and $10,000$ retained samples. Reference covariances are estimated from endpoints of independent ULMC chains using the known quadratic preconditioner, mapped back to lattice coordinates and averaged over translations. To obtain the relative condition numbers presented in Figure~\ref{fig:phi4-scaling}, the reference number of steps was determined automatically as $N_{\mathrm{ref}}=\lceil T_{\mathrm{ref}}/h_{\mathrm{ref}}\rceil$, where $T_{\mathrm{ref}} = 5$ is an empirically selected target integration time. Step sizes were empirically tuned for each of the experiments.}
\label{tab:phi4-experimental-parameters}
\end{table}

\section{Additional Algorithms}
\label{app:additional-algs}
In this appendix, we elaborate on two variants of the algorithms presented in Section~\ref{sec:algorithm} above: (i) a precise description of the algorithm used for translation-invariant distributions in Algorithm~\ref{alg:gaussian-cooling-translation-invariant} and (ii) a variant of Algorithm~\ref{alg:gaussian-cooling} that does not require a symmetric factorization of the covariance from the previous stage. In both cases, $\widehat{\Sigma}_n^{(-1)} = L^{-1} I$ so that $(\widehat{\Sigma}_n^{(-1)})^{1/2} = L^{-1/2} I$.
\begin{algorithm}
\caption{Preconditioning via Gaussian Cooling, Translation-Invariant Case}
\label{alg:gaussian-cooling-translation-invariant}
\begin{algorithmic}[1]
\REQUIRE Smoothness parameter $L > 0$, global minimizer $x^* \in \R^d$, potential $U : \R^d \to \R$ with evaluation and gradient oracle, number of Gaussian cooling stages $K$, number of samples per stage $n$, number of ULMC steps per sample $N$
\ENSURE A preconditioner $\widetilde{\Sigma}_n^{(K)}$ satisfying Equation~\ref{eq:final-precond-guarantee} with high probability
\FOR{$k = 0, \ldots, K$}
    \STATE Form $\widetilde{U}_k(x) = U_k((\widetilde{\Sigma}_n^{(k-1)})^{1/2} x)$
    \STATE $x_1^{(k)}, \ldots, x_n^{(k)} \gets \mathrm{ULMC}\left(\widetilde{U}_k, (\widetilde{\Sigma}_n^{(k-1)})^{-1/2} x^*, \tilde{\gamma}, L, h, N, n\right)$ 
    \STATE $\widehat{\sigma}^{(k)}_n \gets \frac{1}{n} \sum_{i=1}^n \sF(x_i^{(k)} - \widehat{m}_n(x_1^{(k)}, \ldots, x_n^{(k)})) \odot \overline{\sF(x_i^{(k)} - \widehat{m}_n(x_1^{(k)}, \ldots, x_n^{(k)}))}$
    \STATE $\widehat{\Sigma}^{(k)}_n \gets \sF \diag(\widehat{\sigma}^{(k)}_n) \sF^*$
    \STATE $\widetilde{\Sigma}_n^{(k)} \gets (\widetilde{\Sigma}_n^{(k-1)})^{1/2} \widehat{\Sigma}_n^{(k)}(\widetilde{\Sigma}_n^{(k-1)})^{1/2}$
\ENDFOR
\RETURN $\widetilde{\Sigma}_n^{(K)}$
\end{algorithmic}
\end{algorithm}

In Algorithm~\ref{alg:gaussian-cooling-translation-invariant}, note that both $(\widehat{\Sigma}_n^{(k)})^{1/2}$ and $(\widehat{\Sigma}_n^{(k)})^{-1/2}$ can be computed in time $\tilde{O}(d)$, as opposed to the $O(d^3)$ scaling of a full Cholesky factorization of the empirical covariance formed at each stage. Indeed, we may write
\begin{align*}
    (\widehat{\Sigma}_n^{(k)})^{1/2} = \sF \diag(\widehat{\sigma}^{(k)}_n)^{1/2} \sF^*,
\end{align*}
and
\begin{align*}
    (\widehat{\Sigma}_n^{(k)})^{-1/2} = \sF \diag(\widehat{\sigma}^{(k)}_n)^{-1/2} \sF^*
\end{align*}
with square roots and inverses taken elementwise.

\subsection{Factorization-Free Cooling}
If dense Cholesky factorizations are undesirable, one may hope to obtain a version of Algorithm~\ref{alg:gaussian-cooling} that does not require symmetric factorizations of the empirical covariance at each stage, even when the target distribution is not translation-invariant. This is indeed possible, with the following algorithm motivated by ALDI \cite{GarbunoInigo2020}. In ALDI, it suffices to take the standard non-symmetric factorization of the empirical covariance:
\begin{align*}
    C_n^{(k)} = \frac{1}{\sqrt{n}} \begin{bmatrix}
    x_1^{(k)} - \widehat{m}_n(x_1^{(k)}, \ldots, x_n^{(k)}) & \ldots & x_n^{(k)} -\widehat{m}_n(x_1^{(k)}, \ldots, x_n^{(k)})
    \end{bmatrix} \in \R^{d \times n}.
\end{align*}
In turn, we have that
\begin{align*}
    C_n^{(k)} (C_n^{(k)})^\top = \frac{1}{n} \sum_{i=1}^n (x_i^{(k)} - \widehat{m}_n(x_1^{(k)}, \ldots, x_n^{(k)}))(x_i^{(k)} - \widehat{m}_n(x_1^{(k)}, \ldots, x_n^{(k)}))^\top
\end{align*}
As is the case for the empirical covariance and mean, we can also view the non-symmetric factor above as a function of samples:
\begin{align}
\label{eq:non-symm-factor}
\widehat{C}_n(x_1, \ldots, x_n) := \frac{1}{\sqrt{n}} \begin{bmatrix}
    x_1 - \widehat{m}_n(x_1, \ldots, x_n) & \ldots & x_n -\widehat{m}_n(x_1, \ldots, x_n)
    \end{bmatrix}.
\end{align}
As it turns out, with some modifications to the underdamped Langevin chain that Algorithm~\ref{alg:gaussian-cooling} utilizes, we can replace the symmetric square roots in the Gaussian cooling algorithm with their non-symmetric counterparts, greatly reducing the computational overhead of the algorithm by removing the need for full Cholesky factorizations.

Taking a step back, suppose that $A \in \R^{d \times d}$ is a matrix that admits two factorizations: $A = BB^\top$ and $A = CC^\top$, where $B \in \R^{d \times d}$ and $C\in \R^{d \times n}$ for some $n \geq d$. The underdamped Langevin diffusion in Equation~\ref{eq:uld}, with potential $\tilde{V}(x) = V(B x)$ and initial point $B^{-1} x^*$ then takes the form
\begin{align*}
    \begin{cases}
    d X_t = P_t dt, \\
    dP_t = - B^\top \nabla V(B X_t) dt - \tilde{\gamma} P_t dt + \sqrt{2 \gamma} d B_t.
    \end{cases}
\end{align*}
Now, suppose that we define processes $Y_t = B X_t$ and $Q_t = B P_t$, which satisfy
\begin{align*}
    \begin{cases}
    d Y_t = Q_t dt, \\
    d Q_t = - A \nabla V(Y_t) dt - \tilde{\gamma} Q_t dt + \sqrt{2 \gamma} B dB_t.
    \end{cases}
\end{align*}
If $W_t$ is an $n$-dimensional Brownian motion, it follows that the increments of $B B_t$ and $C W_t$ have the same law because $BB^\top = CC^\top$. Hence, we may instead consider the diffusion process
\begin{align}
\label{eq:transformed-process}
    \begin{cases}
    d Y_t = Q_t dt, \\
    d Q_t = - A \nabla V(Y_t) dt - \tilde{\gamma} Q_t dt + \sqrt{2 \gamma} C d W_t,
    \end{cases}
\end{align}
initialized at $x^* = BB^{-1} x^*$. To simulate this via the discretization in Equation~\ref{eq:ulmc-disc}, we perform the updates
\begin{align*}
    &\widehat{Y}_{(k+1) h} = \widehat{Y}_{k h} + \frac{1 - e^{-\tilde{\gamma} h }}{\tilde{\gamma}} \widehat{Q}_{kh} - \frac{\tilde{\gamma} h - 1 + e^{-\tilde{\gamma} h}}{\tilde{\gamma}^2} A \nabla V(\widehat{Y}_{kh}) + \eta_k \\
    &\widehat{Q}_{(k+1)h} = e^{-\tilde{\gamma} h}\, \widehat{Q}_{kh} - \frac{1 - e^{-\tilde{\gamma} h}}{\tilde{\gamma}} A \nabla V(\widehat{Y}_{kh}) + \xi_k,
\end{align*}
where $(\eta_k, \xi_k) \sim \sN(0, \widehat{\Sigma}_{\widetilde{\gamma}, h})$, with
\begin{align*}
    \widehat{\Sigma}_{\widetilde{\gamma}, h} := \begin{bmatrix}
    \frac{2}{\tilde{\gamma}}[h - \frac{2}{\tilde{\gamma}}(1 - \exp(-\tilde{\gamma} h)) + \frac{1}{2 \tilde{\gamma}}(1 - \exp(-2 \tilde{\gamma} h))] & \frac{1}{\tilde{\gamma}}(1 - 2 \exp(-\tilde{\gamma} h) + \exp(-2\tilde{\gamma} h)) \\
    \frac{1}{\tilde{\gamma}}(1 - 2 \exp(-\tilde{\gamma} h) + \exp(-2\tilde{\gamma} h)) & 1 - \exp(-2\tilde{\gamma} h)
    \end{bmatrix} \otimes A.
\end{align*}
To sample the noise $(\eta_k, \xi_k)$ in this formulation, factorize $\widehat{\Sigma}_{\tilde{\gamma}, h} = \widehat{R}_{\tilde{\gamma}, h} \widehat{R}_{\tilde{\gamma}, h}^\top \otimes A$. Then, if $g \sim \sN(0, I_{2n})$, we observe that
\begin{align*}
    \covar((\widehat{R}_{\tilde{\gamma}, h} \otimes C) g) = \widehat{R}_{\tilde{\gamma}, h} \widehat{R}_{\tilde{\gamma}, h}^\top \otimes CC^\top = \widehat{R}_{\tilde{\gamma}, h} \widehat{R}_{\tilde{\gamma}, h}^\top \otimes A.
\end{align*}
Hence, the noise can be sampled without ever forming a symmetric factorization of $A$. In summary, we can sample from the ULMC discretization of the process in Equation~\ref{eq:transformed-process} using the following algorithm:

\begin{algorithm}
\caption{Factorization Free Underdamped Langevin Monte Carlo}
\label{alg:ulmc-fact-free}
\begin{algorithmic}[1]
\REQUIRE Potential $V : \mathbb{R}^d \to \mathbb{R}$ with gradient oracle $\nabla V$, non-symmetric factorization $C \in \R^{d \times n}$ with $A = CC^\top$ global minimizer $x^*$, friction parameter $\tilde{\gamma} > 0$, smoothness parameter $L > 0$, step size $h > 0$, Cholesky factor $\widehat{R}_{\tilde{\gamma}, h} \in \R^{2 \times 2}$, number of iterations $N$ 
\ENSURE Sample $\widehat{Y}_{Nh}$
\STATE Draw $g_Y, g_Q \overset{i.i.d.}{\sim} \sN(0, I_n)$
\STATE $(\widehat{Y}_0, \widehat{Q}_0) \gets (x^* + L^{-1/2} C g_Y, C g_Q)$
\FOR{$m = 0, 1, \ldots, N-1$}
    \STATE $\widehat{V}_{mh} \gets \nabla V(\widehat{Y}_{mh})$
    \STATE Sample $g \sim \mathcal{N}(0, I_{2n})$
    \STATE  $(\eta_m, \xi_m) \gets (\widehat{R}_{\tilde{\gamma}, h} \otimes C) g$
    \STATE $\widehat{Y}_{(m+1)h} \gets \widehat{Y}_{mh} + \dfrac{1 - e^{-\tilde{\gamma} h}}{\tilde{\gamma}}\, \widehat{Q}_{mh} - \dfrac{\tilde{\gamma} h - 1 + e^{-\tilde{\gamma} h}}{\tilde{\gamma}^2} C C^\top \widehat{V}_{mh} + \eta_m$
    \STATE $\widehat{Q}_{(m+1)h} \gets e^{-\tilde{\gamma} h}\, \widehat{Q}_{mh} - \dfrac{1 - e^{-\tilde{\gamma} h}}{\tilde{\gamma}} C C^\top \widehat{V}_{mh} + \xi_m$
\ENDFOR
\RETURN $\widehat{Y}_{N h}$
\end{algorithmic}
\end{algorithm}
By construction, $\widehat{Y}_{Nh} \overset{d}{=} B \widehat{X}_{N h}$ so that if $y_1^{(k)}, \ldots, y_n^{(k)}$ are sampled from Algorithm~\ref{alg:ulmc-fact-free}, then $x_{i}^{(k)} := B^{-1} y_i^{(k)}$ are samples from Algorithm~\ref{alg:ulmc}. In turn,
\begin{align}
\label{eq:fact-free-emp-covariance}
    \widehat{S}_n(y_1^{(k)}, \ldots, y_n^{(k)}) = \widehat{S}_n(B x_1^{(k)}, \ldots, B x_n^{(k)}) = B \widehat{S}_n(x_1^{(k)}, \ldots, x_n^{(k)}) B^\top,
\end{align}
which is the desired output from Algorithm~\ref{alg:gaussian-cooling} at step $K$.
Importantly, we can entirely sidestep any symmetric factorizations using this approach. Below, we introduce the shorthand
\begin{align*}
    y_1, \ldots, y_n = \mathrm{ULMC_{FF}}(V, C, x^*, \tilde{\gamma}, L, h, N, n)
\end{align*}
to call Algorithm~\ref{alg:ulmc-fact-free} with the specified parameters. With this notation, we introduce Algorithm~\ref{alg:gaussian-cooling-fact-free}, which carries out Gaussian cooling without any symmetric factorizations of the empirical covariance.

\begin{algorithm}
\caption{Factorization-Free Preconditioning via Gaussian Cooling}
\label{alg:gaussian-cooling-fact-free}
\begin{algorithmic}[1]
\REQUIRE Preconditioning tolerance $\delta > 0$, smoothness parameter $L > 0$, global minimizer $x^* \in \R^d$, potential $U : \R^d \to \R$ with evaluation and gradient oracle, number of Gaussian cooling stages $K$, number of samples per stage $n$, number of ULMC steps per sample $N$
\ENSURE A preconditioner $\widehat{\Sigma}_n^{(K)}$ satisfying Equation~\ref{eq:final-precond-guarantee} with high probability
\FOR{$k = 0, \ldots, K$}
    \STATE $y_1^{(k)}, \ldots, y_n^{(k)} \gets \mathrm{ULMC_{FF}}\left(U_k, C_n^{(k-1)}, x^*, \tilde{\gamma}, L, h, N, n\right)$ 
    \STATE $C_n^{(k)} \gets \widehat{C}_n(y_1^{(k)}, \ldots, y_n^{(k)})$
\ENDFOR
\STATE $\widetilde{\Sigma}_n^{(K)} \gets C_n^{(K)} (C_n^{(K)})^\top$
\RETURN $\widetilde{\Sigma}_n^{(K)}$
\end{algorithmic}
\end{algorithm}

In the above algorithm, we take
\begin{align*}
    C_n^{(-1)} := L^{-1/2} \begin{bmatrix}
    e_1 & e_2 & \cdots & e_d & \mathbf{0}_d & \cdots & \mathbf{0}_d
    \end{bmatrix} \in \R^{d \times n},
\end{align*}
where $e_i \in \R^d$ is the $i$th standard basis vector and $\mathbf{0}_d \in \R^d$ is the vector of all zeros. Indeed, it then follows that $C_n^{(-1)} (C_n^{(-1)})^\top = L^{-1} I$, giving the desired non-symmetric factorization of $\widetilde{\Sigma}_n^{(-1)}$ from Algorithm~\ref{alg:gaussian-cooling}. 

\section{A Cautionary (Counter)example}
\label{app:cautionary-tale}
Although many distributions of interest will satisfy Equation~\ref{eq:true-conditioning}, it is well-known in the HMC community that \textit{nonlinear} correlations can stymie even the most effective mass matrix estimation procedures \cite{Carpenter2017}. This should come as no surprise, as Equation~\ref{eq:true-conditioning} imposes structure on the global conditioning of the potential, but it cannot capture local behavior of the Hessian. In turn, for any choice of $L > \mu > 0$, it is possible to construct potentials that are $\mu$-strongly log-concave and $L$-smooth but, in order for Equation~\ref{eq:true-conditioning} to hold for some $B \succ 0$, one must choose $\alpha = O(1 / \kappa)$. Unfortunately, for such potentials, Gaussian cooling appears to provide little practical benefit due to its query complexity's polynomial scaling in $\alpha$, as presented in Theorem~\ref{th:gc-query-complexity-general}.

For the sake of concreteness, we construct a simple example of such a potential here, for an arbitrary dimension $d \in \N$. For fixed $L > \mu > 0$, define
\begin{align}
\label{eq:bad-potential}
U(x) := \sum_{i=1}^d \left(\frac{L + \mu}{4} x_i^2 + \frac{L - \mu}{8} \cos(2 x_i)\right).
\end{align}
Because this potential is separable, its Hessian is diagonal, with
\begin{align*}
    \nabla^2 U(x)_{ii} = \frac{L + \mu}{2} - \frac{L - \mu}{2} \cos(2 x_i).
\end{align*}
In particular, we have that
\begin{align*}
    \mu \leq \nabla^2 U(x)_{ii} \leq L, \quad i \in [d]
\end{align*}
so that $U$ does indeed induce a $\mu$-strongly log-concave and $L$-smooth measure with condition number $\kappa = \frac{L}{\mu}$.

Define $x^{(1)} = \frac{\pi}{2} e_1$ and $x^{(2)} = \frac{\pi}{2} e_2$, where $e_i \in \R^d$ denotes the $i$th standard basis vector. By construction, $\nabla^2 U(x^{(1)})_{11} = L$ and $\nabla^2 U(x^{(1)})_{22} = \mu$, while $\nabla^2 U(x^{(2)})_{11} = \mu$ and $\nabla^2 U(x^{(2)})_{22} = L$. If
\begin{align*}
    \alpha B \preceq \nabla^2 U(x) \preceq B
\end{align*}
is to hold for some $B \succ 0$ and $\alpha \in (0, 1)$ simultaneously for all $x \in \R^d$, then we would require that $\alpha B_{11} \leq \nabla^2 U(x)_{11} \leq B_{11}$ for all $x \in \R^d$. This condition implies that $L \leq B_{11}$ so that $\alpha L \leq \alpha B_{11} \leq \mu$. In other words, $\alpha \leq \frac{\mu}{L} = \frac{1}{\kappa}$.

As the $\phi^4$ example in Section~\ref{sec:num-exps} indicates, the presence of a nonlinear term in the Hessian of the underlying potential does \textit{not} necessarily mean that Equation~\ref{eq:true-conditioning} cannot hold for some $\alpha = \Theta(1)$.
\end{document}